%% file: 1_main_v3.tex
\documentclass[10pt]{article}
\usepackage[a4paper,margin=1.7cm,columnsep=0.7cm]{geometry}

\input{0_preamble.tex}

\title{
\textbf{Quantum communication and Bell nonlocality require infinite classical communication to simulate
}}

\author[1]{Carlos de Gois\textsuperscript{*}}
\author[2]{Thyago S. R. Santos\textsuperscript{\dag}}
\author[2,3]{Carlos Vieira\textsuperscript{\ddag}}

\affil[1]{CPHT, LIX, CNRS, Inria, \'{E}cole polytechnique, Institut Polytechnique de Paris, Palaiseau, France}
\affil[2]{Instituto de Matemática, Estatística e Computação Científica, Universidade Estadual de Campinas, 13083-859, Campinas, São Paulo, Brazil}
\affil[3]{Sorbonne Université, CNRS, LIP6, F-75005 Paris, France}

\date{}

\begin{document}
\etocdepthtag.toc{main}

\twocolumn[
\begin{@twocolumnfalse}

\maketitle

\vspace{-3em}
\begin{center}
\begin{minipage}{0.94\textwidth}
\begin{abstract}
    A quantum system of any fixed dimension can be prepared in a continuum of states, yet it cannot be used to transmit an unlimited amount of classical information.
    Similarly, the correlations observed between measurement outcomes on separate parts of a shared quantum system can be stronger than classical correlations, but they cannot transmit information.
    These fundamental limitations suggest that the statistics observed from quantum communication and quantum correlations may admit a simulation using a finite amount of classical communication.
    This expectation is confirmed in the smallest nontrivial quantum dimension, with two classical bits being necessary and sufficient to exactly simulate qubit communication and all correlations between qubits.
    Despite significant efforts during the previous decades, this remained the only solved case.
    Here we resolve both problems for every quantum dimension.
    The solution reveals an unexpected qualitative transition starting at dimension four: no finite amount of classical communication can exactly simulate ququart communication nor all quantum correlations of two entangled ququarts, even with unlimited shared randomness.
    One might have expected this transition, if it existed, to appear already for qutrits.
    Instead, we construct an explicit protocol that exactly simulates qutrit communication using $357$ classical bits, and consequently, all correlations of two entangled qutrits.
\end{abstract}
\end{minipage}
\end{center}

\vspace{1.5em}
\end{@twocolumnfalse}
]

\begingroup
\renewcommand{\thefootnote}{\fnsymbol{footnote}}
\footnotetext[1]{%
  \href{mailto:carlos.belini-de-gois@inria.fr}
       {\texttt{carlos.belini-de-gois@inria.fr}}%
}
\footnotetext[2]{%
  \href{mailto:thyagosr@unicamp.br}
       {\texttt{thyagosr@unicamp.br}}%
}
\footnotetext[3]{%
  \href{mailto:carlosvieira@ime.unicamp.br}
       {\texttt{carlosvieira@ime.unicamp.br}}%
}
\endgroup
\setcounter{footnote}{0}
\renewcommand{\thefootnote}{\arabic{footnote}}

\section*{Introduction}

Communication always relies on a physical carrier of information.
What can ultimately be communicated therefore depends on the physical laws obeyed by that carrier.
For most of the history of communication, information carriers were adequately described by classical physics.
The emergence of quantum information raised a new possibility, that carriers governed by quantum mechanics might enable forms of communication with no classical counterpart.
There is an immediate reason to suspect that they do.
Whereas a $d$-dimensional classical system can encode one of $d$ possible messages, a quantum system of the same dimension has a continuum of possible pure states.
At first sight, this continuum might seem to allow an unlimited amount of classical information to be transmitted.
A foundational result of quantum information theory, known as Holevo's bound, rules out this possibility, establishing that at most $\log_2 d$ bits of classical information can be retrieved from a $d$-dimensional quantum system \cite{holevo1973bounds}.
This limitation poses a fundamental question: if only a finite amount of classical information can be read out of quantum states, can their statistics be simulated with a finite amount of classical communication (\cref{fig:pam-scenario}(a))?
A result by Frenkel and Weiner provides direct support for this possibility, by showing that for any fixed measurement the behavior of all $d$-dimensional quantum states can be simulated exactly using a classical $d$-dimensional system and shared randomness \cite{FrenkelWeiner2015}.
This, however, does not settle the problem, as a simulation of quantum communication must not only reproduce a single measurement, but rather \emph{all} possible measurements.

\begin{figure*}[t]
    \centering
    \includegraphics[width=.7\textwidth]{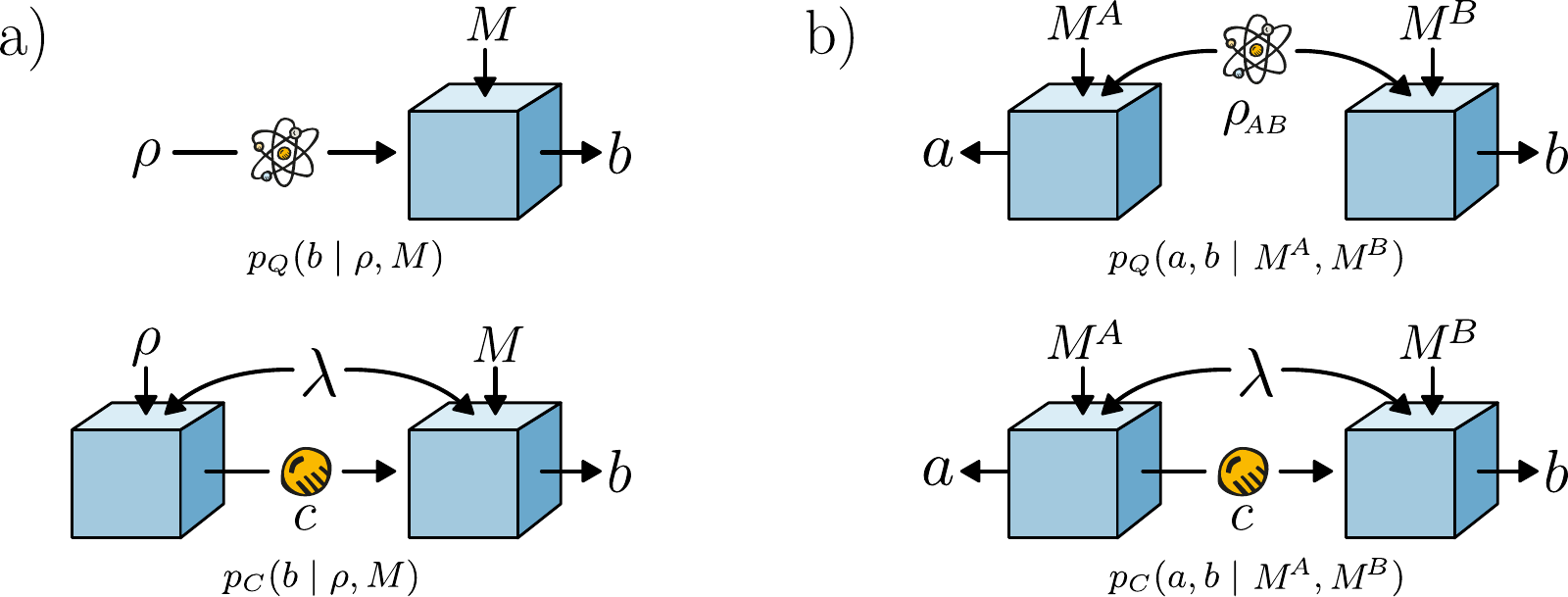}
    \caption{
    \textbf{Classical simulation of quantum communication and Bell nonlocality.}
        (a) In the communication setting, a quantum state $\rho$ of dimension $d_Q$ is communicated to a measurement device. Given a measurement $M$, an outcome $b$ is produced with probability $p_Q(b \mid \rho, M)$.
        In a classical simulation, the quantum system is encoded in a message $c$ of finite dimension using unlimited shared randomness, and decoded producing $b$ with probability $p_C(b \mid \rho, M)$. A classical simulation exists if there is an encoder and decoder functions realizing $p_C(b \mid \rho, M) = p_Q(b \mid \rho, M)$ for all possible inputs $\rho$ and $M$.
        (b) In the Bell nonlocality setting, a quantum state $\rho_{AB}$ is shared between two experimenters, who locally measure their systems with $M^A$ and $M^B$, producing the outcomes $a, b$ with probability $p_Q(a, b \mid M^A, M^B)$. A classical simulation replaces the shared quantum state with unlimited shared randomness and allows a message $c$ of finite dimension to be communicated between the devices. The simulation exists if $p_C(a, b \mid M^A, M^B) = p_Q(a, b \mid M^A, M^B)$ for all possible $M^A$ and $M^B$.
}
    \label{fig:pam-scenario}
\end{figure*}

A closely related question has been studied since the early 1990s in the context of Bell nonlocality \cite{maudlin1992bell,Steiner2000,BrassardCost1999}.
Bell's theorem \cite{Bell:1964PHY}, whose experimental proofs were recognized by the 2022 Nobel Prize in Physics \cite{Freedman:1972PRL,Aspect:1982PRL,Weihs:1998PRL}, establishes that two experimenters locally measuring a shared quantum state without communicating can produce correlations beyond any classical model based on shared randomness and local responses \cite{Bell:1964PHY,Clauser:1969PRL}.
Although communication is forbidden in a Bell experiment, allowing it in a classical simulation is a natural way to quantify their nonlocality.
For any fixed Bell scenario with finitely many settings and outcomes, finite classical communication always suffices \cite{Pironio2003BellCommunicationLowerBounds}.
The question is whether the required communication remains bounded over \emph{all} quantum correlations achievable at a fixed quantum dimension (\cref{fig:pam-scenario}(b)).
Since quantum correlations alone cannot transmit information, one might expect a small amount of classical communication to always suffice.

A notable result by Toner and Bacon showed that Bell correlations and quantum communication can be simulated with only two classical bits in the special case of qubits and projective measurements \cite{TonerBacon}.
Later work extended this result to arbitrary quantum measurements, thereby completely resolving the qubit case \cite{RennerTavakoliQuintino}.
Despite significant efforts in the following decades \cite{MontinaCommunication2011,HavlicekSimpleCommunication2020,RennerTavakoliQuintino,SchlosserKleinmann2026,Zartab2026}, no upper bounds have been obtained for beyond qubits.
A central difficulty is that the known qubit constructions rely heavily on the simple geometry of the Bloch sphere, which has no higher-dimensional analogue \cite{Kimura2003,Eltschka2021shapeofhigher}.
A partial result, holding for any finite dimension but only in the special case of Bell correlations with uniform marginals produced by binary projective measurements, states that two bits suffice \cite{RegevToner}.
Complementarily, known lower bounds show that the required classical dimension can grow at least exponentially with the quantum dimension \cite{BuhrmanQuantumClassical1998,BrassardCost1999,MontinaCommunication2011,HavlicekSimpleCommunication2020}, and that classical dimensions at least five and ten are necessary to simulate qutrit and ququart communication, respectively \cite{HavlicekSimpleCommunication2020}.

Here we resolve, for every finite quantum dimension, whether quantum communication and Bell nonlocality can be exactly reproduced with a finite amount of classical communication.
The solution reveals a surprising qualitative transition: for ququarts and above, no finite amount of classical communication can do so, even when the sender and receiver share unlimited randomness in advance.
Whether the required communication could become infinite was previously unclear \cite{RegevToner,RennerTavakoliQuintino,SchlosserKleinmann2026}, and if such a transition existed, qutrits would have been a natural place to expect it, since they are the first systems beyond the simple geometry of qubits.
Instead, we show that the qutrit cases can still be simulated exactly, by constructing an explicit protocol that uses $357$ classical bits, a bound that we do not expect to be optimal.

At the level of observable statistics, qubits and qutrits can therefore be regarded as more efficient realizations of finite classical systems. Starting at dimension four, this interpretation fails.
No finite classical alphabet, however large, can reproduce all the statistics obtained from preparing and measuring a ququart.
Similarly, there exist entangled pairs of ququarts whose correlations across all local measurements cannot be reproduced using any finite amount of classical communication.

\section*{Classical simulation of quantum communication}

We first consider the quantum communication problem and return to Bell correlations in a later section.
A quantum communication protocol involves two devices.
The first, commonly called Alice, prepares a $d_Q$-dimensional quantum system $\rho$ and sends it to the second, Bob.
To read information out of $\rho$, Bob performs a measurement $M$ and obtains a classical outcome.
See \cref{fig:pam-scenario}(a) for an illustration.
The cases $d_Q=2$, $3$, and $4$ correspond respectively to a qubit, a qutrit and a ququart.

Mathematically, quantum theory associates a finite-dimensional quantum state with a positive semidefinite matrix in $\rho \in \C^{d \times d}$ with $\tr(\rho) = 1$.
A measurement $M$ with $B$ possible outcomes is described by a collection of matrices, $M = \{ M_b \}_{b = 0}^{B-1}$, called measurement effects.
The subscript $b$ in each effect is a label identifying a possible measurement outcome.
Each of these effects is a positive semidefinite matrix in $\C^{d \times d}$, and they sum to the identity operator.
Quantum theory then specifies that upon measuring $M$ on $\rho$, outcome $b$ happens with probability
\begin{equation}
    p_Q(b \mid \rho, M) = \tr(\rho M_b) .
\label{eq:born-rule-main}
\end{equation}
Varying both the state $\rho$ prepared by Alice and the measurement $M$ chosen by Bob gives the quantum behavior $p_Q(b\mid\rho,M)$.

We ask whether the same behavior can be obtained using only classical communication and shared randomness.
This means that, instead of a quantum system, the preparation device sends a classical message $c$ chosen out of $d_C$ possibilities, and we also let the two devices coordinate their actions classically before the protocol starts.
This coordination is denoted by a variable $\lambda$, called shared randomness, which is accessible to both devices.
While $\lambda$ may contain an unlimited amount of classical information, it is distributed before the protocol, thus it cannot contain a description of $\rho$ and $M$, which are only determined during the protocol.
This is an abstract definition for the most general classical communication protocols that can be performed with a single round of communication from the preparation to the measurement device.
For any finite $d_C$, its possible behaviors are described by
\begin{equation}
\begin{split}
    p_C(b \mid \rho, M) = & \int_\Lambda \sum_{c =1}^{d_C} \,p_A(c \mid \rho, \lambda) \\
    &\quad \times p_B(b \mid c, M, \lambda) \, \dd \mu(\lambda)
\label{eq:classical-behaviors}
\end{split}
\end{equation}
where $\mu(\lambda)$ is a probability measure over the space $\Lambda$ of possible shared variables.
In this model, the preparation device uses the encoder distribution $p_A$ to decide which message $c$ to communicate given $\lambda$ and its knowledge of $\rho$.
Upon receiving $c$, the measurement device uses the decoding distribution $p_B$ to decide its output $b$ based on its knowledge of $c$, the measurement $M$ and the shared $\lambda$; see \cref{fig:pam-scenario}(a).
The resulting probability of outcome $b$ is $p_C(b\mid\rho,M)$ as defined above; see \cref{sec:classical-simulation-model} for a precise definition.
This is often called the classical prepare-and-measure model, and has been extensively studied during the past decades \cite{ambainis2008quantum,gallego2010pam,FrenkelWeiner2015,VicenteSharedRandomnness2017,cgois_classicality_2021,Doolittle_2021_Certifying,brask2026quantumcorrelationsprepareandmeasurescenarios}.

Quantum communication can be simulated by $d_C$-dimensional classical communication if there exists an encoder $p_A$ and a decoder $p_B$ such that,
\begin{equation}
    p_Q(b \mid \rho, M) = p_C(b \mid \rho, M) .
    \label{eq:classical-simulation-main}
\end{equation}
for every possible quantum state $\rho$ and measurement $M$.
The classical simulation cost is then defined as
\begin{equation}
	C_{\mathrm{cl}}(d_Q) = \inf\{d_C\in\N:\text{\cref{eq:classical-simulation-main} holds for }d_Q\},
\end{equation}
with $C_{\mathrm{cl}}(d_Q) := \infty$ if the set is empty.
It is known that $C_{\mathrm{cl}}(2) = 4$, $C_{\mathrm{cl}}(3) \geq 5$, $C_{\mathrm{cl}}(4) \geq 10$ and $C_{\mathrm{cl}}(2^n) \geq 2^{c \, 2^n}$, where $c$ is a constant \cite{BuhrmanQuantumClassical1998,BrassardCost1999,TonerBacon,MontinaCommunication2011,HavlicekSimpleCommunication2020,RennerTavakoliQuintino,SchlosserKleinmann2026}.
It remained unknown whether $C_{\mathrm{cl}}(d_Q)$ is finite for $d_Q \ge 3$.

\section*{No finite simulation of ququarts}
\label{sec:ququarts-main}

We now prove our central result: for any fixed quantum dimension $d_Q \geq 4$, the quantum behaviors $p_Q$ cannot be reproduced with any finite amount of classical communication.
Here we present the main ideas behind the argument.
Complete proofs are provided in \cref{app:d4}.

To prove the claim, it is enough to identify a restricted family of ququart states and measurements that already admits no finite classical simulation.
Any protocol capable of simulating ququart communication would necessarily simulate this family as well.
Moreover, once the impossibility is established for $d_Q = 4$, the result immediately extends to every $d_Q > 4$ by embedding the ququart states and measurements into larger-dimensional systems.

We therefore restrict our attention to the particularly simple setting of pure ququart states and binary projective measurements having one rank-$1$ effect. 
A pure state in $d_Q = 4$ is a normalized vector $\ket{x} \in \C^4$, while a binary measurement with a rank-$1$ effect can be described by a single vector $\ket{y} \in \C^4$ by writing $M_y = \{ M_{0|y}, M_{1|y} \} = \{ \dyad{y}, \eye_4 - \dyad{y} \}$, where the second effect is determined from normalization.
Since multiplying $\ket{x}$ or $\ket{y}$ by a global phase does not change the state or measurement, we identify them with rays $x$ and $y$ in the complex projective space $\CP^3$ \cite{Bengtsson_Zyczkowski_2006,Cunha2005Emaranhamento}. For simplicity, we write $p_Q(b | x,y) := p_Q(b | \, \ket{x},M_y).$

When $x$ is measured with $M_y$, the probability of obtaining the outcome $b = 0$, $p_Q(0|x,y)$, is given by \cref{eq:born-rule-main}, which in this case simplifies to the fidelity function $f(x, y) :=  \xy ^2$. A valid classical simulation must reproduce $f(x, y)$ for all $x$ and $y$.

\begin{figure}[t]
    \centering
    \includegraphics[width=.87\columnwidth]{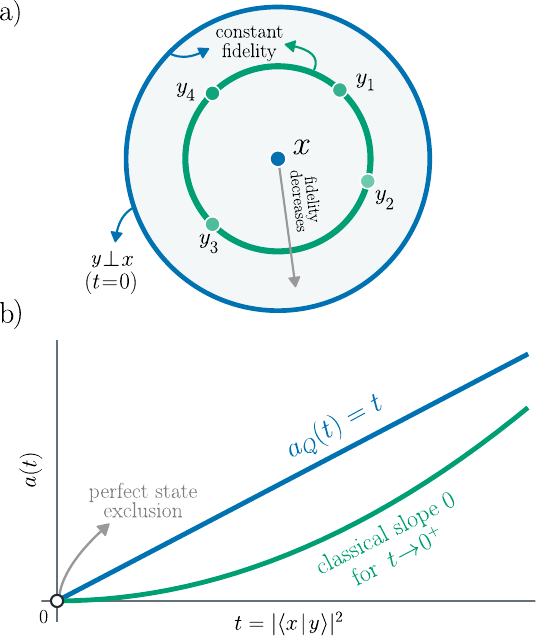}
    \caption{\textbf{The central argument in the ququart proof.}
        (a) Fix a preparation $x\in\CP^3$ and group the measurement effects $y$ according to their fidelity $t=|\langle x\vert y\rangle|^2$ with the state $x$. Each ring represents all tests with the same fidelity. We denote by $a(t)$ the average of $p(0\mid x,y)$ over all pairs $(x,y)$ of fidelity $f(x,y) = t$.
        The proof studies $a(t)$ as $t\to0^+$.
        (b) For every pair $(x,y)$ with fixed fidelity $f(x,y)=t$, quantum theory predicts $p_Q(0\mid x,y)=t$. Consequently, averaging over all such pairs gives $a_Q(t)=t$, which has slope $1$ at $t=0$. Any classical simulation using finitely many messages would instead yield an averaged curve $a_C(t)$ with slope $0$ as $t\to0^+$. (The green curve is illustrative; only its derivative near $t=0$ is relevant to the proof.)
    }
    \label{fig:d4-fidelity-and-slope}
\end{figure}
 
There is an important symmetry in $f$ that will play a central role in the proof: its value does not change if the state $x$ and the measurement $y$ are rotated in the same way, so for every unitary $U$, $f(U x, Uy) = f(x, y)$. 
During the proof we will exploit this symmetry by averaging $p(0|x,y)$ over all pairs of state and measurement having the same fidelity.
To do that, let us first define $\eta_t$ as the uniform probability measure over all pairs $(x,y)\in \CP^3 \times \CP^3$ such that $f(x,y)=t$ for any $t\in[0,1]$.
Such a pair can be generated by starting from any fixed pair of fidelity $t$ and applying the same uniformly random unitary rotation to both elements.
Equivalently, if $\sigma$ denotes the uniform probability measure on $\CP^3$, we can first sample $x \sim \sigma$ and then sample $y$ uniformly from the constant fidelity ring around $x$; see \cref{fig:d4-fidelity-and-slope}(a) and \cref{sec:fixed-fidelity-pairs}. 

Computing the average of $p_Q(0|x,y)$ over $\eta_t$, we get that
\begin{equation}
\begin{split}
a_Q(t) :&= \int_{\CP^3 \times \CP^3} p_Q(0|x,y) \dd\eta_t(x,y) \\
&= \int_{\CP^3 \times \CP^3} f(x,y) \dd\eta_t(x,y) = t .
\end{split}
\end{equation}
Analogously, the average of a classical simulation can be obtained from 
\begin{equation}
\begin{split}
    a_C(t) :&= \int_{\CP^3 \times \CP^3} p_C(0|x,y) \dd\eta_t(x,y) \\
    &=
    \int_{\CP^3\times\CP^3} \int_\Lambda
    \sum_{c=1}^{d_C} p_A(c\mid x,\lambda) \\
    &\qquad \times p_B(0\mid y,c,\lambda)
    \,\dd\mu(\lambda) \,\dd\eta_t(x,y) .
\end{split}
\label{eq:classical-simulation-average-main}
\end{equation}

If a classical simulation in the sense of \cref{eq:classical-simulation-main} exists, we must have $a_C(t) = a_Q(t) = t$ for every $0 \leq t \leq 1$.
The core argument in the proof comes from comparing the behaviors of the functions $a_C(t)$ and $a_Q(t)$.
For this comparison, we will focus on the region near $t = 0$.
This is motivated by previous observations that perfect state exclusion at $t = 0$ imposes strong constraints on the possible classical simulation functions \cite{MontinaCommunication2011,SchlosserKleinmann2026}.
There, the quantum average starts from $a_Q(0) = 0$ and increases with slope $a_Q'(0)=1$.
This would require that $a_C(0) = 0$ and $a_C'(0) = 1$ (all derivatives are to be understood with $t \to 0^+$).

The crucial step in the proof is to show that for any finite $d_C$, forcing $a_C(0) = 0$ implies instead that $a_C'(0) = 0$; see \cref{fig:d4-fidelity-and-slope}(b).
There is an important technical difficulty in proving this: the encoding and decoding response functions $p_A$ and $p_B$ are not required to be differentiable, therefore $a_C'(0)$ cannot be directly obtained.
We overcome this difficulty with two observations.
First, although $p_B$ may be highly irregular, averaging it over all measurement rays that have zero fidelity with a given $x$ produces a function of $x$ whose second-order spatial variation is well defined.
It then turns out that this spatial variation is related to the fidelity variation near $t = 0$, yielding $a_C'(0)=0$.
Details on zero slope argument and a proof outline are presented in the Methods section.

\section*{Finite simulation of qutrits}
\label{sec:finite-qutrit-main}

Having ruled out finite classical simulation of ququart communication, we now show that the transition occurs precisely at $d_Q = 4$ by constructing a finite simulation for qutrits.
In the most general case (that is, for any POVM), the protocol will use classical messages of dimension $(2\sqrt{24})^4(112 + (2\sqrt{24})^4)^{26}$, thus proving that $C_{\mathrm{cl}}(3) < 2^{357}$.

It is instructive to start with an outline of the simulation protocol (detailed explanations and proofs are presented in \cref{app:d3}).
Consider for the moment the task of simulating the statistics generated by pure states $\ket{x} \in \C^3$ and binary projective measurements $M_y = \{ \dyad{y}, \eye_3 - \dyad{y} \}$.
The protocol constructed for this case can later be extended to all states and measurements.
We once more use the equivalent representation $x, y \in \CP^2$.
Since the measurements are binary, it is sufficient to reproduce $p_Q(0 \mid x, y) = \xy^2$, with the probability of outcome $b = 1$ following from normalization.

The key idea behind the protocol is already visible in a construction that, even without communication, reproduces the desired quantum statistics up to a positive constant. Alice and Bob share a random variable, generated independently of their inputs, consisting of a ray $h$ and two numbers, $\alpha$ and $\beta$:
\begin{equation}
\theta = (h, \alpha, \beta) \in \CP^2 \times \mathcal R,
\end{equation}
where
\begin{equation}
	\mathcal R
	=
	\left\{ (\alpha,\beta)\in[0,1]^2 : \alpha\geq\beta, \, \alpha + \beta\geq1 \right\} .
	\label{eq:d3-main-threshold-region}
\end{equation}
The ray $h$ is sampled uniformly on $\CP^2$, according to the normalized Haar measure.
Independently, the thresholds $\alpha$ and $\beta$ are sampled as follows.
With probability $1/3$, we take
\begin{equation}
	\alpha_k\sim\operatorname{Unif}\left[\frac12, \, 1\right],
	\quad
	\beta_k=1-\alpha_k ,
\label{eq:nu-line-main}
\end{equation}
and with probability $2/3$, we sample $(\alpha_k, \beta_k)$ uniformly from $\mathcal R$; see \cref{fig:average-communication-protocol}(a).
The resulting probability distribution of $\theta$ will be denoted by $\nu$.

\begin{figure}[t]
	\centering
		\includegraphics[width=0.9\linewidth]{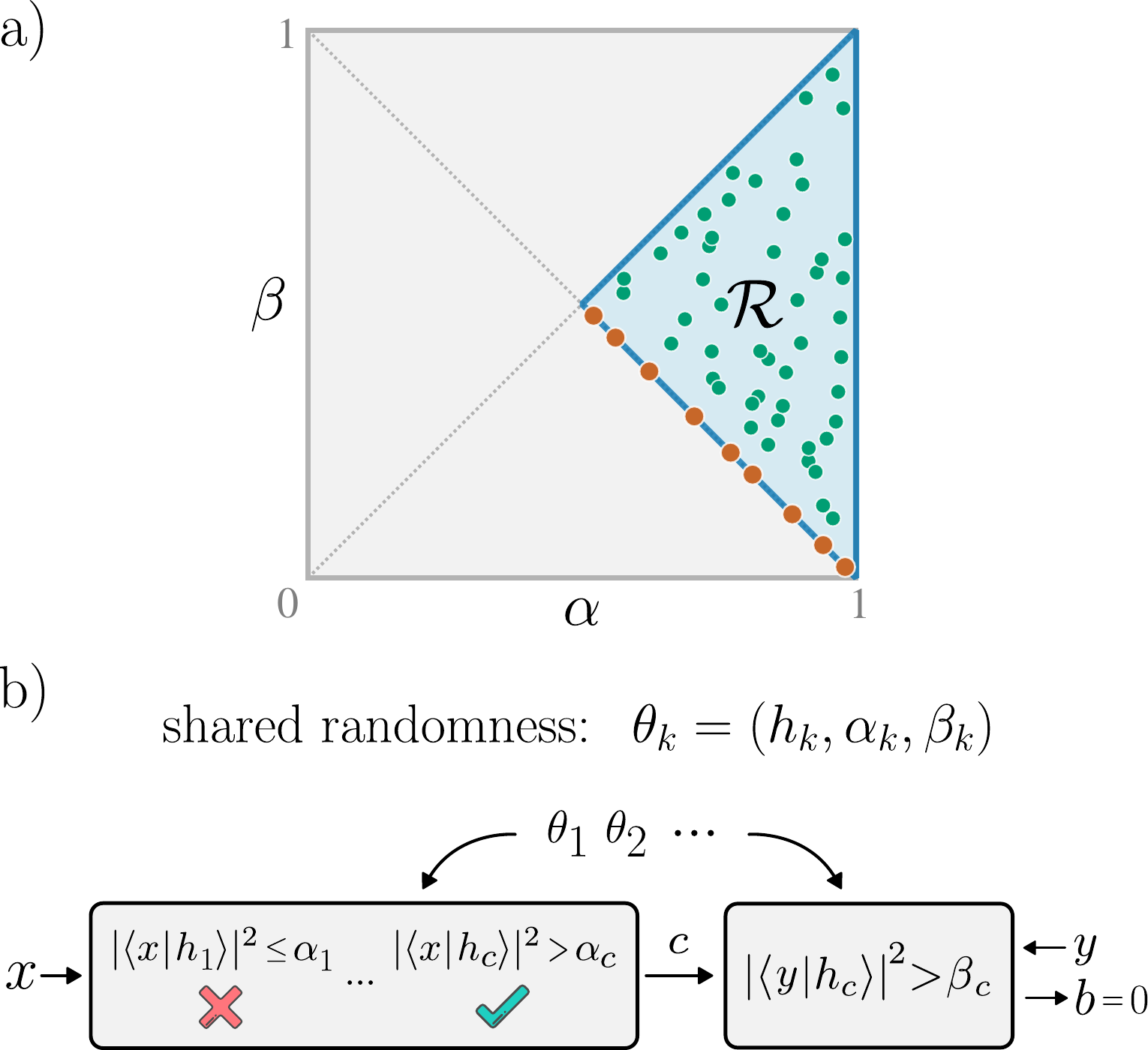}
	\caption{
		\textbf{The protocol with finite average communication.}
        (a) Region $\mathcal R$ from which the thresholds $\alpha$ and $\beta$ are sampled. Under the probability law $\nu$, with probability $1/3$ the pair $(\alpha, \beta)$ is sampled uniformly along the line $\beta=1-\alpha$ (orange circles), and with probability $2/3$ it is sampled uniformly from the blue region inside the triangle (green circles).
		(b) Alice and Bob share an infinite sequence of independent samples $\theta_1,\theta_2,\ldots \sim \nu$.
        Alice computes $a_k=\abs{\langle x \vert h_k\rangle}^2$ for each $\theta_k$ in order and communicates the first index $c$ satisfying $a_c > \alpha_c$.
        Bob outputs $0$ when $\abs{\langle y \vert h_c\rangle}^2>\beta_c$, and outputs $1$ otherwise.
        In this protocol, the probability of outcome $0$ is $\abs{\langle x \vert y\rangle}^2$ and the communicated index $c$ is finite almost surely, but it is not bounded.
	}	
	\label{fig:average-communication-protocol}
\end{figure}

After receiving $x$, Alice computes $\abs{\langle x \vert h\rangle}^2$ and declares her test successful when $\abs{\langle x \vert h\rangle}^2 > \alpha$.
Similarly, after receiving $y$, Bob computes $\abs{\langle y \vert h\rangle}^2$ and declares his test successful when $\abs{\langle y \vert h\rangle}^2 > \beta$.
The probability that both local tests succeed is thus
\begin{equation}
\begin{split}
&\Pr_\nu \{
|\langle x \vert h\rangle|^2>\alpha,\ 
|\langle y \vert h\rangle|^2>\beta
\} \\
&\qquad=
\int_{\CP^2\times\mathcal R}
\1_{\{|\langle x \vert h\rangle|^2>\alpha\}}
\1_{\{|\langle y \vert h\rangle|^2>\beta\}}
\,\dd\nu .
\end{split}
\label{eq:acceptance-integral-no-comm-main}
\end{equation}
By \cref{lem:qutrit-joint-acceptance}, this integral equals $\abs{ \langle x \vert y \rangle }^2 / 18$.

To determine Alice's success probability, set $y = x$ above.
Then $b = a$, and since $\alpha \geq \beta$ almost surely under $\nu$, whenever her test succeed, Bob's also succeeds.
Hence,
\begin{equation}
\Pr_\nu\left\{
|\langle x \vert h\rangle|^2>\alpha
\right\}
=
\frac1{18}.
\label{eq:alice-normalization-main}
\end{equation}
Dividing the joint success probability in \cref{eq:acceptance-integral-no-comm-main} by Alice's success probability in \cref{eq:alice-normalization-main} gives
\begin{equation}
  \Pr_\nu \left\{
|\langle y \vert h\rangle|^2>\beta
 \bigm| 
|\langle x \vert h\rangle|^2>\alpha
\right\}  = \xy^2 .
\end{equation}
This already contains an essential ingredient of our simulation protocol: if $\theta$ were sampled from $\nu$ conditioned on $\abs{ \langle x \vert h \rangle }^2 > \alpha$, then Bob's success would exactly reproduce the quantum probability $p_Q(0 \mid x, y) = \xy^2$.
Such conditioned sampling is not directly possible because it depends on Alice's input $x$, which is received only after the shared variable is generated, while an unconditioned sample passes her test with probability only $1/18$.

A natural way to implement this conditioning is to repeat the sampling until Alice's test succeeds.
Accordingly, before receiving their inputs, Alice and Bob share an infinite sequence
\begin{equation}
\lambda=(\theta_1,\theta_2,\ldots),
\qquad
\theta_k=(h_k,\alpha_k,\beta_k),
\label{eq:average-protocol-ordinary-coordinates-main}
\end{equation}
whose coordinates are sampled independently according to $\nu$.
The shared randomness space is therefore
\begin{equation}
\Lambda=(\CP^2\times\mathcal R)^{\mathbb N},
\qquad
\mu=\nu^{\otimes\mathbb N}.
\end{equation}
After receiving $x$, Alice considers the samples in order and communicates the index of the first one that passes her test,
\begin{equation}
c(x,\lambda)
=
\min\{ k\geq 1 \bigm| \abs{\langle x \vert h_k\rangle}^2>\alpha_k \}.
\end{equation}
After receiving $c$ and his input $y$, Bob outputs $0$ when $\abs{\langle y \vert h_c\rangle}^2>\beta_c$ and outputs $1$ otherwise; see \cref{fig:average-communication-protocol}(b).
In terms of the classical model, these rules define
\begin{subequations}
\begin{align}
p_A(c \mid x, \lambda)
&= \1_{{c = c(x,\lambda)}}, \\
p_B(0 \mid y, c, \lambda)
&=
\1_{{\abs{\langle y \vert h_k\rangle}^2 > \beta_k}},
\end{align}
\end{subequations}
with $p_B(1\mid y,c,\lambda)=1-p_B(0\mid y,c,\lambda)$.

Because the samples are independent, the selected sample $\theta_c$ is distributed exactly as a single $\theta\sim\nu$ conditioned on Alice's test succeeding.
The conditional identity derived above therefore shows that Bob outputs $0$ with probability
\begin{equation}
p_C(0\mid x,y)
= \xy^2
= p_Q(0\mid x,y).
\end{equation}
Each sample passes Alice's test with probability $1/18$, so she finds a successful sample almost surely.
In this protocol, the length of the message is finite almost surely.
Using a simple but suboptimal encoding, \cref{rmk:average-communication-cost} shows an average communication cost of $18$ bits.
However, the worst-case cost of this protocol is unbounded: for every finite $N$, there is a positive probability that Alice rejects the first $N$ samples and must communicate an index larger than $N$.

To obtain a finite message alphabet, the natural first idea is to stop Alice's search after the first $N$ samples and, in the case where non index is accepted, provide an alternative procedure.
A protocol of this form is presented in the Methods section; see also \cref{fig:worst-case-communication-protocol}.
Its extension to generalized measurements with any number of outcomes and mixed states is presented in \cref{sec:d3-more-outcomes}.

\section*{Classical simulation of Bell nonlocality}

We now turn to a closely related simulation problem, studied since the early 1990s in the context of Bell nonlocality \cite{maudlin1992bell,BrassardCost1999,Steiner2000,CerfGisinMassar2000,massar2001classical,DegorreLaplanteRoland2005}. In the quantum communication setting considered above, Alice receives a state $\rho$ and the goal is to reproduce the statistics obtained when Bob measures it. Here, instead, we fix a bipartite state $\rho_{AB}$ and ask whether the statistics generated by arbitrary local measurements on this state can be reproduced by a classical communication model.

Let $\rho_{AB}$ be a state on $\C^{d_A}\otimes\C^{d_B}$. Alice and Bob choose arbitrary local measurements $M^A=\{M_a^A\}$ and $M^B=\{M_b^B\}$. Quantum theory assigns the joint probabilities
\begin{equation}
p_Q(a,b\mid M^A,M^B)
=
\tr\left[
\rho_{AB}
\left(
M_a^A\otimes M_b^B
\right)
\right].
\label{eq:quantum-bell-probability-main}
\end{equation}

Since this quantum experiment involves no communication between Alice and Bob, it is natural to first ask whether its statistics can be reproduced using shared randomness alone. States for which this is possible for all local measurements are called \emph{Bell local}. Not every bipartite state has this property: there exist entangled states whose local measurement statistics admit no such classical description \cite{Bell:1964PHY,Clauser:1969PRL,Freedman:1972PRL,Aspect:1982PRL,Weihs:1998PRL}. Such states are called \emph{Bell nonlocal}.

For Bell nonlocal states, we may then ask whether the missing classical resource can be supplied by communication and, if so, how much is required. In this model, Alice produces an outcome $a$ together with a message $c\in[d_C]$, while Bob uses $c$, his measurement $M^B$, and shared randomness $\lambda$ to produce $b$, \cref{fig:pam-scenario}(b). The resulting statistics have the form
\begin{equation}
\begin{split}
p_C(a,b\mid M^A,M^B)
=
&\int_\Lambda
\sum_{c=1}^{d_C}
p_A(a,c\mid M^A,\lambda)
\\
&\quad\times
p_B(b\mid M^B,c,\lambda)
\,\dd\mu(\lambda),
\end{split}
\label{eq:lhv+cc-model-main}
\end{equation}
A more precise definition is given in \cref{def:bell-classical-simulation}.

An exact simulation of a fixed state $\rho_{AB}$ requires
\begin{equation}
p_C(a,b\mid M^A,M^B)
=
p_Q(a,b\mid M^A,M^B)
\end{equation}
for every pair of local measurements and outcomes. If some finite $d_C$ suffices for every state on $\C^{d_Q}\otimes\C^{d_Q}$, we say that Bell nonlocality in dimension $d_Q$ can be simulated with finite classical communication. Conversely, if the local measurement statistics of some state in this dimension require an infinite message dimension, then Bell nonlocality in dimension $d_Q$ cannot be simulated with finite classical communication.

Finite-communication simulations of Bell nonlocality have been investigated in a variety of settings \cite{maudlin1992bell,BrassardCost1999,Steiner2000,CerfGisinMassar2000,massar2001classical,DegorreLaplanteRoland2005,TonerBacon,RennerTavakoliQuintino}. So far, the problem was completely solved only for qubits,
where two bits of classical communication suffice to reproduce the statistics of arbitrary local POVMs on any two-qubit state \cite{TonerBacon,SchlosserKleinmann2026, RennerTavakoliQuintino}. 

The connection with the quantum communication problem gives useful reductions in both directions. It is known that any finite classical simulation of arbitrary $d$-dimensional quantum communication can be converted into a simulation of the local measurement statistics of any bipartite state whose Bob subsystem has dimension $d$, without increasing the message dimension\cite{CerfGisinMassar2000,RennerTavakoliQuintino}. Consequently, our qutrit protocol immediately implies that Bell nonlocality remains finitely simulable whenever Bob's subsystem is a qutrit; in particular, $357$ bits of classical communication suffice. We give the reduction in \cref{sec:bell-qutrit}.

For maximally entangled states, a converse relation also holds at the level of finiteness. A construction closely related to quantum teleportation \cite{bennett1993teleporting} converts a finite classical simulation of all local measurement statistics on the maximally entangled state of dimension $d$, $\ket{\Phi_d}$, into a finite classical simulation of arbitrary $d$-dimensional quantum communication. Combined with our quantum communication result, this already implies that the local measurement statistics of $\ket{\Phi_4}$ require an infinite classical message dimension; see \cref{sec:bell-ququart-full}. This construction uses a $d^2$-outcome measurement on Alice's side, closely reflecting the structure of teleportation. It is therefore natural to ask whether the same obstruction already appears for the minimal scenario of binary measurements.

This question becomes particularly interesting in light of the result of Regev and Toner \cite{RegevToner}. They showed that, for arbitrary finite-dimensional bipartite states and binary local measurements, the corresponding quantum correlators can be simulated exactly using only two bits of classical communication. Their protocol, however, reproduces only the correlator and not in general the complete joint distribution: its local marginals are uniform and may differ from the quantum marginals. Thus finite communication for binary correlators does not imply finite communication for the full local measurement statistics.

Our result shows a sharp separation between these two tasks. For the maximally entangled ququart, already a family of rank-one binary projective measurements has full statistics that cannot be reproduced with any finite classical message dimension, even though their correlators admit a two-bit simulation by the Regev--Toner result.

The proof follows the same zero-slope mechanism as in the quantum communication problem. Consider the maximally entangled ququart state $\rho_{AB}=\ketbra{\Phi_4}$, where $\ket{\Phi_4} = (1/2) \sum_{j=0}^3\ket{jj}$. For $x,y\in\CP^3$, define $P_x=\ketbra{x}$ and $f(x,y)=\tr(P_xP_y)$, and let Alice and Bob perform the binary projective measurements
\begin{equation}
    M_x^A=\{P_{\bar{x}},\eye_4-P_{\bar{x}}\},
    \quad
    M_y^B=\{P_y,\eye_4-P_y\},
\end{equation}
where $\bar{x}$ denotes complex conjugation in the computational basis. For the outcome $a=b=0$,
\begin{equation}
    p_{\Phi_4}(0,0\mid M_x^A,M_y^B)
    =\frac{1}{4}f(x,y) ,
\end{equation}
where we have used the fact that $\tr( \Phi_4 \, P_{\bar{x}} \otimes P_y) = (1 / 4) \tr( P_x P_y)$.
A finite classical simulation would therefore require
\begin{equation}
\begin{split}
    \frac{1}{4}f(x,y) = &
    \int_\Lambda \sum_{c=1}^{d_C}
    p_A(0,c\mid M_x^A,\lambda) \\
    &\quad \times
    p_B(0\mid M_y^B,c,\lambda)
    \,\dd\mu(\lambda).
\end{split}
\end{equation}
Apart from the factor $1/4$, this has the same structure appearing in the quantum communication case.
From here on, the same reasoning applies: one considers the fixed-fidelity distribution $\eta_t$ and obtains that the slope of the quantum fidelity curve is $a_Q'(0)=1/4$.
Since a nonnegative factor plays no role in the zero-slope argument (\cref{lem:zero-slope}), we conclude that for any finite $d_C$, the condition $a_C(0)=0$ forces $a_C'(0)=0$.This contradicts $a_Q'(0)=1/4$.

Hence no finite classical message dimension can reproduce these binary local measurement statistics of the maximally entangled ququart. The complete argument is given in \cref{sec:bell-ququart-binary}, and the same obstruction extends to maximally entangled states in every dimension $d\geq4$. We discuss this contrast with the Regev--Toner result in more detail in \cref{sec:regev-toner}.

\section*{Conclusion}

Since the earliest days of quantum information science, quantum systems have been known to perform certain communication tasks more efficiently than classical systems \cite{BuhrmanQuantumClassical1998,BarYossefJayramKerenidis2004,Barreiro2008Beating,Kumar2019Experimental,Zhong2021EfficientExperimental}.
Such advantages raise a fundamental question about the nature of quantum information: do quantum states merely provide a more efficient way of implementing classical information protocols, or are they qualitatively different even at the level of their observable statistics?
Quantum states are, of course, not physically equivalent to classical states.
But, all information about a quantum state is ultimately obtained through classical measurement outcomes.
Our results resolve this question for arguably the two most elementary building blocks of quantum information protocols, namely, direct quantum communication and bipartite quantum correlations.
For qubits and qutrits, the difference remains purely quantitative, whereas starting at ququarts, a finite classical explanation of their observable statistics is no longer possible.

A natural next step is to turn this separation into an explicit sequence of communication tasks and Bell tests whose classical simulation cost grows progressively even at fixed quantum dimension.
Finding a sequence of tasks with some noise tolerance and efficient verification would yield experimentally testable separations of unbounded size.
This would have direct consequences to \textit{quantum information supremacy} proofs.
Aaronson, Buhrman, and Kretschmer proposed this notion as an unconditional quantum advantage in information resources, and it was recently demonstrated by recasting a one-way communication task as state preparation followed by measurement on a quantum processor \cite{aaronson_et_al:LIPIcs.ITCS.2024.1,KretschmerEtAl2025}.
In that construction, as in earlier asymptotic communication separations \cite{BarYossefJayramKerenidis2004,Montanaro2019quantumstatescannot}, larger gaps require larger quantum systems.
Our results instead suggest the classical cost can grow arbitrarily while the quantum system remains fixed at as few as two qubits.

\section*{Acknowledgments}

We are grateful to Marco Túlio Quintino and Kiara Hansenne for helpful discussions and comments on an earlier version of this article.

C.G. was partially supported by the ANR for the JCJC grants LINKS (No.
ANR-23-CE47-0003), the T-ERC QNET (No. ANR24-ERCS-0008), the project QUANTINT, as well as the European Union's Horizon 2020 Research and Innovation Programme under QuantERA Grant Agreements No. 731473 and No. 101017733. 
T.S.R.S was partially supported by FAPESP Grant No. 2024/15587-1.
C.V was partially supported by FAPESP Grant No. 2024/16657-3, 2025/01058-0 and 2025/27304-7 and by the Conselho Nacional de Desenvolvimento Científico e Tecnológico (CNPq), through a grant from the Conhecimento Brasil Program - Line 1.

\section*{Statement on AI usage}

After the authors reduced the quantum communication problem to the state exclusion formulation, OpenAI GPT-5.5 and GPT-5.6 Sol were used, through extended dialogues, to explore proof strategies.
Useful insights, including core ideas for the case $d_Q = 4$ and an early version of the finite communication protocol for $d_Q = 3$, were suggested by the models.
These suggestions were treated exclusively as unverified preliminary material.
The development and synthesis of the mathematical arguments was carried out by the authors, who selected the useful ideas and independently verified, solved technical issues, simplified, generalized and wrote the proofs and protocols.
The resulting article, including all arguments, figures and references was produced by the authors, who take full responsibility for the correctness of the results.

\section*{Methods}

\subsection*{Outline of the zero slope proof}

As previously discussed, the crucial step in the proof that there exists no finite classical simulation of ququarts is to show that for any finite $d_C$, forcing $a_C(0) = 0$ implies instead that $a_C'(0) = 0$.
Let us analyze this in more detail.

Consider separately each value $\lambda$ of the shared randomness and each message $c$, and let
\begin{align}
A_{\lambda,c}(x)&:=p_A(c\mid x,\lambda), \\
B_{\lambda,c}(y)&:=p_B(0\mid y,c,\lambda).
\end{align}
Both these functions are measurable and bounded, with values in $[0,1]$, and since $\sigma$ is a probability measure, they belong to $L^2(\CP^3)$.

We write $\eta_t(\dd y\mid x)$ for the uniform distribution over measurements $y$ having fidelity $t$ with a fixed state $x$.
Using this conditional distribution, we can define for any $B \in L^2(\CP^3)$ the averaged decoder response
\begin{equation}
(R_tB)(x) := \int_{\CP^3}B(y) \eta_t(\dd y\mid x).
\end{equation}
Thus, $R_tB$ is obtained by averaging $B$ over all measurement rays $y$ having fidelity $t$ with $x$.
This allows us to rewrite each fixed branch $(\lambda, c)$ of the function $a_C(t)$ defined in \cref{eq:classical-simulation-average-main} as an inner product
\begin{align}
a_{\lambda,c}(t)
&:=
\big\langle A_{\lambda,c} \vert R_tB_{\lambda,c}\big\rangle \nonumber\\
&=
\int_{\CP^3}
A_{\lambda,c}(x)(R_tB_{\lambda,c})(x) \dd\sigma(x),
\end{align}
from which $a_C(t)=\int_\Lambda\sum_{c=1}^{d_C} a_{\lambda,c}(t) \dd\mu(\lambda)$, where the exchange of the integration order is justified by Tonelli's theorem, since the integrand is nonnegative and measurable.
The special role of the point $t = 0$ now becomes apparent: as the contribution of every branch is nonnegative, $a_C(0)=0$ forces $a_{\lambda,c}(0)=0$ for every $c$ and $\mu$-almost every $\lambda$.
Furthermore, the right derivative of each branch, if it exists, is given by
\begin{equation}
    a'_{\lambda, c}(0) = \lim_{t \to 0^+} \Big\langle A_{\lambda, c} \Big\vert \frac{R_t - R_0}{t} B_{\lambda, c} \Big\rangle ,
\end{equation}
so determining its slope reduces to understanding the variation of $R_t$ near $t = 0$.

To understand the effect of the averaging operator $R_t$, we use the spectral decomposition of the nonpositive Laplace--Beltrami operator $\Delta$.
This operator is the standard generalization of the Laplacian from an Euclidean to a curved space and describes the second-order variation of a function.
On every complex projective space $\CP^q$, the eigenspaces of $\Delta$ form a complete orthogonal decomposition of $L^2(\CP^q)$ \cite{Koornwinder,Lu1998}.
This decomposition is the analogue on complex projective spaces of the Fourier decomposition of a function on the circle.
For $\CP^3$, we denote the $\ell$-th eigenspace by $\mathcal H_\ell$ and write the decomposition as
\begin{equation}
L^2(\CP^3) = \widehat{\bigoplus}_{\ell\geq 0}\mathcal H_\ell. 
\end{equation}
Then every $B\in L^2(\CP^3)$ can be written as $B=\sum_{\ell\geq 0}B_\ell$, with $B_\ell \in \mathcal{H}_\ell$ and the corresponding eigenvalues
\begin{equation}
\Delta B_\ell = -4\ell(\ell+3)B_\ell.
\label{eq:laplacian-eigenvalues-main}
\end{equation}

In \cref{prop:multiplier}, we show that due to the unitary symmetry, $R_t$ preserves each angular mode of $B$, changing only its amplitude:
\begin{equation}
R_tB_\ell = \phi_\ell (t) B_\ell ,
\label{eq:Rt-on-eigenfunctions-main}
\end{equation}
where the scalars $\phi_\ell(t)$ satisfy
\begin{align}
\phi_\ell(0) &= \frac{2(-1)^\ell}{(\ell + 1) (\ell + 2)}, \label{eq:coefficient-square-decay-main}\\
\phi_\ell'(0) &= -\ell(\ell+3)\phi_\ell(0). \label{eq:coefficient-derivative-main}
\end{align}
The quadratic decay $\abs{\phi_\ell(0)} \simeq \ell^{-2}$ is a key feature of the argument.
First, it exactly compensates for the quadratic growth of the eigenvalues in \cref{eq:laplacian-eigenvalues-main}, so that $\Delta R_0B \in L^2(\CP^3)$.
Thus, although $B$ need not be differentiable, the averaged function $R_0B$ has the two derivatives needed for its Laplacian $\Delta R_0B$ to be well defined as an $L^2(\CP^3)$ function.
Second, combining \cref{eq:laplacian-eigenvalues-main,eq:Rt-on-eigenfunctions-main,eq:coefficient-square-decay-main,eq:coefficient-derivative-main} shows that, mode by mode,
\begin{equation}
    \phi_\ell'(0) B_\ell = \frac{1}{4} \Delta( \phi_\ell(0) B_\ell ) = \frac14\Delta R_0B_\ell .
\end{equation}

The next step is to extend this mode-by-mode identity to the full function $B$.
Because $B$ may contain infinitely many modes, a relation proved for each mode separately does not automatically remain valid after all modes are summed.
One must first justify exchanging the limit $t \to 0^+$ with this infinite sum.
This is established in \cref{sec:mode-by-mode-to-whole-function}, yielding
\begin{equation}
\frac{R_t-R_0}{t}B \longrightarrow \frac{1}{4}\Delta(R_0B) \;\text{ as }t\to0^+.
\label{eq:fidelity-laplacian}
\end{equation}
Applying \cref{eq:fidelity-laplacian} to each branch $(\lambda, c)$ then gives
\begin{equation}
a_{\lambda,c}'(0)
=
\frac{1}{4} \big\langle A_{\lambda,c} \vert \Delta R_0B_{\lambda,c} \big\rangle.
\label{eq:branch-derivative-main}
\end{equation}

It remains to show that the inner product on the right-hand side above vanishes for every $c$ and $\mu$-almost every $\lambda$, and then to combine these branchwise conclusions to obtain $a_C'(0)=0$.
Since $R_0B_{\lambda, c}$ is nonnegative, perfect exclusion implies that
\begin{equation}
    A_{\lambda, c}(x) R_0B_{\lambda, c}(x) = 0
\end{equation}
for $\sigma$-almost every $x$.
Equivalently, $A_{\lambda,c}$ can only be nonzero on the set where $R_0B_{\lambda,c} = 0$.
The regularity implied by the $\ell^{-2}$ decay in \cref{eq:coefficient-square-decay-main} is sufficient to show that $\Delta R_0B_{\lambda, c}$ also vanishes $\sigma$-almost everywhere on this zero set; see \cref{sec:perfect-exclusion-forces-zero}.
Consequently, $a_{\lambda, c}'(0) = 0$ for every $c$ and $\mu$-almost every $\lambda$.
Averaging over $\lambda$ and summing over any finite number of classical messages as in \cref{eq:classical-simulation-main} then yields $a_C'(0)=0$, in contradiction to $a_Q'(0) = 1$; see \cref{sec:d4-proof-of-main-theorem}.

Interestingly, the $\ell^{-2}$ decay in \cref{eq:coefficient-square-decay-main} is exactly the threshold needed for $R_0B$ to acquire the two derivatives used above. The decay is faster for $d_Q > 4$, but too slow to provide this regularity for $d_Q < 4$. This is discussed in \cref{sec:dimension-threshold} and explains why the proof applies precisely from $d_Q=4$ onward.

\subsection*{Qutrit protocol with finite communication}

\begin{figure*}[t]
	\centering
		\includegraphics[width=0.93\linewidth]{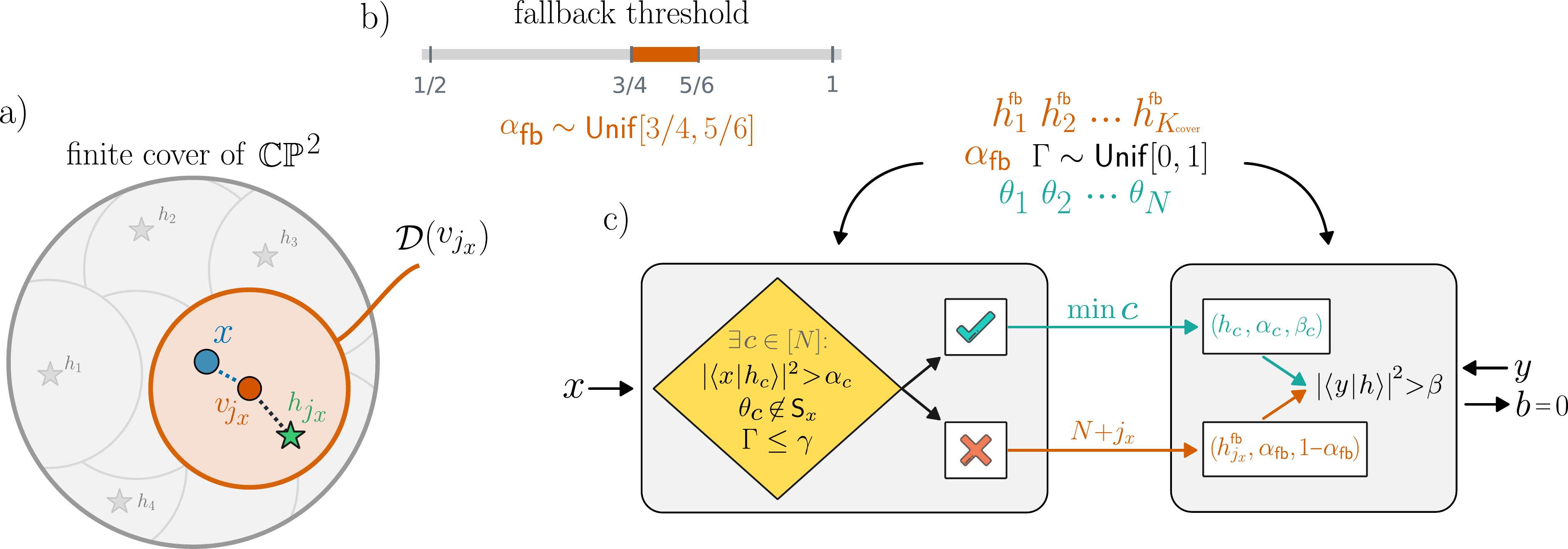}
	\caption{
    \textbf{Qutrit simulation protocol with finite communication.}
        (a) A cover $\CP^2$ by $K_{\mathrm{cover}}$ balls $\mathcal D(v_j)$. Alice selects the smallest $j_x$ where $x \in \mathcal{D}(v_{j_x})$. A ray $h_j^{\mathrm{fb}} \in \mathcal D(v_j)$ is sampled uniformly from each ball.
        (b) The fallback thresholds are restricted to $\alpha_{\mathrm{fb}}\in[3/4,5/6]$ and $\beta_{\mathrm{fb}}=1-\alpha_{\mathrm{fb}}$, and the sets $\mathrm{S}_{x} \subset \Lambda$, defined in \cref{eq:reserved-slice}, are reserved for the fallback.
        (c) Alice searches a finite set of $N$ ordinary coordinates $\theta_k$.
        If there is at least one accepted ordinary coordinate (yellow lozenge), she communicates the smallest accepted index.
        Otherwise, she communicates $j_x$.
        Bob selects the shared coordinate $(h, \alpha, \beta)$ corresponding to Alice's message and outputs $0$ if $\abs{ \langle y \vert h \rangle }^2 > \beta$ or $1$ otherwise.
        The shared random number $\Gamma$ is used to keep an ordinary success with probability $\gamma$, ensuring that the ordinary and fallback branches combined with the exact probabilities necessary to reproduce the quantum statistics.
	}
	\label{fig:worst-case-communication-protocol}
\end{figure*}

Here we present the the protocol for qutrit simulation with finite classical communication.
Recall that the protocol presented previously has finite average communication, but for every finite $N$, there is a positive probability that Alice rejects the first $N$ proposals $\theta_k$, and therefore must communicate an index larger than $N$.
To transform this into a finite message alphabet, the natural first idea is to stop Alice's search after a sufficiently large number of samples which, for reasons explained later, will turn out to be $N = 112$.
We call this part of the protocol the \emph{ordinary branch}.
This creates one difficulty: with nonzero probability, Alice rejects all $N$ samples and has no accepted index to send.
We handle such cases using a second procedure, called the \emph{fallback branch}, which is used whenever the ordinary branch is not selected (the exact conditions for her choice of branch will be specified later).

For the fallback to complement the ordinary branch, it must be defined in every situation and provide Alice with a valid message for any state $x$.
A finite covering of $\CP^2$ provides a natural way to guarantee this.
For this covering, we will use balls $\mathcal D(v_j)$ centered at $v_j\in\CP^2$, all with radius $r=1/\sqrt{24}$.
This choice of $r$ is convenient rather than optimized and directly affects the number of classical messages required by the protocol; see \cref{rem:binary-protocol-optimization}.
In \cref{lem:finite-cap-cover} shows that $\CP^2$ admits such a cover with $K_{\mathrm{cover}}\leq (2\sqrt{24})^{4}= 9216$, and using the metric presented in \cref{sec:finite-covering}, any two rays $x, h$ in the same ball satisfy $| \langle x \vert h\rangle |^2 > 5/6$.
We fix a covering and let it be known to both parties; see \cref{fig:worst-case-communication-protocol}(a).

In addition to the covering, the fallback will require a few other shared random variables.
Alice and Bob sample one ray $h^{\mathrm{fb}}_j$ uniformly from each ball $\mathcal D(v_j)$ and a threshold
\begin{equation}
	\alpha_{\mathrm{fb}}
	\sim\operatorname{Unif}\left[\frac34, \,\frac56 \right],
	\quad
	\beta_{\mathrm{fb}}=1-\alpha_{\mathrm{fb}} .
\end{equation}
These will constitute the fallback coordinates
\begin{equation}
\theta_{\mathrm{fb}}
= \bigl(h_j^{\mathrm{fb}}, \alpha_\mathrm{fb}, 1 - \alpha_{\mathrm{fb}} \bigr) ,
\end{equation}
which play a similar role to the ordinary coordinates defined in \cref{eq:average-protocol-ordinary-coordinates-main}.
Lastly, they sample a uniform random number $\Gamma\in[0,1]$.
This will serve as a correction to combine the ordinary and the fallback procedures with the required probabilities to reproduce the quantum behavior.
The shared variable of this protocol is therefore
\begin{equation}
	\lambda
	=
	\bigl(\theta_1,\ldots,\theta_N,\, \Gamma,\, h^{\mathrm{fb}}_1,\ldots,h^{\mathrm{fb}}_{K_{\mathrm{cover}}},\, \alpha_{\mathrm{fb}}	\bigr).
\end{equation}

Alice's usage of the fallback is straightforward:
after receiving $x$, she determines the first covering ball that contains it, $j_x := \min\{j:x\in\mathcal D(v_j)\}$.
Since the shared ray $h^{\mathrm{fb}}_{j_x}$ was sampled from this same ball, its squared overlap with $x$ is larger than $5/6$.
Therefore the fallback coordinate $\bigl( h^{\mathrm{fb}}_{j_x}, \alpha_{\mathrm{fb}}, 1-\alpha_{\mathrm{fb}}\bigr)$
is always successful in Alice's test, because $| \langle x \vert h_{j_x}^{\mathrm{fb}} \rangle |^2 > 5/6 \ge  \alpha_{\mathrm{fb}}$.

It remains to specify how the ordinary and the fallback branches are combined so that the resulting protocol simulates the quantum behavior.
We first reserve parts of the shared randomness space for the fallback.
These parts are denoted $\mathrm{S}_{x}$ and defined in \cref{eq:reserved-slice}.
Alice searches the first $N$ ordinary tests $\theta_k$ for one that is accepted and lies outside this reserved part.
If she finds such a candidate and $\Gamma\leq\gamma$ (where $\gamma$ is a constant defined in \cref{eq:finite-protocol-constants} to combine the branches with the correct probabilities), she communicates the accepted index $k \in \{1, \ldots, N\}$.
Otherwise, she communicates the fallback label $N+j_x \in \{ N+1, \ldots, N+K_{\mathrm{cover}} \}$.
Bob then selects the corresponding ordinary candidate $\theta_c$ or the fallback candidate $\theta^{\mathrm{fb}}_{j_x} = (h^{\mathrm{fb}}_{j_x}, \alpha_{\mathrm{fb}}, 1 - \alpha_{\mathrm{fb}})$, and we denote his selected coordinate by $(h, \alpha, \beta)$. He outputs $0$ if $\abs{ \langle y \vert h \rangle }^2 > \beta$ or $1$ otherwise. This protocol is illustrated in \cref{fig:worst-case-communication-protocol}(c).

In \cref{sec:worst-case-correctness} we show that this protocol exactly reproduces $p_Q(0\mid x,M)$.
The classical communication used is of dimension $d_C = N + K_{\mathrm{cover}} = 9328 < 2^{14}$ in every round, thus the finiteness of the qutrit simulation problem is proven in the case of binary projective measurements.

The extension to mixed states and generalized measurements with more outcomes is presented in \cref{sec:d3-more-outcomes}.
It is done by first binarizing the measurement and then carefully combining finitely many independent copies of the binary protocol above.
The resulting protocol can be implemented with $d_C = (2\sqrt{24})^4 [112+(2\sqrt{24})^4]^{26} = 9216(9328)^{26} < 2^{357}$, proving that $C_{\mathrm{cl}}(3) < 2^{357}$.
This proves finiteness but does not solve the optimal cost. Already for binary measurements, our construction uses $d_C = 9328$, whereas a lower bound of five is known for binary measurements \cite{SchlosserKleinmann2026}. Note that no effort was made to optimize our protocol.
In fact, a more careful choice of the constants used can already lead to improvements; see \cref{rem:binary-protocol-optimization}.

\printbibliography

\clearpage
\onecolumn

\newgeometry{
    margin=1in
}

\appendix
\etocdepthtag.toc{appendices}

\crefalias{section}{appendix}
\crefalias{subsection}{appendix}
\crefalias{subsubsection}{appendix}

\fontsize{11pt}{13.6pt}\selectfont

\titleformat{\subsubsection}
  {\normalfont\fontsize{11pt}{13.6pt}\selectfont\bfseries}
  {\thesubsubsection}{1em}{}

\titleformat{\paragraph}[runin]
  {\normalfont\fontsize{11pt}{13.6pt}\selectfont\bfseries}
  {}{0pt}{}[]

\section*{Appendices}

\etocsettagdepth{appendices}{subsubsection}
\etocsettagdepth{main}{none}
\tableofcontents

\clearpage
\pagestyle{appendixstyle}

\section{Preliminaries}
\label{app:preliminaries}
    \input{2_appendix_introduction}
    
\section{No finite classical simulation for ququart communication}
\sectionmark{Ququarts}
\label{app:d4}
    \input{3_appendix_d4proof}

\section{The problem of classical simulation in Bell Nonlocality}
\sectionmark{Bell Nonlocality}
\label{app:bell}
    \input{7_appendix_bell_v2}   

\section{Finite classical simulation for qutrit communication}
\sectionmark{Qutrits}
\label{app:d3}
    \input{4_appendix_d3proof_v3}

\section{Additional calculations}
\label{app:technical-calculations}
    \input{5_appendix_technicalities}

\end{document}

%% file: 0_preamble.tex
\usepackage[utf8]{inputenc}
\usepackage[T1]{fontenc}
\usepackage[english]{babel}
\usepackage{etoc}
\usepackage{titlesec}
\usepackage{authblk}

\usepackage{fancyhdr}
\fancypagestyle{appendixstyle}{
  \fancyhf{}
  \fancyhead[L]{\nouppercase{\leftmark}}
  \fancyhead[R]{\nouppercase{\rightmark}}
  \fancyfoot[C]{\thepage}
  
}
\renewcommand{\sectionmark}[1]{%
  \markboth{App.\ \thesection: #1}{}%
}

\usepackage{lmodern}
\usepackage{amsmath,amssymb,amsthm,mathtools,thmtools,dsfont,physics,empheq}
\usepackage[scr=boondoxo]{mathalpha}
\usepackage[dvipsnames]{xcolor}
\usepackage{enumitem}
\setlist[itemize]{
  itemsep=4pt,
  topsep=6pt,
  parsep=0pt,
  partopsep=0pt
}
\usepackage{csquotes}

\usepackage{microtype}
\usepackage{graphicx,subcaption}

\usepackage{nicefrac}
\usepackage{babel}

\usepackage[
  backend=biber,
  style=numeric-comp,
  giveninits=true,
  sorting=none,
  maxbibnames=99,
  maxcitenames=4,
  doi=false,
  url=false
]
{biblatex}

\DeclareFieldFormat[article]{title}{%
  \mkbibquote{%
    \iffieldundef{doi}
      {\iffieldundef{url}
         {#1}
         {\href{\thefield{url}}{#1}}}
      {\href{https://doi.org/\thefield{doi}}{#1}}%
  }%
}
\AtEveryBibitem{%
  \iffieldundef{journaltitle}
    {}
    {\clearfield{eprint}%
     \clearfield{eprinttype}%
     \clearfield{eprintclass}}%
}

\definecolor{DarkRed}{rgb}{0.8,0,0}
\usepackage[
    linktocpage=true,
	colorlinks,
	linkcolor=DarkRed,
    citecolor=ForestGreen,
    bookmarks,
    bookmarksopen,
    bookmarksnumbered]
{hyperref}

\usepackage[normalem]{ulem}

\usepackage[nameinlink,capitalise]{cleveref}
\crefname{section}{Sec.}{Secs.}
\crefname{appendix}{App.}{Apps.}
\Crefname{section}{Appendix}{Appendices}

\theoremstyle{plain}
\newtheorem{theorem}{Theorem}[section]
\newtheorem{lemma}[theorem]{Lemma}
\newtheorem{proposition}[theorem]{Proposition}
\newtheorem{corollary}[theorem]{Corollary}
\newtheorem{definition}[theorem]{Definition}
\newtheorem{remark}[theorem]{Remark}

\RenewCommandCopy{\theHtheorem}{\thetheorem}
\RenewCommandCopy{\theHlemma}{\thelemma}
\RenewCommandCopy{\theHproposition}{\theproposition}
\RenewCommandCopy{\theHcorollary}{\thecorollary}
\RenewCommandCopy{\theHdefinition}{\thedefinition}
\RenewCommandCopy{\theHremark}{\theremark}

\usepackage[most]{tcolorbox}

\newtcolorbox{simplebox}{
  enhanced,
  breakable,
  colback=gray!5,
  colframe=gray!50,
  boxrule=0.5pt,
  arc=2pt,
  left=8pt,
  right=8pt,
  top=6pt,
  bottom=6pt,
  before skip=10pt,
  after skip=10pt
}
\newcommand{\R}{\mathbb R}
\newcommand{\C}{\mathbb C}
\newcommand{\N}{\mathbb N}
\newcommand{\CP}{\mathbb{CP}}
\newcommand{\1}{\mathbf 1}

\newcommand{\xy}{\lvert \langle x | y \rangle \rvert}
\newcommand{\eye}{\mathds{1}}

\renewcommand{\dd}{\mathrm d}

\newcommand{\cH}{\mathcal H}
\renewcommand{\ip}[2]{\left\langle #1 \middle\vert #2\right\rangle}
\renewcommand{\norm}[1]{\left\lVert #1\right\rVert}
\renewcommand{\abs}[1]{\left\lvert #1\right\rvert}

%% file: 2_appendix_introduction.tex
\subsection{Mathematical setting and notation}
The proofs of our main results involve averaging over continuous families of quantum states and measurements, as well as over the shared randomness used by a classical simulation. We therefore begin by fixing the relevant measurable structures and notation. Since all quantum state and measurement spaces considered here are finite dimensional, their natural topology provides canonical Borel $\sigma$-algebras, which will allow us to formulate the simulation model and the integrals appearing throughout the proofs precisely.

For $d\in\N$, let
\begin{equation}
    \operatorname{Herm}(\C^d)=\{A\in\C^{d\times d}: A=A^\dagger\}
\end{equation}
be the real Hilbert space of Hermitian operators on $\C^d$, equipped with the Hilbert--Schmidt inner product
\begin{equation}
    \langle A,B\rangle_{\mathrm{HS}}=\tr(AB),
\qquad A,B\in\operatorname{Herm}(\C^d) .    
\end{equation}
All topological notions below, including compactness and Borel measurability, refer to the topology induced by this finite-dimensional Hilbert-space structure.

We write $A\succeq0$ when $A$ is positive semidefinite. The set of quantum states in dimension $d$ is
\begin{equation}
    \mathcal S_d=\{\rho\in\operatorname{Herm}(\C^d):
    \rho\succeq0,\ \tr\rho=1\}.
 \end{equation}
This is a compact subset of $\operatorname{Herm}(\C^d)$, and we equip it with the Borel $\sigma$-algebra $\mathcal B(\mathcal S_d)$.

An $n$-outcome positive-operator-valued measure (POVM) in dimension $d$
is a tuple $M = (M_1,\ldots,M_n)$ of Hermitian operators satisfying
\begin{equation}
    M_b\succeq0
    \quad\text{for every } b\in[n],
    \qquad
    \sum_{b=1}^n M_b=\eye_d, 
\end{equation}
where $\eye_d$ denotes the identity operator on $\mathbb C^d$ and $[n] := \{1,\ldots,n\}$. We denote the space of all such POVMs by,
\begin{equation}
     \mathcal M_{d,n}
 =\left\{(M_1,\ldots,M_n):M_b\succeq0,\ 
   \sum_{b=1}^nM_b=\eye_d\right\}.    
\end{equation}
As a compact subset of the finite-dimensional real vector space
$\operatorname{Herm}(\C^d)^n$, the space $\mathcal M_{d,n}$ carries the corresponding Borel $\sigma$-algebra
$\mathcal B(\mathcal M_{d,n})$.

If a preparation device prepares a state $\rho\in\mathcal S_d$ and a measurement device performs $M\in\mathcal M_{d,n}$, quantum theory assigns the conditional probabilities
\begin{equation}
p_Q(b\mid\rho,M)
:=
\tr(\rho M_b),
\qquad
b\in[n].
\label{eq:born-rule}
\end{equation}
These probabilities will be the target statistics of the classical simulation model introduced next.

\subsection{Classical simulation model for quantum communication}
\label{sec:classical-simulation-model}

We now formalize the notion of classical simulation used throughout the paper. The goal is to reproduce the statistics of quantum communication using only one-way classical communication and shared randomness. Instead of transmitting the quantum state $\rho$, the preparation device sends a classical message $c$ chosen from a finite alphabet. The preparation and measurement devices may additionally share an arbitrary random variable $\lambda$, sampled independently of the state $\rho$ and the measurement $M$. Their local responses may depend on $\lambda$, but the only information about $\rho$ that can reach the measurement device during the protocol is the communicated message $c$. An exact simulation is one that reproduces the quantum probabilities $\tr(\rho M_b)$ for every state and measurement.

\begin{definition}[classical simulation with finite message alphabet]
\label{def:classical-simulation}

Fix $d_Q,d_C\in\N$. We say that $d_Q$-dimensional quantum communication
admits an classical simulation with classical message dimension
$d_C$ if there exist:

\begin{enumerate}
    \item a probability space $(\Lambda,\Sigma_\Lambda,\mu)$ describing
    shared randomness sampled independently of the state and measurement;

    \item an encoder described by probabilities
    \begin{equation}
        p_A(c\mid\rho,\lambda),
        \qquad
        c\in[d_C],\ \rho\in\mathcal S_{d_Q},\ \lambda\in\Lambda,
    \end{equation}
    and, for every $n\geq1$, a decoder described by probabilities
    \begin{equation}
        p_B(b\mid M,c,\lambda),
        \qquad
        b\in[n],\ M\in\mathcal M_{d_Q,n},\
        c\in[d_C],\ \lambda\in\Lambda,
    \end{equation}
    satisfying
    \begin{equation}
        p_A(c\mid\rho,\lambda)\geq0,
        \qquad
        \sum_{c=1}^{d_C}p_A(c\mid\rho,\lambda)=1,
    \end{equation}
    and
    \begin{equation}
        p_B(b\mid M,c,\lambda)\geq0,
        \qquad
        \sum_{b=1}^n p_B(b\mid M,c,\lambda)=1;
    \end{equation}
    \item the measurability conditions that, for every $c\in[d_C]$,
    \begin{equation}
        (\rho,\lambda)\longmapsto p_A(c\mid\rho,\lambda)
    \end{equation}
    is $\mathcal B(\mathcal S_{d_Q})\otimes\Sigma_\Lambda$-measurable,
    and, for every $n\geq1$, $c\in[d_C]$, and $b\in[n]$,
    \begin{equation}
        (M,\lambda)\longmapsto p_B(b\mid M,c,\lambda)
    \end{equation}
    is $\mathcal B(\mathcal M_{d_Q,n})\otimes\Sigma_\Lambda$-measurable.

    \item exact reproduction of the quantum statistics: for every $n\ge1$, every $\rho\in\mathcal S_{d_Q}$, every $M=(M_1,\ldots,M_n)\in\mathcal M_{d_Q,n}$, and every $b\in[n]$,
    \begin{equation}
        \tr(\rho M_b) = \int_\Lambda \sum_{c=1}^{d_C} p_A(c\mid\rho,\lambda)\, p_B(b\mid M,c,\lambda) \,\dd\mu(\lambda). 
    \label{eq:classical-simulation}
    \end{equation}
\end{enumerate}
\end{definition}

No restriction beyond measurability is imposed on the shared randomness space $\Lambda$; in particular, it may be infinite or continuous. A classical message alphabet of size $d_C$ can be encoded using $\lceil\log_2 d_C\rceil$ bits.

\begin{remark}[Increasing the classical and quantum dimensions]
\label{rem:bigger_classical_and_quantum_dimensions}
If $d_Q$-dimensional quantum communication admits an exact classical simulation with message dimension $d_C$, then it also admits one with any $d_C'>d_C$. Indeed, one may use the original protocol and simply assign zero probability to the additional messages.

Conversely, if no finite classical simulation exists in dimension $d_Q$, then no such simulation can exist in any larger quantum dimension $d_Q'>d_Q$. To see this, choose an isometry
\begin{equation}
    V:\C^{d_Q}\longrightarrow\C^{d_Q'}.    
\end{equation}
The $d_Q$-dimensional states and measurements can then be embedded into dimension $d_Q'$ through this isometry in such a way that all measurement statistics are preserved. Hence, any exact simulation in dimension $d_Q'$ would restrict to an exact simulation of the embedded $d_Q$-dimensional scenario.
\end{remark}

\subsection{Bell scenarios and classical simulation of bipartite states}
\label{sec:bell-preliminaries}

There is a closely related simulation problem for bipartite quantum states. In the communication setting considered above, Alice receives a quantum state $\rho$ and the classical protocol must reproduce the statistics obtained when this state is measured by Bob. Here, instead, we fix a bipartite state $\rho_{AB}$ and ask whether the statistics generated by arbitrary local measurements on this state can be reproduced using shared randomness and classical communication.

More precisely, let $\rho_{AB}$ be a state on $\C^{d_A}\otimes\C^{d_B}$. For fixed numbers of outcomes $n_A$ and $n_B$, Alice receives an arbitrary local measurement
\begin{equation}
M^A=(M_1^A,\ldots,M_{n_A}^A)\in\mathcal M_{d_A,n_A},
\end{equation}
while Bob receives
\begin{equation}
M^B=(M_1^B,\ldots,M_{n_B}^B)\in\mathcal M_{d_B,n_B}.
\end{equation}
Quantum theory assigns the joint probabilities
\begin{equation}
p_{Q}(a,b\mid M^A,M^B)
=
\tr\left[
\rho_{AB}
\left(
M_a^A\otimes M_b^B
\right)
\right],
\label{eq:quantum-bell-state-behavior}
\end{equation}
for $a\in[n_A]$ and $b\in[n_B]$.

Since the target quantum experiment itself involves no communication between Alice and Bob, it is natural to first ask whether its statistics can be reproduced using shared randomness alone. A bipartite state $\rho_{AB}$ is called \emph{Bell local} if there exist a probability space $(\Lambda,\Sigma_\Lambda,\mu)$ and local response functions such that, for every pair of local measurements,
\begin{equation}
p_{\rho_{AB}}(a,b\mid M^A,M^B)
=
\int_\Lambda
p_A(a\mid M^A,\lambda)
p_B(b\mid M^B,\lambda)
\,\dd\mu(\lambda).
\label{eq:lhv-model}
\end{equation}
The shared random variable $\lambda$ is sampled independently of the measurements. Not every bipartite quantum state has this property. In particular, there exist entangled states whose local measurement statistics admit no such classical description. Such states are called \emph{Bell nonlocal}.

Thus, Bell nonlocal states are precisely those bipartite states whose local measurement statistics cannot be reproduced using shared randomness alone. This naturally leads to a quantitative extension of the problem: if classical communication is allowed, how much is required to reproduce these statistics exactly?

We consider one-way classical communication from Alice to Bob, in direct analogy with the model introduced in \cref{def:classical-simulation}. After receiving her local measurement $M^A$, Alice produces an output $a$ together with a classical message $c\in[d_C]$. Bob receives $c$ and uses it, together with his local measurement $M^B$ and the shared randomness, to produce his output $b$. The goal is to reproduce \cref{eq:quantum-bell-state-behavior} for all local measurements.

The target quantum experiment itself involves no communication between Alice and Bob. Classical communication is introduced only as a resource for simulating its statistics.

\begin{definition}[Classical simulation of a entangled state with finite message alphabet]
\label{def:bell-classical-simulation}

Fix $d_A,d_B,d_C\in\N$, and let
$\rho_{AB}\in\mathcal S_{d_A d_B}$ be a bipartite quantum state acting on
$\C^{d_A}\otimes\C^{d_B}$. We say that local measurements on $\rho_{AB}$ admits an classical
simulation with classical message dimension $d_C$ if there exist:

\begin{enumerate}
    \item a probability space $(\Lambda,\Sigma_\Lambda,\mu)$ describing
    shared randomness sampled independently of the local measurements;

    \item for every $n_A\geq1$, response probabilities
    \begin{equation}
        p_A(a,c\mid M^A,\lambda),
        \qquad
        a\in[n_A],\ c\in[d_C],\
        M^A\in\mathcal M_{d_A,n_A},
    \end{equation}
    and, for every $n_B\geq1$, response probabilities
    \begin{equation}
        p_B(b\mid M^B,c,\lambda),
        \qquad
        b\in[n_B],\ c\in[d_C],\
        M^B\in\mathcal M_{d_B,n_B},
    \end{equation}
    satisfying
    \begin{equation}
        p_A(a,c\mid M^A,\lambda)\geq0,
        \qquad
        \sum_{a=1}^{n_A}\sum_{c=1}^{d_C}
        p_A(a,c\mid M^A,\lambda)=1,
    \end{equation}
    and
    \begin{equation}
        p_B(b\mid M^B,c,\lambda)\geq0,
        \qquad
        \sum_{b=1}^{n_B}
        p_B(b\mid M^B,c,\lambda)=1;
    \end{equation}

    \item the measurability conditions that, for every $n_A\geq1$,
    $a\in[n_A]$, and $c\in[d_C]$,
    \begin{equation}
        (M^A,\lambda)
        \longmapsto
        p_A(a,c\mid M^A,\lambda)
    \end{equation}
    is
    $\mathcal B(\mathcal M_{d_A,n_A})\otimes\Sigma_\Lambda$-measurable,
    and, for every $n_B\geq1$, $b\in[n_B]$, and $c\in[d_C]$,
    \begin{equation}
        (M^B,\lambda)
        \longmapsto
        p_B(b\mid M^B,c,\lambda)
    \end{equation}
    is
    $\mathcal B(\mathcal M_{d_B,n_B})\otimes\Sigma_\Lambda$-measurable;

    \item exact reproduction of the quantum statistics: for every
    $n_A,n_B\geq1$, every
    $M^A=(M_1^A,\ldots,M_{n_A}^A)\in\mathcal M_{d_A,n_A}$,
    every
    $M^B=(M_1^B,\ldots,M_{n_B}^B)\in\mathcal M_{d_B,n_B}$,
    and every $a\in[n_A]$, $b\in[n_B]$,
    \begin{equation}
        \tr\left[
            \rho_{AB}
            \left(
                M_a^A\otimes M_b^B
            \right)
        \right]
        =
        \int_\Lambda
        \sum_{c=1}^{d_C}
        p_A(a,c\mid M^A,\lambda)
        p_B(b\mid M^B,c,\lambda)
        \,\dd\mu(\lambda).
        \label{eq:bell-classical-simulation}
    \end{equation}
\end{enumerate}
\end{definition}

As in \cref{def:classical-simulation}, no restriction beyond measurability
is imposed on the shared randomness space. The only bounded resource is the
classical message sent from Alice to Bob. In particular, $d_C=1$ corresponds
to a Bell-local model for the state.

%% file: 3_appendix_d4proof.tex
\subsection{Statement of the result}
\label{sec:d4-problem-statement}

Recall that $C_{\mathrm{cl}}(d_Q)$ denotes the minimum classical message dimension required for an exact simulation of $d_Q$-dimensional quantum communication,
\begin{equation}
C_{\mathrm{cl}}(d_Q) = \inf\left\{d_C\in\N: \text{an exact classical simulation with message dimension }d_C \text{ exists} \right\},
\end{equation}
with $C_{\mathrm{cl}}(d_Q)=\infty$ if no finite-dimensional classical simulation exists.
Here we establish that no finite classical message dimension is sufficient starting from $d_Q=4$.

\begin{restatable}[No finite classical simulation for $d_Q\geq4$]{theorem}{mainthm}
\label{thm:d4-infinite}
The exact classical communication cost satisfies
\begin{equation}
    C_{\mathrm{cl}}(4)=\infty.    
\end{equation}
Consequently,
\begin{equation}
    C_{\mathrm{cl}}(d_Q)=\infty
    \qquad
    \text{for every } d_Q\geq4.    
\end{equation}
\end{restatable}

\subsection{Reduction to pure states and binary measurements}
\label{sec:ReductionPureStatesAndMeas_ququart}

To prove \cref{thm:d4-infinite}, it is enough to rule out a finite classical simulation for a suitable restricted family of ququart preparations and measurements. Indeed, any protocol that exactly simulates all ququart states and POVMs must, in particular, reproduce the statistics of every restricted family of such preparations and measurements.

We therefore restrict Alice's preparations to pure ququart states and Bob's measurements to binary projective measurements with a rank-one effect. We first fix a convenient representation of these objects.

Pure ququart states are commonly represented by unit vectors $\ket{x}\in\C^4$. However, vectors that differ only by a global phase represent the same physical state. The natural space parametrizing pure ququart states is therefore the complex projective space $\CP^3$, the space of one-dimensional complex subspaces, or rays, of $\C^4$. In particular, if $\ket{x}$ and $\ket{x'}$ are unit-vector representatives of the same ray, then $\ket{x'}=e^{i\theta}\ket{x}$ for some $\theta\in\R$.

There is an equivalent representation of pure states by rank-one orthogonal projectors. For each ray $x\in\CP^3$, define 
\begin{equation}
    P_x:=\ketbra{x},
\end{equation}
where $\ket{x}$ is any unit-vector representative of $x$. This is independent of the chosen representative. Hence, $\CP^3$ can be identified with the set of rank-one orthogonal projectors on $\C^4$, and we will use $x\in\CP^3$ and its associated projector $P_x$ interchangeably when no confusion can arise.

Likewise, every $y\in\CP^3$ determines the binary projective measurement
\begin{equation}
    M_y:=\{P_y,\eye_4-P_y\}.
\end{equation}
We label by $b=0$ the outcome associated with the rank-one effect $P_y$, and by $b=1$ the outcome associated with $(\eye_4-P_y)$.
Thus, throughout the remainder of the proof, Alice's preparations and Bob's measurements will be parametrized by rays $x,y\in\CP^3$, respectively. In this restricted scenario, we write
\begin{equation}
    p_Q(b\mid x,y):=p_Q(b\mid P_x,M_y).
\end{equation}
The corresponding quantum behavior is
\begin{equation}\label{eq:q_behaviour_b0}
    p_Q(0\mid x,y) = \tr(P_xP_y) = |\langle x,y\rangle|^2,
\end{equation}
and
\begin{equation}\label{eq:q_behaviour_b1}
    p_Q(1\mid x,y) =  \tr(P_x(\eye_4-P_y)) = 1-|\langle x,y\rangle|^2.    
\end{equation}
Hence the entire behavior is determined by the overlap between the two rays. The proof will exploit this dependence by averaging over pairs $(x,y)$ with a fixed value of this overlap.

\subsection{Fixed-fidelity averages}
\subsubsection{Fixed-fidelity pairs}
\label{sec:fixed-fidelity-pairs}
The quantity appearing in \cref{eq:q_behaviour_b0,eq:q_behaviour_b1} is the fidelity between the rank-one projectors associated with the rays $x$ and $y$. We denote it by
\begin{equation}
    f(x,y) := \tr(P_xP_y)
    = \xy^2.    
\end{equation}
Although $y$ parametrizes Bob's measurement in the restricted scenario, the effect $P_y$ is itself a rank-one projector, so the same quantity is naturally the fidelity between $P_x$ and $P_y$. In terms of $f$,
\begin{equation}
    p_Q(0\mid x,y)=f(x,y),
    \qquad
    p_Q(1\mid x,y)=1-f(x,y).    
\end{equation}
The fidelity is well defined on rays and is therefore independent of the chosen unit-vector representatives.

For each $0\leq t<1$, we define a probability measure $\eta_t$ on pairs of rays with fixed fidelity $f(x,y) = t$. Fix two orthonormal vectors $\ket{0}, \ket{1}\in\C^4$ and sample a Haar-random unitary $W\in U(4)$. Define
\begin{equation}
    \mathsf X(W):=[W\ket{0}],
\qquad
\mathsf Y_t(W)
:=
\left[
W\left(\sqrt{t}\ket{0}+\sqrt{1-t}\ket{1}\right)
\right],    
\end{equation}
where $[\ket{v}]$ denotes the ray spanned by the nonzero vector $\ket{v}$. We denote the distribution of the pair $(\mathsf X,\mathsf Y_t)$ by $\eta_t$. More precisely, for every measurable set $A\subseteq\CP^3\times\CP^3$,
\begin{equation}
    \eta_t(A)
=
\mu_{\mathrm{Haar}}
\left(
\left\{
W\in U(4):
\bigl(\mathsf X(W),\mathsf Y_t(W)\bigr)\in A
\right\}
\right).    
\end{equation}
Since unitary transformations preserve inner products,
\begin{equation}
    f\bigl(\mathsf X(W),\mathsf Y_t(W)\bigr)
=
t
\qquad
\text{for every }W\in U(4).    
\end{equation}
Thus, $\eta_t$ is supported on pairs of rays whose fidelity is exactly $t$.

We denote by $\sigma$ the unitarily invariant probability measure on $\CP^3$. Equivalently, a pair distributed according to $\eta_t$ may be obtained by first sampling $\mathsf X=x$ according to $\sigma$ and then sampling $\mathsf Y_t$ uniformly among all rays $y$ satisfying
\begin{equation}
    f(x,y)=t.    
\end{equation}
For $t=0$, the second ray is therefore sampled uniformly among the rays orthogonal to $x$, namely from the projective space of $x^\perp$.

Both marginals of $\eta_t$ are equal to $\sigma$. Hence $\mathsf X$ and $\mathsf Y_t$, considered separately, are uniformly distributed on $\CP^3$.

\subsubsection{Fixed-fidelity averaging}
\label{sec:fixed-fidelity-averaging}
Before returning to the classical simulation problem, we introduce a convenient notation for averaging arbitrary functions over pairs of rays with fixed fidelity. This allows us to isolate the main technical ingredient of the proof and state it in a simple and general form.

Let $A,B:\CP^3\longrightarrow[0,1]$ be measurable functions. For each $0\leq t<1$, define their \emph{fixed-fidelity average} by
\begin{equation}
a_{A,B}(t)
:=
\int_{\CP^3\times\CP^3}
A(x)B(y)
\dd\eta_t(x,y).
\label{eq:fixed-fidelity-functional}
\end{equation}
Thus, $a_{A,B}(t)$ is the average value of $A(x)B(y)$ over pairs $(x,y)$ whose fidelity is exactly $t$.

The relevance of this construction is captured by the following lemma. It states that if a fixed-fidelity average vanishes at $t=0$, then it cannot initially grow linearly with the fidelity.

\begin{restatable}{lemma}{zeroslope}
\label{lem:zero-slope}
Let $A,B:\CP^3\to[0,1]$ be measurable functions, and let $a_{A,B}(t)$ be defined by \cref{eq:fixed-fidelity-functional}. If $a_{A,B}(0)=0,$ then
\begin{equation}
\lim_{t\to0^+}\frac{a_{A,B}(t)}{t}
= 0 
\label{eq:zero-slope}
\end{equation}
and, for $0<t\leq\frac12$,
\begin{equation}
0
\leq
\frac{a_{A,B}(t)}{t}
\leq
8.
\label{eq:uniform-slope}
\end{equation}
\end{restatable}

The proof of \cref{lem:zero-slope} is postponed to \cref{sec:zero-slope-proof}. We now apply this result to the classical simulation problem.

\subsection{Proof of the main theorem}
\label{sec:d4-proof-of-main-theorem}

\mainthm*

\begin{proof}
Assume, toward a contradiction, that ququart communication admits an exact classical simulation with finite classical message dimension $d_C$ (see \cref{def:classical-simulation}). In particular, the same model must exactly reproduce the statistics of the restricted scenario consisting of pure preparations $P_x$ and binary measurements $M_y$, with $x,y\in\CP^3$. Therefore, for every $x,y\in\CP^3$, and $b\in\{0,1\}$,
\begin{equation}
p_Q(b\mid x,y)
=\int_\Lambda
\sum_{c=1}^{d_C}
p_A(c\mid x,\lambda)
p_B(b\mid y,c,\lambda)
\,\dd\mu(\lambda).
\label{eq:restricted-classical-simulation}
\end{equation}
We will focus on the outcome $b=0$. For each $\lambda\in\Lambda$ and $c\in[d_C]$, define
\begin{align}
A_{\lambda,c}(x)
&:=
p_A(c\mid x,\lambda),
\label{eq:branch-encoder}
\\
B_{\lambda,c}(y)
&:=
p_B(0\mid y,c,\lambda).
\label{eq:branch-decoder}
\end{align}
For fixed $(\lambda,c)$, the functions $A_{\lambda,c}$ and $B_{\lambda,c}$ are measurable and take values in $[0,1]$, by \cref{def:classical-simulation}. 

Taking $b=0$ in \cref{eq:restricted-classical-simulation} and using $p_Q(0\mid x,y)=f(x,y)$ gives, for every $x,y\in\CP^3$,
\begin{equation}
f(x,y)
=
\sum_{c=1}^{d_C}
\int_\Lambda
A_{\lambda,c}(x)
B_{\lambda,c}(y)
\,\dd\mu(\lambda).
\label{eq:factor-average}
\end{equation}
We now average this identity over pairs $(x,y)$ with fixed fidelity $t$, using the probability measure $\eta_t$. By construction, $\eta_t$ is supported on pairs satisfying $f(x,y)=t$. Therefore,
\begin{equation}
\int_{\CP^3\times\CP^3}
f(x,y)
\,\dd\eta_t(x,y)
=
\int_{\CP^3\times\CP^3}
t
\,\dd\eta_t(x,y)
=
t,
\label{eq:quantum-fixed-fidelity-average}
\end{equation}
where the last equality follows because $\eta_t$ is a probability measure.

On the classical side, the integrand is measurable and nonnegative. Fubini-Tonelli's theorem therefore allows us to exchange the order of integration in the average of \cref{eq:factor-average}, yielding
\begin{equation}
t
=
\sum_{c=1}^{d_C}
\int_\Lambda
\left[
\int_{\CP^3\times\CP^3}
A_{\lambda,c}(x)
B_{\lambda,c}(y)
\,\dd\eta_t(x,y)
\right]
\,\dd\mu(\lambda).
\label{eq:averaged-factorization}
\end{equation}
The inner integral is precisely the fixed-fidelity average introduced in \cref{eq:fixed-fidelity-functional}. Accordingly, for each $(\lambda,c)$, we write
\begin{equation}
a_{\lambda,c}(t)
:=
a_{A_{\lambda,c},B_{\lambda,c}}(t)
=
\int_{\CP^3\times\CP^3}
A_{\lambda,c}(x)
B_{\lambda,c}(y)
\,\dd\eta_t(x,y).
\label{eq:branch-curve}
\end{equation}
Therefore, \cref{eq:averaged-factorization} takes the form
\begin{equation}
t
=
\sum_{c=1}^{d_C}
\int_\Lambda
a_{\lambda,c}(t)
\,\dd\mu(\lambda).
\label{eq:averaged-curve}
\end{equation}

Evaluating \cref{eq:averaged-curve} at $t=0$ gives
\begin{equation}
0
=
\sum_{c=1}^{d_C}
\int_\Lambda
a_{\lambda,c}(0)
\,\dd\mu(\lambda).
\label{eq:average-at-zero}
\end{equation}
Since $a_{\lambda,c}(0)\geq0$ for every $\lambda$ and $c$, each integral in the sum is nonnegative. Hence, for every $c\in[d_C]$, $a_{\lambda,c}(0)=0$
for $\mu$-almost every $\lambda$.

For each $c\in[d_C]$, define
\begin{equation}
\Gamma_c
:=
\left\{
\lambda\in\Lambda:
a_{\lambda,c}(0)>0
\right\}.
\label{eq:almost-every-branch-excludes}
\end{equation}
Then $\mu(\Gamma_c)=0$. Since there are only finitely many messages
\begin{equation}
\Gamma
:=
\bigcup_{c\in[d_C]}
\Gamma_c
\end{equation}
is also $\mu$-null. Therefore, $a_{\lambda,c}(0)=0,$
or every $\lambda\in\Lambda\setminus\Gamma$ and every $c\in[d_C]$. Hence, applying \cref{lem:zero-slope} to $a_{\lambda,c}(t)$ gives
\begin{equation}
    \frac{a_{\lambda,c}(t)}{t}
\longrightarrow
0
\qquad
\text{as }t\to0^+,    
\end{equation}
and, for $0<t\leq\frac12$,
\begin{equation}
    0
\leq
\frac{a_{\lambda,c}(t)}{t}
\leq
8.    
\end{equation}

We now return to the averaged identity \cref{eq:averaged-curve}, dividing it by $t$ gives
\begin{align}
1
&=
\sum_{c=1}^{d_C}
\int_\Lambda
\frac{a_{\lambda,c}(t)}{t}
\,\dd\mu(\lambda)
=
\sum_{c=1}^{d_C}
\int_{\Lambda\setminus\Gamma}
\frac{a_{\lambda,c}(t)}{t}
\,\dd\mu(\lambda),
\label{eq:final-average}
\end{align}
where the second equality follows from $\mu(\Gamma)=0$.

Since $\mu$ is a probability measure, the constant function $8$ is integrable. The dominated convergence theorem therefore gives, for each $c\in[d_C]$,
\begin{equation}
\lim_{t\to0^+}
\int_{\Lambda\setminus\Gamma}
\frac{a_{\lambda,c}(t)}{t}
,\dd\mu(\lambda)
=
0.
\end{equation}
Because the number of messages is finite, we may pass the limit through the sum over $c$, obtaining
\begin{equation}
\lim_{t\to0^+}
\sum_{c=1}^{d_C}
\int_{\Lambda\setminus\Gamma}
\frac{a_{\lambda,c}(t)}{t}
,\dd\mu(\lambda)
=
0.
\end{equation}
This contradicts \cref{eq:final-average}, since the expression inside the limit is equal to $1$ for every $t>0$.

Therefore,
\begin{equation}
    C_{\mathrm{cl}}(4)=\infty.    
\end{equation}
The conclusion for every $d_Q\geq4$ then follows directly from \cref{rem:bigger_classical_and_quantum_dimensions}.
\end{proof}

The reason why dimension $d = 4$ appears lies in the averaging mechanism behind \cref{lem:zero-slope}.
Starting at $d=4$, averaging over all tests orthogonal to a fixed state smooths the classical response sufficiently to force the zero slope behavior.
In smaller dimensions, this averaging is not strong enough (see \cref{sec:dimension-threshold} for details).
Therefore, this argument does not contradict the fact that the simulation cost for $d = 2$ is finite \cite{TonerBacon,RennerTavakoliQuintino},
and it is not sufficient to determine whether the classical communication cost for $d=3$ is finite.

\subsection{Proof of the zero slope lemma}
\label{sec:zero-slope-proof}

We now turn to the proof of \cref{lem:zero-slope}. The main point is to
reformulate the correlation function appearing in the statement of the
lemma in terms of a family of averaging operators on complex projective
space. This formulation makes the $U(4)$-symmetry of the problem explicit
and allows us to exploit the spectral decomposition of the
Laplace--Beltrami operator.

Throughout this section, let
\begin{equation}
  X:=\CP^3
\end{equation}
denote complex projective three-space equipped with the Fubini--Study
Riemannian metric. If $[x]\in X$ and $x\in\mathbb C^4$ is a unit
representative, we identify the tangent space $T_{[x]}X$ with the
orthogonal complement
\begin{equation}
  x^\perp
  :=
  \{v\in\mathbb C^4:\langle x|v\rangle=0\}.
\end{equation}
With our normalization, the Fubini--Study metric is given by
\begin{equation}
  g_{[x]}(v,w):=g_{\mathrm{FS},[x]}(v,w)
  :=
  \operatorname{Re}\langle v|w\rangle,
  \qquad
  v,w\in x^\perp.
\end{equation}

We denote by $\sigma$ the normalized
$U(4)$-invariant Riemannian measure on $X$, so that
\begin{equation}
  \sigma(X)=1.
\end{equation}
The nonpositive Laplace--Beltrami operator associated with this metric will
be denoted by
\begin{equation}
  \Delta:=\Delta_{\CP^3}
\end{equation}

We work in the Hilbert space $L^2(X,\sigma)$, equipped with the inner
product
\begin{equation}
  \ip{F}{H}
  :=
  \int_X F(x)\overline{H(x)}\,\dd\sigma(x),
\end{equation}
and the corresponding norm
\begin{equation}
  \norm{F}_2:=\ip{F}{F}^{1/2}.
\end{equation}

For each $0\le t<1$, let $\eta_t$ be the joint probability measure on
$X\times X$ introduced above, and let
$\eta_t(\dd y\mid x)$ denote its conditional distribution in the second
variable given the first. Associated with $\eta_t$ is the averaging
operator
\begin{equation}
  (R_tB)(x)
  :=
  \int_X B(y)\,\eta_t(\dd y\mid x),
  \qquad
  B\in L^2(X).
  \label{eq:Rt-definition}
\end{equation}
Equivalently,
\begin{equation}
  (R_tB)(x)
  =
  \mathbb E\!\left[B(\mathsf{Y}_t)\mid \mathsf{X}=x\right].
\end{equation}
Thus $R_tB$ is obtained by averaging $B$ over all points $y$ having the
prescribed fidelity parameter $t$ relative to $x$. In particular, $R_t$
is a Markov averaging operator.

For measurable functions $A,B:X\to[0,1]$, define
\begin{equation}
  a_{A,B}(t)
  :=
  \int_{X\times X}
  A(x)B(y)\, \dd\eta_t(x, y),
  \qquad 0\le t<1.
  \label{eq:aAB-definition}
\end{equation}
Since $A$ and $B$ will remain fixed throughout the argument, we write
simply
\begin{equation}
  a(t):=a_{A,B}(t).
\end{equation}
Using the disintegration of $\eta_t$ with respect to its first marginal,
which is $\sigma$, Tonelli's theorem gives
\begin{align}
  a(t)
  &=
  \int_X
  A(x)
  \left(
    \int_X B(y)\,\eta_t(\dd y\mid x)
  \right)
  \dd\sigma(x)=
  \ip{A}{R_tB}.
  \label{eq:B38}
\end{align}
Hence the scalar quantity $a(t)$ is completely encoded by the action of
the operator $R_t$ on $B$.

We first record the basic $L^2$ bound for $R_t$. By Jensen's inequality,
for $\sigma$-almost every $x\in X$,
\begin{equation}
  \abs{R_tB(x)}^2
  \le
  \int_X
  \abs{B(y)}^2\,\eta_t(\dd y\mid x).
\end{equation}
Integrating with respect to $\sigma$ and using the fact that the second
marginal of $\eta_t$ is again $\sigma$, we obtain
\begin{align}
  \norm{R_tB}_2^2
  &\le
  \int_X
  \int_X
  \abs{B(y)}^2
  \,\eta_t(\dd y\mid x)\,
  \dd\sigma(x)=
  \int_X
  \abs{B(y)}^2\,\dd\sigma(y)
  =
  \norm{B}_2^2.
  \label{eq:Rt-contraction-squared}
\end{align}
Therefore
\begin{equation}
  \norm{R_tB}_2
  \le
  \norm{B}_2,
  \qquad
  0\le t<1,
  \label{eq:Rt-contraction}
\end{equation}
so each $R_t$ is a contraction on $L^2(X)$. The significance of this reformulation becomes apparent at the
perfect-exclusion endpoint $t=0$. Under the hypothesis
\begin{equation}
  a(0)=0,
\end{equation}
we have, by \cref{eq:B38},
\begin{equation}
  \frac{a(t)}{t}
  =
  \frac{a(t)-a(0)}{t}
  =
  \ip{A}{
    \frac{R_t-R_0}{t}B
  },
  \qquad t>0.
\end{equation}
Thus the zero-slope assertion in \cref{lem:zero-slope} is reduced to
understanding the strong $L^2$ behavior of the operator difference
quotient
\begin{equation}
  \frac{R_t-R_0}{t}
  \qquad\text{as }t\downarrow0.
\end{equation}

The remainder of the proof is devoted to this question: The
$U(4)$-symmetry of $X$ and of the measures $\eta_t$ implies that $R_t$
commutes with the natural $U(4)$-action on $L^2(X)$. We shall therefore
diagonalize $R_t$ with respect to the irreducible $U(4)$-decomposition of
$L^2(X)$. The corresponding multipliers are expressed in terms of
Jacobi polynomials. Their behavior at $t=0$ will show, first mode by mode
and then after recombination of all modes, that
\begin{equation}
  \frac{R_t-R_0}{t}B
  \longrightarrow
  \frac14\Delta R_0B
  \qquad\text{in }L^2(X).
\end{equation}
The final step will be to use the positivity assumptions
$A,B\ge0$ together with $a(0)=0$ to prove that
\begin{equation}
  \ip{A}{\Delta R_0B}=0,
\end{equation}
which forces the first-order slope of $a(t)$ at the perfect-exclusion
endpoint to vanish.

\subsubsection{Spectral diagonalization of the averaging operators}

The spectral theory of the Fubini--Study Laplacian gives the orthogonal
Hilbert decomposition
\begin{equation}
  L^2(X)
  =
  \widehat{\bigoplus}_{\ell=0}^{\infty}\cH_\ell.
  \label{eq:B48}
\end{equation}
Each $\cH_\ell$ is finite-dimensional,
\begin{equation}
  \dim\cH_\ell=N_\ell<\infty,
\end{equation}
and is an eigenspace of the Laplace--Beltrami operator:
\begin{equation}
  \Delta F
  =
  -4\ell(\ell+3)F,
  \qquad
  F\in\cH_\ell.
  \label{eq:B49}
\end{equation}
This eigenvalue normalization agrees with the metric fixed above; see
Lu~\cite{Lu1998}.

The natural representation of $U(4)$ on each $\cH_\ell$ is irreducible,
and the decomposition in \cref{eq:B48} is multiplicity free. Moreover, the
corresponding zonal kernels are expressed in terms of Jacobi polynomials. See \cite{Koornwinder1973}.

The $U(4)$-equivariance of $\eta_t$ implies that $R_t$ commutes with the
natural action of $U(4)$. Since the decomposition in \cref{eq:B48} is
multiplicity free, each space $\cH_\ell$ is preserved by $R_t$.
Schur's lemma then shows that $R_t$ acts on every angular frequency space
by multiplication by a scalar. The following proposition identifies that
scalar explicitly.

\begin{proposition}[Multiplier formula]
\label{prop:multiplier}
For every $F\in\cH_\ell$ and $0\le t<1$,
\begin{equation}
  R_tF=\phi_\ell(t)F,
  \label{eq:B50}
\end{equation}
where
\begin{equation}
  \phi_\ell(t)
  =
  \frac{P_\ell^{(2,0)}(2t-1)}
       {P_\ell^{(2,0)}(1)}.
  \label{eq:B51}
\end{equation}
Here $P_\ell^{(\alpha,\beta)}$ denotes the Jacobi polynomial in the
standard normalization.
\end{proposition}

\begin{proof}
Fix $t\in[0,1)$. Let $U(4)$ act on functions on $X$ by
\begin{equation}
  (g\cdot F)(x):=F(g^{-1}x),
  \qquad g\in U(4).
\end{equation}
The $U(4)$-equivariance of the conditional measures
$\eta_t(\dd y\mid x)$ implies
\begin{equation}
  R_t(g\cdot F)
  =
  g\cdot(R_tF).
\end{equation}
Thus $R_t$ is an intertwining operator for the $U(4)$-representation.

Since \cref{eq:B48} is multiplicity free, $R_t$ preserves each
$\cH_\ell$. The representation on $\cH_\ell$ is irreducible, so Schur's
lemma gives
\begin{equation}
  R_tF=\phi_\ell(t)F,
  \qquad F\in\cH_\ell,
\end{equation}
for some scalar $\phi_\ell(t)$. It remains to identify this scalar.

Let $\{Y_{\ell,j}\}_{j=1}^{N_\ell}$ be an orthonormal basis of
$\cH_\ell$, and define its reproducing kernel by
\begin{equation}
  K_\ell(y,x)
  :=
  \sum_{j=1}^{N_\ell}
  Y_{\ell,j}(y)\overline{Y_{\ell,j}(x)}.
  \label{eq:B52}
\end{equation}
For each fixed $x\in X$, the function
\begin{equation}
  y\longmapsto K_\ell(y,x)
\end{equation}
belongs to $\cH_\ell$. Applying the multiplier identity to this function
and evaluating at $x$ gives
\begin{equation}
  \phi_\ell(t)K_\ell(x,x)
  =
  \int_XK_\ell(y,x)\,\eta_t(\dd y\mid x).
  \label{eq:B54a}
\end{equation}

The addition formula on complex projective space gives
\begin{equation}
  K_\ell(y,x)
  =
  c_\ell
  P_\ell^{(2,0)}
  \!\left(
    2\abs{\langle x,y\rangle}^2-1
  \right),
  \qquad
  c_\ell>0.
  \label{eq:B53}
\end{equation}
See \cite[Eq.~(3.2)]{Koornwinder1973}. In the notation of
that reference, $X=\CP^{q-1}$ with $q=4$, which yields the Jacobi
parameters
\begin{equation}
  (q-2,0)=(2,0).
\end{equation}

By construction, $\eta_t(\,\cdot\,\mid x)$ is supported on
\begin{equation}
  \left\{
    y\in X:
    \abs{\langle x,y\rangle}^2=t
  \right\}.
\end{equation}
Hence the right-hand side of \cref{eq:B54a} is constant on the support
of the conditional measure. Since $\eta_t(\,\cdot\,\mid x)$ has total
mass one,
\begin{equation}
  \int_XK_\ell(y,x)\,\eta_t(\dd y\mid x)
  =
  c_\ell P_\ell^{(2,0)}(2t-1).
  \label{eq:B54}
\end{equation}
On the other hand, setting $y=x$ in \cref{eq:B53} gives
\begin{equation}
  K_\ell(x,x)
  =
  c_\ell P_\ell^{(2,0)}(1).
  \label{eq:B55}
\end{equation}
Since $P_\ell^{(2,0)}(1)\neq0$, division yields \cref{eq:B51}.
\end{proof}

\subsubsection{The exclusion multiplier and the smoothing effect of $R_0$}

We next examine the multiplier at the perfect-exclusion endpoint $t=0$.
The endpoint values of Jacobi polynomials,
\begin{equation}
  P_\ell^{(\alpha,\beta)}(1)
  =
  \binom{\ell+\alpha}{\ell},
  \qquad
  P_\ell^{(\alpha,\beta)}(-1)
  =
  (-1)^\ell
  \binom{\ell+\beta}{\ell},
  \label{eq:B56}
\end{equation}
give
\begin{equation}
  r_\ell
  :=
  \phi_\ell(0)
  =
  \frac{2(-1)^\ell}
       {(\ell+1)(\ell+2)}.
  \label{eq:B57}
\end{equation}
In particular,
\begin{equation}
  \abs{r_\ell}
  =
  \frac{2}{(\ell+1)(\ell+2)}
  =
  O(\ell^{-2}).
\end{equation}
Thus the perfect-exclusion operator $R_0$ suppresses the $\ell$th angular
mode at a quadratic rate.

Let
\begin{equation}
  B
  =
  \sum_{\ell=0}^{\infty}B_\ell,
  \qquad
  B_\ell\in\cH_\ell,
  \qquad
  \sum_{\ell=0}^{\infty}\norm{B_\ell}_2^2
  =
  \norm{B}_2^2
  \label{eq:B58}
\end{equation}
be the orthogonal spectral decomposition of $B$. By the multiplier
formula,
\begin{equation}
  R_0B
  =
  \sum_{\ell=0}^{\infty}r_\ell B_\ell
  \qquad\text{in }L^2(X).
  \label{eq:B59}
\end{equation}

The quadratic decay of $r_\ell$ precisely compensates for the quadratic
growth of the Laplace--Beltrami eigenvalues. Indeed,
\begin{equation}
  4\ell(\ell+3)\abs{r_\ell}
  =
  \frac{8\ell(\ell+3)}
       {(\ell+1)(\ell+2)}
  <8.
  \label{eq:B60}
\end{equation}
Therefore the spectral multipliers
\begin{equation}
  -4\ell(\ell+3)r_\ell
\end{equation}
are uniformly bounded, and the series
\begin{equation}
  \Delta R_0B
  :=
  \sum_{\ell=0}^{\infty}
  -4\ell(\ell+3)r_\ell B_\ell
  \label{eq:B61}
\end{equation}
converges in $L^2(X)$. By Parseval's identity,
\begin{align}
  \norm{\Delta R_0B}_2^2
  &=
  \sum_{\ell=0}^{\infty}
  \bigl(
    4\ell(\ell+3)r_\ell
  \bigr)^2
  \norm{B_\ell}_2^2
\le
  64
  \sum_{\ell=0}^{\infty}\norm{B_\ell}_2^2
  =
  64\norm{B}_2^2.
  \label{eq:B62}
\end{align}
Hence
\begin{equation}
  R_0B\in\mathcal D_{L^2}(\Delta),
  \qquad
  \norm{\Delta R_0B}_2
  \le
  8\norm{B}_2.
\end{equation}
In this spectral sense, $R_0$ exhibits a two-derivative smoothing effect.

\subsubsection*{Differentiation of the multipliers at $t=0$}

We now identify the first-order behavior of each spectral multiplier.
Recall that
\begin{equation}
  y(s)=P_\ell^{(\alpha,\beta)}(s)
\end{equation}
satisfies the Jacobi differential equation
\begin{equation}
  (1-s^2)y''(s)
  +
  \bigl(
    \beta-\alpha-(\alpha+\beta+2)s
  \bigr)y'(s)
  +
  \ell(\ell+\alpha+\beta+1)y(s)
  =
  0.
\end{equation}
Specializing to $(\alpha,\beta)=(2,0)$ and setting
\begin{equation}
  s=2t-1,
\end{equation}
we have
\begin{equation}
  \frac{\dd}{\dd s}
  =
  \frac12\frac{\dd}{\dd t},
  \qquad
  \frac{\dd^2}{\dd s^2}
  =
  \frac14\frac{\dd^2}{\dd t^2},
\end{equation}
as well as
\begin{equation}
  1-s^2=4t(1-t),
  \qquad
  -2-4s=2(1-4t).
\end{equation}
Consequently, the normalized multiplier
\begin{equation}
  \phi_\ell(t)
  =
  \frac{P_\ell^{(2,0)}(2t-1)}
       {P_\ell^{(2,0)}(1)}
\end{equation}
satisfies
\begin{equation}
  t(1-t)\phi_\ell''(t)
  +
  (1-4t)\phi_\ell'(t)
  +
  \ell(\ell+3)\phi_\ell(t)
  =
  0.
  \label{eq:B63}
\end{equation}

Since $\phi_\ell$ is a polynomial, \cref{eq:B63} extends regularly to
$t=0$. Evaluating there and using $\phi_\ell(0)=r_\ell$, we obtain
\begin{equation}
  \phi_\ell'(0)
  =
  -\ell(\ell+3)r_\ell.
  \label{eq:B64}
\end{equation}
Therefore, for every $F\in\cH_\ell$,
\begin{equation}
  \left.
  \frac{\dd}{\dd t}R_tF
  \right|_{t=0^+}
  =
  -\ell(\ell+3)r_\ell F.
\end{equation}
Since
\begin{equation}
  \Delta F=-4\ell(\ell+3)F
  \qquad\text{and}\qquad
  R_0F=r_\ell F,
\end{equation}
we conclude that
\begin{equation}
  \left.
  \frac{\dd}{\dd t}R_tF
  \right|_{t=0^+}
  =
  \frac14\Delta R_0F.
  \label{eq:B65}
\end{equation}
Thus the desired differentiation formula is already valid mode by mode.

\subsubsection{A frequency-uniform difference-quotient estimate}
\label{sec:mode-by-mode-to-whole-function}

To pass from \cref{eq:B65} to a general function
$B\in L^2(X)$, pointwise convergence for each fixed $\ell$ is not enough:
the spectral expansion of $B$ may contain infinitely many modes. We need a
bound on the difference quotients that is uniform in $\ell$. To prove this, define the auxiliary polynomial
\begin{equation}
  \psi_\ell(t)
  :=
  \phi_\ell(t)
  -
  \frac{1-t}{2}\phi_\ell'(t).
  \label{eq:B67}
\end{equation}
Its particular form is chosen so that it satisfies another Jacobi
equation with uniformly bounded normalized solutions.

Introduce
\begin{equation}
  E_\ell[\phi]
  :=
  t(1-t)\phi''
  +(1-4t)\phi'
  +\ell(\ell+3)\phi.
\end{equation}
By \cref{eq:B63},
\begin{equation}
  E_\ell[\phi_\ell]=0.
\end{equation}
A direct differentiation of \cref{eq:B67} and collection of terms gives
\begin{equation}
\begin{aligned}
  &t(1-t)\psi_\ell''(t)
  +(2-4t)\psi_\ell'(t)
  +\ell(\ell+3)\psi_\ell(t)
  \\
  &\qquad
  =
  E_\ell[\phi_\ell](t)
  -
  \frac{1-t}{2}
  \frac{\dd}{\dd t}E_\ell[\phi_\ell](t)
  =
  0.
  \label{eq:B68}
\end{aligned}
\end{equation}
After the change of variables $s=2t-1$, this is precisely the Jacobi
equation with parameters $(\alpha,\beta)=(1,1)$.

Since \cref{eq:B51} is a polynomial identity, it extends to $t=1$ and
gives
\begin{equation}
  \phi_\ell(1)=1.
\end{equation}
The factor $1-t$ in \cref{eq:B67} then yields
\begin{equation}
  \psi_\ell(1)=1.
\end{equation}
By uniqueness of the normalized polynomial solution,
\begin{equation}
  \psi_\ell(t)
  =
  \frac{P_\ell^{(1,1)}(2t-1)}
       {P_\ell^{(1,1)}(1)}.
  \label{eq:B69}
\end{equation}

We now use the standard Jacobi estimate
\begin{equation}
  \abs{P_\ell^{(\alpha,\beta)}(s)}
  \le
  P_\ell^{(\alpha,\beta)}(1),
  \qquad
  -1\le s\le1,
  \quad
  \alpha\ge\beta>-1,
  \quad
  \alpha\ge-\frac12,
  \label{eq:B70}
\end{equation}
see \cite[Eq.~18.14.1]{DLMF}. Applying this inequality to
$(\alpha,\beta)=(2,0)$ and $(1,1)$ gives
\begin{equation}
  \abs{\phi_\ell(t)}\le1,
  \qquad
  \abs{\psi_\ell(t)}\le1,
  \qquad
  0\le t\le1,
\end{equation}
uniformly in $\ell$.

Solving \cref{eq:B67} for $\phi_\ell'$ therefore gives, for
$0\le t\le\frac12$,
\begin{align}
  \abs{\phi_\ell'(t)}
  &=
  \frac{2}{1-t}
  \abs{\phi_\ell(t)-\psi_\ell(t)}\le
  \frac{
    2\bigl(
      \abs{\phi_\ell(t)}
      +
      \abs{\psi_\ell(t)}
    \bigr)
  }{1-t}
  \le
  \frac4{1-t}
  \le8.
  \label{eq:B71}
\end{align}
The mean-value theorem now yields
\begin{equation}
  \abs{
    \frac{\phi_\ell(t)-r_\ell}{t}
  }
  \le8,
  \qquad
  0<t\le\frac12,
  \label{eq:B72}
\end{equation}
uniformly in $\ell$.

\begin{remark}
The restriction $0\le t\le \frac12$ is only a convenient choice and is
not essential. More generally, for any fixed $\varepsilon\in(0,1)$, the
same argument applies on the interval
\begin{equation}
  0\le t\le 1-\varepsilon.
\end{equation}
Indeed, from \cref{eq:B67} and the bounds
\begin{equation}
  \abs{\phi_\ell(t)}\le1,
  \qquad
  \abs{\psi_\ell(t)}\le1,
\end{equation}
we obtain
\begin{equation}
  \abs{\phi_\ell'(t)}
  \le
  \frac{4}{1-t}
  \le
  \frac{4}{\varepsilon},
  \qquad
  0\le t\le1-\varepsilon,
\end{equation}
uniformly in $\ell$. Consequently, the mean-value theorem gives
\begin{equation}
  \abs{
    \frac{\phi_\ell(t)-r_\ell}{t}
  }
  \le
  \frac{4}{\varepsilon},
  \qquad
  0<t\le1-\varepsilon.
\end{equation}
Thus any compact subinterval of $[0,1)$ yields a uniform bound of the same
type. The choice $t\le\frac12$ corresponds to $\varepsilon=\frac12$ and
gives the convenient constant $8$. As the interval approaches the endpoint
$t=1$, however, the resulting constant deteriorates like
$(1-t)^{-1}$.
\end{remark}

We can now pass from the modewise identity \cref{eq:B65} to a strong
$L^2$ statement. Define
\begin{equation}
  d_\ell(t)
  :=
  \frac{\phi_\ell(t)-r_\ell}{t}.
  \label{eq:B73}
\end{equation}
For each fixed $\ell$,
\begin{equation}
  \lim_{t\downarrow0}d_\ell(t)
  =
  \phi_\ell'(0).
  \label{eq:B74}
\end{equation}
Moreover, for
\begin{equation}
  B=\sum_{\ell=0}^{\infty}B_\ell,
  \qquad
  B_\ell\in\cH_\ell,
\end{equation}
we have
\begin{equation}
  \frac{R_t-R_0}{t}B
  =
  \sum_{\ell=0}^{\infty}d_\ell(t)B_\ell,
  \qquad
  \frac14\Delta R_0B
  =
  \sum_{\ell=0}^{\infty}\phi_\ell'(0)B_\ell.
  \label{eq:B75}
\end{equation}

Combining \cref{eq:B57,eq:B64}, we also have
\begin{equation}
  \abs{\phi_\ell'(0)}
  =
  \frac{
    2\ell(\ell+3)
  }{
    (\ell+1)(\ell+2)
  }
  <2.
  \label{eq:B76}
\end{equation}
Therefore, by \cref{eq:B72},
\begin{equation}
  \abs{d_\ell(t)-\phi_\ell'(0)}
  \le
  \abs{d_\ell(t)}
  +
  \abs{\phi_\ell'(0)}
  <10,
  \qquad
  0<t\le\frac12.
\end{equation}
It follows that
\begin{equation}
  \abs{
    d_\ell(t)-\phi_\ell'(0)
  }^2
  \norm{B_\ell}_2^2
  \le
  100\norm{B_\ell}_2^2,
  \qquad
  0<t\le\frac12.
  \label{eq:B77}
\end{equation}
The dominating sequence is summable:
\begin{equation}
  \sum_{\ell=0}^{\infty}
  100\norm{B_\ell}_2^2
  =
  100\norm{B}_2^2
  <\infty.
  \label{eq:B78}
\end{equation}

By the orthogonality of the spaces $\cH_\ell$,
\begin{equation}
  \norm{
    \frac{R_t-R_0}{t}B
    -
    \frac14\Delta R_0B
  }_2^2
  =
  \sum_{\ell=0}^{\infty}
  \abs{
    d_\ell(t)-\phi_\ell'(0)
  }^2
  \norm{B_\ell}_2^2.
  \label{eq:B79}
\end{equation}
For each fixed $\ell$, the summand tends to zero by \cref{eq:B74}, while
\cref{eq:B77,eq:B78} provide a summable dominating sequence.
Dominated convergence therefore gives
\begin{equation}
  \lim_{t\downarrow0}
  \sum_{\ell=0}^{\infty}
  \abs{
    d_\ell(t)-\phi_\ell'(0)
  }^2
  \norm{B_\ell}_2^2
  =
  0.
  \label{eq:B80}
\end{equation}
Consequently,
\begin{equation}
  \frac{R_t-R_0}{t}B
  \longrightarrow
  \frac14\Delta R_0B
  \qquad\text{in }L^2(X)
  \quad\text{as }t\downarrow0,
  \label{eq:B81}
\end{equation}
or equivalently,
\begin{equation}
  R_tB
  =
  R_0B
  +
  \frac{t}{4}\Delta R_0B
  +
  o_{L^2}(t).
\end{equation}

The same multiplier estimate gives the uniform bound
\begin{align}
  \norm{
    \frac{R_t-R_0}{t}B
  }_2^2
  &=
  \sum_{\ell=0}^{\infty}
  \abs{d_\ell(t)}^2
  \norm{B_\ell}_2^2
  \le
  64\sum_{\ell=0}^{\infty}\norm{B_\ell}_2^2
  =
  64\norm{B}_2^2.
\end{align}
Therefore
\begin{equation}
  \norm{
    \frac{R_t-R_0}{t}B
  }_2
  \le
  8\norm{B}_2,
  \qquad
  0<t\le\frac12.
  \label{eq:B82}
\end{equation}

\subsubsection{Perfect exclusion forces the first-order term to vanish}
\label{sec:perfect-exclusion-forces-zero}

It remains to exploit the positivity assumptions. Set
\begin{equation}
  h:=R_0B.
\end{equation}
Since $B\ge0$ and $R_0$ is a positive Markov operator,
\begin{equation}
  h\ge0
  \qquad\text{almost everywhere on }X.
\end{equation}
The perfect-exclusion hypothesis $a(0)=0$, together with
\cref{eq:B38}, gives
\begin{equation}
  0
  =
  a(0)
  =
  \ip{A}{h}.
  \label{eq:B83}
\end{equation}
Since $A,h\ge0$, it follows that
\begin{equation}
  A(x)h(x)=0
  \qquad
  \text{for $\sigma$-almost every }x\in X.
  \label{eq:B84}
\end{equation}

Let
\begin{equation}
  E:=\{x\in X:h(x)=0\}.
  \label{eq:B85}
\end{equation}
Then
\begin{equation}
  A=0
  \qquad\text{almost everywhere on }X\setminus E.
\end{equation}
To prove
\begin{equation}
  \ip{A}{\Delta h}=0,
\end{equation}
it therefore remains to show that
\begin{equation}
  \Delta h=0
  \qquad\text{almost everywhere on }E,
  \qquad
  E:=\{x\in X:h(x)=0\}.
\end{equation}
We shall use this standard Sobolev locality argument here but we include its proof in \cref{sec:laplacian-vanishes} for completeness.
We emphasize that this conclusion is not an immediate pointwise consequence of $h=0$. Its justification requires
some regularity of $h$ and a locality property of weak derivatives.

Assuming this fact for the moment, recall that the perfect-exclusion
condition gives
\begin{equation}
  A=0
  \qquad
  \text{almost everywhere on }X\setminus E.
\end{equation}
Hence we have the complementary vanishing properties
\begin{equation}
  A=0
  \quad\text{a.e. on }X\setminus E,
  \qquad
  \Delta h=0
  \quad\text{a.e. on }E.
\end{equation}
It follows that
\begin{equation}
  A\,\Delta h=0
  \qquad\text{almost everywhere on }X,
\end{equation}
and therefore
\begin{equation}
  \ip{A}{\Delta R_0B}
  =
  \ip{A}{\Delta h}
  =
  0.
  \label{eq:B90}
\end{equation}
\subsubsection*{Conclusion}

We can now return to the behavior of $a(t)$ at $t=0$. Since $a(0)=0$,
\begin{equation}
  \frac{a(t)}{t}
  =
  \ip{A}{
    \frac{R_t-R_0}{t}B
  }.
\end{equation}
By Cauchy--Schwarz,
\begin{equation}
\begin{aligned}
  \abs{
    \frac{a(t)}{t}
    -
    \frac14\ip{A}{\Delta R_0B}
  }
  &=
  \abs{
    \ip{A}{
      \frac{R_t-R_0}{t}B
      -
      \frac14\Delta R_0B
    }
  }
  \\
  &\le
  \norm{A}_2
  \norm{
    \frac{R_t-R_0}{t}B
    -
    \frac14\Delta R_0B
  }_2.
  \label{eq:B91}
\end{aligned}
\end{equation}
By the strong $L^2$ convergence in \cref{eq:B81}, the right-hand side tends
to zero as $t\downarrow0$. Combining this with \cref{eq:B90}, we obtain
\begin{equation}
  \frac{a(t)}{t}
  \longrightarrow
  \frac14\ip{A}{\Delta R_0B}
  =
  0
  \qquad
  \text{as }t\downarrow0.
  \label{eq:B92}
\end{equation}
This proves the zero-slope assertion \cref{eq:zero-slope}.

Finally, \cref{eq:B82} gives the required uniform bound. Since
$a(t)\ge0$ and $a(0)=0$, for $0<t\le\frac12$,
\begin{equation}
\begin{aligned}
  0
  \le
  \frac{a(t)}{t}
  &=
  \ip{A}{
    \frac{R_t-R_0}{t}B
  }
  \\
  &\le
  \norm{A}_2
  \norm{
    \frac{R_t-R_0}{t}B
  }_2
  \\
  &\le
  8\norm{A}_2\norm{B}_2\\
  &  \le8,
  \label{eq:B93}
\end{aligned}
\end{equation}
where the last inequality follows from $0\le A,B\le1$ and
$\sigma(X)=1$. This proves \cref{eq:uniform-slope} and completes the
proof of \cref{lem:zero-slope}.
\qed

\begin{remark}\label{rmk:zero-slope-sharpness}
The nonnegativity assumption in \cref{lem:zero-slope} is essential, and in this sense the hypothesis $A,B:\CP^3\to[0,1]$ is sharp. Indeed, the conclusion fails if the range $[0,1]$ is replaced by $[-1,1]$, for example.

To see this, let $F\in\mathcal H_1$ be a real-valued, nonzero function normalized so that
$$
    \|F\|_{L^2(\CP^3)}=1.
$$
Since $\CP^3$ is compact and $F$ is smooth, $F$ is bounded. Hence, for $\varepsilon>0$ sufficiently small, the functions

$$
    A=B
    :=
    \varepsilon
    \left(
        \frac{1}{\sqrt{3}}+F
    \right)
$$
take values in $[-1,1]$. Moreover, since $F\in\mathcal H_1$,
$$
    R_tF
    =
    \frac{4t-1}{3}\,F.
$$
Using the orthogonality of $\mathcal H_0$ and $\mathcal H_1$, we therefore obtain

$$
\begin{aligned}
    a_{A,B}(t)
    &=
    \langle A|R_tB\rangle_{L^2(\CP^3)}
 =
    \varepsilon^2
    \left(
        \frac13
        +
        \frac{4t-1}{3}\|F\|_{L^2(\CP^3)}^2
    \right)
=
    \frac{4\varepsilon^2}{3}\,t.
\end{aligned}
$$
Consequently,

$$
    a_{A,B}(0)=0,
$$

while

$$
    \lim_{t\to0^+}
    \frac{a_{A,B}(t)}{t}
    =
    \frac{4\varepsilon^2}{3}
    >0.
$$
Thus, allowing $A$ and $B$ to change sign destroys the zero-slope conclusion.
\end{remark}

\subsubsection{Why the Laplacian vanishes on the level set?}
\label{sec:laplacian-vanishes}

We now justify the fact used in the proof above. The statement we need is
\begin{equation}
  E=\{x\in X:h(x)=0\}
  \qquad\Longrightarrow\qquad
  \Delta h=0
  \quad\text{almost everywhere on }E.
  \label{eq:laplacian-zero-level-set}
\end{equation}
This deserves some care. The assertion is not the pointwise implication
\begin{equation}
  h(x)=0
  \quad\Longrightarrow\quad
  \Delta h(x)=0,
\end{equation}
which is false even for smooth functions. For example, $h(x)=x^2$
satisfies $h(0)=0$ but $h''(0)=2$. What is relevant here is instead an
almost-everywhere statement on the whole level set $\{h=0\}$.

The argument uses Sobolev regularity. Recall that, for an integer
$k\ge0$ and $1\le p\le\infty$, $W^{k,p}(X)$ denotes the standard Sobolev
space of order $k$ on $X$. In local coordinates, membership in
$W^{k,p}(X)$ means that all weak derivatives up to order $k$ belong
locally to $L^p$. In particular, these derivatives are not required to
exist in the classical pointwise sense.

From the smoothing estimate in \cref{eq:B62}, with $h=R_0B$, we know that
\begin{equation}
  h\in L^2(X),
  \qquad
  \Delta h\in L^2(X),
  \label{eq:B86}
\end{equation}
where $\Delta h$ is initially understood in the distributional,
equivalently spectral, sense. Since the Laplace--Beltrami operator is
elliptic of order two, elliptic regularity implies
\begin{equation}
  h\in H^2(X)=W^{2,2}(X).
  \label{eq:h-H2-regularity}
\end{equation}
See, for instance, \cite[Thm.~15.1]{Dyatlov2022}. Thus $h$ possesses first and second weak derivatives in $L^2$. We shall use the following standard locality property of Sobolev functions.

\begin{lemma}[Locality of weak derivatives on level sets]
\label{lem:locality}
Let $\Omega\subset\R^n$ be open, let $1\le p<\infty$, and let
$u\in W^{1,p}(\Omega)$. Then, for every $c\in\R$,
\begin{equation}
  \nabla u=0
  \qquad
  \text{almost everywhere on }\{u=c\}.
  \label{eq:first-order-locality}
\end{equation}
Consequently, if $u\in W^{2,p}(\Omega)$, then, for every
$i,j\in\{1,\ldots,n\}$,
\begin{equation}
  \partial_i u=0,
  \qquad
  \partial_j\partial_i u=0
  \qquad
  \text{almost everywhere on }\{u=c\}.
  \label{eq:second-order-locality}
\end{equation}
\end{lemma}

We now apply this lemma to $h$ on the manifold $X=\CP^3$. Let
\begin{equation}
  (U;y^1,\ldots,y^6)
\end{equation}
be a smooth real coordinate chart. By
\cref{eq:h-H2-regularity}, the coordinate representation of $h$ belongs
to $W^{2,2}_{\mathrm{loc}}$. Therefore, on
\begin{equation}
  E\cap U=\{x\in U:h(x)=0\},
\end{equation}
Lemma~\ref{lem:locality} gives
\begin{equation}
  \partial_i h=0
  \qquad
  \text{almost everywhere on }E\cap U,
  \qquad i=1,\ldots,6,
  \label{eq:B87}
\end{equation}
and
\begin{equation}
  \partial_i\partial_j h=0
  \qquad
  \text{almost everywhere on }E\cap U,
  \qquad i,j=1,\ldots,6.
  \label{eq:B88}
\end{equation}

It remains to relate these weak derivatives to the Laplace--Beltrami
operator. In the coordinates $(y^1,\ldots,y^6)$,
\begin{equation}
  \Delta h
  =
  |g|^{-1/2}
  \sum_{i=1}^{6}
  \partial_i
  \left(
    |g|^{1/2}
    \sum_{j=1}^{6}
    g^{ij}\partial_jh
  \right),
  \label{eq:laplace-beltrami-local}
\end{equation}
where
\begin{equation}
  g=(g_{ij}),
  \qquad
  |g|=\det(g_{ij}),
\end{equation}
and $(g^{ij})$ denotes the inverse metric matrix.

Since the coefficients of the metric are smooth, expanding
\cref{eq:laplace-beltrami-local} gives
\begin{equation}
  \Delta h
  =
  \sum_{i=1}^{6}\sum_{j=1}^{6}
  g^{ij}\partial_i\partial_jh
  +
  \sum_{j=1}^{6}
  b^j\partial_jh,
  \label{eq:laplace-beltrami-expanded}
\end{equation}
for suitable smooth coefficients $b^1,\ldots,b^6$. Because $h\in W^{2,2}(X)$, all terms in
\cref{eq:laplace-beltrami-expanded} are well-defined as $L^2$ functions,
and the formula agrees almost everywhere with the distributional
Laplacian in \cref{eq:B86}. On $E\cap U$, however,
\cref{eq:B87,eq:B88} imply
\begin{equation}
  \partial_jh=0,
  \qquad j=1,\ldots,6,
\end{equation}
and
\begin{equation}
  \partial_i\partial_jh=0,
  \qquad i,j=1,\ldots,6,
\end{equation}
almost everywhere. Consequently,
\begin{equation}
  \Delta h=0
  \qquad
  \text{almost everywhere on }E\cap U.
\end{equation}

Since $X$ is compact, it can be covered by finitely many such coordinate
charts. Combining the corresponding almost-everywhere statements gives
the global conclusion
\begin{equation}
  \Delta h=0
  \qquad
  \text{almost everywhere on }E=\{h=0\}.
  \label{eq:global-laplacian-zero-level}
\end{equation}

\subsection{Why the argument breaks down below dimension four}
\label{sec:dimension-threshold}

The role of the dimension becomes transparent at the spectral level.
Let
\begin{equation}
  X_d:=\CP^{d-1}.
\end{equation}
The same argument used above gives the multiplier formula
\begin{equation}
  \phi_\ell^{(d)}(t)
  =
  \frac{P_\ell^{(d-2,0)}(2t-1)}
       {P_\ell^{(d-2,0)}(1)},
  \qquad
  \Delta F
  =
  -4\ell(\ell+d-1)F,
  \quad F\in\mathcal H_\ell.
  \label{eq:general-phi}
\end{equation}
At the perfect-exclusion point $t=0$, the Jacobi endpoint identities give
\begin{equation}
  r_\ell^{(d)}
  :=
  \phi_\ell^{(d)}(0)
  =
  \frac{(-1)^\ell}
       {\binom{\ell+d-2}{\ell}}.
  \label{eq:general-r}
\end{equation}
Hence
\begin{equation}
  \abs{r_\ell^{(d)}}
  \simeq
  \ell^{-(d-2)}
  \qquad
  (\ell\to\infty).
  \label{eq:general-r-decay}
\end{equation}
Thus $R_0$ gains $d-2$ derivatives in the spectral scale associated with
the Laplace--Beltrami operator.

The proof of \cref{lem:zero-slope} requires at least two derivatives of
smoothing, since one needs
\begin{equation}
  \Delta R_0B\in L^2(X_d)
\end{equation}
for arbitrary $B\in L^2(X_d)$. On the $\ell$th eigenspace, the relevant
multiplier is
\begin{equation}
  4\ell(\ell+d-1)r_\ell^{(d)},
\end{equation}
whose magnitude satisfies
\begin{equation}
  4\ell(\ell+d-1)
  \abs{r_\ell^{(d)}}
  \simeq
  \ell^{4-d}.
  \label{eq:dimension-threshold-multiplier}
\end{equation}
This immediately identifies $d=4$ as the threshold:
\begin{equation}
  \ell^{4-d}
  \begin{cases}
    \longrightarrow\infty, & d<4,\\
    \simeq 1, & d=4,\\
    \longrightarrow0, & d>4.
  \end{cases}
\end{equation}
Consequently, $\Delta R_0$ is not bounded on $L^2$ when $d<4$, is
bounded at the critical dimension $d=4$, and becomes increasingly
regularizing for $d>4$.

The same threshold also appears in the first variation at $t=0$. The
Jacobi differential equation gives
\begin{equation}
  \phi_\ell^{(d)\prime}(0)
  =
  -\ell(\ell+d-1)r_\ell^{(d)},
  \label{eq:general-derivative}
\end{equation}
and therefore
\begin{equation}
  \abs{\phi_\ell^{(d)\prime}(0)}
  \simeq
  \ell^{4-d}.
  \label{eq:general-derivative-growth}
\end{equation}
Thus the derivatives of the individual spectral multipliers are already
unbounded in $\ell$ when $d<4$. In particular, no frequency-uniform
$L^2$ estimate for the difference quotients analogous to the one used in
the proof above can hold in these dimensions.

For example, when $d=3$, so that $X_3=\CP^2$,
\begin{equation}
  r_\ell^{(3)}
  =
  \frac{(-1)^\ell}{\ell+1},
  \qquad
  \Delta F
  =
  -4\ell(\ell+2)F,
  \quad F\in\mathcal H_\ell.
\end{equation}
Hence
\begin{equation}
  \abs{\phi_\ell^{(3)\prime}(0)}
  =
  \frac{\ell(\ell+2)}{\ell+1}
  \simeq
  \ell,
\end{equation}
which diverges as $\ell\to\infty$. Equivalently, $R_0$ gains only one
derivative, whereas the zero-set argument requires control of the
second-order quantity $\Delta R_0B$. This is precisely where the present
proof fails for qutrits.

By contrast, there is no analogous spectral obstruction in dimensions
$d\ge4$. At $d=4$, the decay
\begin{equation}
  \abs{r_\ell^{(4)}}\simeq\ell^{-2}
\end{equation}
exactly compensates for the quadratic growth of the Laplacian and gives
the $H^2$ regularity used above. For $d>4$, one has the stronger decay
\begin{equation}
  \abs{r_\ell^{(d)}}\simeq\ell^{-(d-2)},
\end{equation}
so that $R_0$ gains strictly more than two derivatives and
\begin{equation}
  \abs{\phi_\ell^{(d)\prime}(0)}
  \simeq
  \ell^{4-d}
\end{equation}
even decays at high frequency. Thus the spectral mechanism underlying the
proof becomes more favorable as the dimension increases. In particular,
once the corresponding frequency-uniform Jacobi estimates are established,
the same strategy is expected to extend to every $d\ge4$.

Therefore, dimension four is not an accidental feature of the argument:
it is the smallest dimension in which perfect-exclusion averaging provides
the two derivatives of smoothing needed to control
$\Delta R_0B$ on its zero set.

This also explains why the present argument does not establish an infinite classical communication cost for qutrits. In
\cref{app:d3}, we show instead that the classical communication cost in
the qutrit case is finite.

%% file: 7_appendix_bell_v2.tex
In the previous section, we considered the exact classical simulation of
$d$-dimensional quantum communication. There, the task is to reproduce, for
every state $\rho\in\mathcal S_d$ and every measurement
$M\in\mathcal M_{d,n}$, the quantum probabilities
$p_Q(b\mid\rho,M)=\tr(\rho M_b)$ using shared randomness and a classical
message of bounded dimension, as in \cref{eq:classical-simulation}. In the
Bell setting, we instead fix a bipartite state $\rho_{AB}$ and ask whether,
for arbitrary local measurements $M^A$ and $M^B$, the probabilities
$p_Q(a,b\mid M^A,M^B)
=\tr[\rho_{AB}(M_a^A\otimes M_b^B)]$ can be reproduced by the classical
model in \cref{eq:bell-classical-simulation}.

These two simulation problems are closely related. It is known that a classical simulation of arbitrary $d$-dimensional quantum communication can be converted into an exact classical simulation of the local measurement statistics of any bipartite state whose Bob subsystem has dimension $d$, without increasing the message dimension \cite{CerfGisinMassar2000,RennerTavakoliQuintino}. For completeness, we briefly recall this construction in \cref{sec:bell-qutrit} and combine it with our qutrit protocol from \cref{app:d3}. As a consequence, the statistics generated by arbitrary local measurements on any bipartite state whose Bob subsystem is a qutrit can be simulated exactly with $357$ bits of one-way classical communication. This bound is not intended to be optimal; our purpose here is only to establish that the required communication is finite.

For maximally entangled states, a converse relation also holds at the level
of finiteness. Using a standard construction closely related to quantum
teleportation \cite{bennett1993teleporting}, a finite classical simulation
of all local measurement statistics of $\ket{\Phi_d}$ can be converted into
a finite classical simulation of arbitrary $d$-dimensional quantum
communication. We establish this reduction in \cref{sec:bell-ququart-full}. Combined with
\cref{thm:d4-infinite}, it implies that the statistics generated by arbitrary local measurements on $\ket{\Phi_d}$ cannot be simulated with any finite
message dimension for $d\geq4$.

The construction used in \cref{sec:bell-ququart-full}, however, relies on $d^2$-outcome measurements on Alice's side. It therefore leaves open whether the same obstruction persists for much simpler local measurements. This question is particularly natural because the impossibility result of \cref{thm:d4-infinite} already arises from binary projective measurements in the quantum communication setting. In \cref{sec:bell-ququart-binary}, we show that the analogous statement holds for the maximally entangled ququart: already the family of rank-one binary projective measurements requires an infinite classical message dimension for exact simulation. The proof is a direct adaptation of the zero-slope argument used to prove \cref{thm:d4-infinite}. This result can be contrasted with the seminal work of Regev and Toner \cite{RegevToner}, showing that arbitrary binary quantum correlators admit finite-communication simulations; we discuss this distinction in detail below.

\subsection{Finite classical simulation of bipartite states with a qutrit subsystem}
\label{sec:bell-qutrit}

Consider a bipartite state $\rho_{AB}$, with $\dim\mathcal H_B=d$, and
arbitrary local measurements $M^A=\{M_a^A\}_a$ and
$M^B=\{M_b^B\}_b$. Alice first generates her outcome according to its quantum
marginal,
\begin{equation}
    p_Q(a\mid M^A)
    =
    \tr\!\left[\rho_{AB}(M_a^A\otimes\eye_d)\right].
\end{equation}
Associated with this outcome is the normalized conditional state on Bob's
subsystem
\begin{equation}
    \rho_{B\mid a,M^A}
    :=
    \frac{
        \tr_A\!\left[(M_a^A\otimes\eye_d)\rho_{AB}\right]
    }{
        p_Q(a\mid M^A)
    },
    \label{eq:conditional-bob-state}
\end{equation}
whenever $p_Q(a\mid M^A)>0$. The joint quantum probability can therefore be
written as
\begin{equation}
    p_Q(a,b\mid M^A,M^B)
    =
    p_Q(a\mid M^A)\,
    \tr\!\left(\rho_{B\mid a,M^A}M_b^B\right).
    \label{eq:bell-as-pm}
\end{equation}
Conditioned on her outcome $a$, Alice can now use the classical
prepare-and-measure protocol to simulate the transmission of
$\rho_{B\mid a,M^A}$ to Bob, who uses $M^B$ as his measurement. Thus any
finite-message simulation of arbitrary $d$-dimensional quantum communication
gives a classical simulation of the local measurement statistics of $\rho_{AB}$ with exactly the same message dimension.

Combining this reduction with the qutrit simulation established in
\cref{app:d3} immediately gives a finite upper bound for the classical communication required to simulate the local measurement statistics of any bipartite state whose Bob subsystem is a qutrit.

\begin{corollary}[Finite simulation of bipartite states with a qutrit subsystem]
\label{cor:bell-qutrit-finite}

Let $\rho_{AB}$ be any bipartite state on
$\mathcal H_A\otimes\C^3$, with $\mathcal H_A$ finite dimensional.
Then the correlations $\tr[\rho_{AB}(M_a^A\otimes M_b^B)]$ generated by $\rho_{AB}$ under arbitrary
finite-outcome local POVMs admit an exact classical simulation with
one-way communication from Alice to Bob using a message alphabet of dimension
\begin{equation}
    d_C < 2^{357}.
\end{equation}
In particular, $357$ classical bits suffice.
\end{corollary}

\begin{proof}
By \cref{app:d3}, arbitrary qutrit states and finite-outcome qutrit POVMs
can be simulated exactly with a classical message alphabet of dimension
$d_C<2^{357}$. Applying this quantum communication simulation to the
conditional states $\rho_{B\mid a,M^A}$ in
\cref{eq:bell-as-pm} gives the claimed simulation of the local measurement statistics without increasing the message dimension.
\end{proof}

\subsection{Infinite cost for maximally entangled ququarts under arbitrary local measurements}
\label{sec:bell-ququart-full}

The reduction discussed above shows that a classical simulation of arbitrary
$d$-dimensional quantum communication immediately yields a simulation of the
local measurement statistics of any bipartite state whose Bob subsystem has
dimension $d$. For maximally entangled states, there is also a converse
relation at the level of finiteness. Namely, a classical simulation of all
local measurement statistics of a maximally entangled state can be converted
into a classical simulation of arbitrary quantum communication of the
corresponding dimension. The construction is closely related to teleportation \cite{bennett1993teleporting}.

For $d\geq2$, let
\begin{equation}
    \ket{\Phi_d}
    :=
    \frac{1}{\sqrt d}
    \sum_{j=0}^{d-1}
    \ket{j}\otimes\ket{j}
    \in\C^d\otimes\C^d
    \label{eq:max-ent-state}
\end{equation}
be the maximally entangled state.

\begin{proposition}[From maximally entangled states to quantum communication]
\label{prop:bell-to-pm-max-ent}

Suppose that the local measurement statistics of $\ket{\Phi_d}$ admit a classical simulation with classical message dimension $d_C$. Then arbitrary $d$-dimensional quantum communication statistics admit a classical simulation with message dimension
\begin{equation}
    d'_{C} \le d^2 d_C.
\end{equation}
\end{proposition}

\begin{proof}
Let
\begin{equation}
    X\ket{j}=\ket{j+1 \!\!\!\pmod d},
    \qquad
    Z\ket{j}=\omega^j\ket{j},
    \qquad
    \omega=e^{2\pi i/d},
\end{equation}
and define the $d^2$ unitaries
\begin{equation}
    U_{r,s}:=X^r Z^s,
    \qquad
    r,s\in\{0,\ldots,d-1\}.
\end{equation}
For convenience, we denote the pair $(r,s)$ by a single index
$k\in[d^2]$ and write the corresponding unitary as $U_k$.

These unitaries satisfy the identity
\begin{equation}
    \frac{1}{d^2}
    \sum_{k=1}^{d^2}
    U_k A U_k^\dagger
    =
    \frac{\tr(A)}{d}\eye_d
    \label{eq:weyl-twirl}
\end{equation}
for every operator $A\in\C^{d\times d}$ \cite{bennett1993teleporting}. In particular, for every quantum
state $\rho\in\mathcal S_d$,
\begin{equation}
    \frac{1}{d^2} \sum_{k=1}^{d^2}
     U_k\rho U_k^\dagger
    =
    \frac{\eye_d}{d}.
    \label{eq:weyl-state-twirl}
\end{equation}

Fix now an arbitrary state $\rho\in\mathcal S_d$. Associated with $\rho$,
consider on Alice's subsystem the $d^2$-outcome POVM
\begin{equation}
    E^\rho
    =
    \left\{
        E_k^\rho
    \right\}_{k=1}^{d^2},
    \qquad
    E_k^\rho
    :=
    \frac1d
    \left(
        U_k\rho U_k^\dagger
    \right)^\top,
    \label{eq:rsp-povm}
\end{equation}
where the transpose is taken in the computational basis used to define
$\ket{\Phi_d}$. By \cref{eq:weyl-state-twirl},
\begin{equation}
    \sum_{k=1}^{d^2}E_k^\rho=\eye_d,
\end{equation}
so $E^\rho$ is indeed a POVM.
A well-known identity for the maximally entangled state is
\begin{equation}
    \bra{\Phi_d}
    A\otimes B
    \ket{\Phi_d}
    =
    \frac1d\tr(A^\top B),
    \label{eq:max-ent-identity-general}
\end{equation}
valid for arbitrary operators $A,B\in\C^{d\times d}$.

Now let
\begin{equation}
    M=\{M_b\}_{b=1}^n\in\mathcal M_{d,n}
\end{equation}
be an arbitrary measurement that is to be performed on $\rho$. For each
$k\in[d^2]$, define
\begin{equation}
    M^{(k)}
    :=
    \left\{
        M_b^{(k)}
    \right\}_{b=1}^n,
    \qquad
    M_b^{(k)}
    :=
    U_k M_b U_k^\dagger.
    \label{eq:rotated-bob-measurement}
\end{equation}
Since unitary conjugation preserves positivity and normalization,
$M^{(k)}$ is a valid POVM.

For the pair of local measurements $(E^\rho,M^{(k)})$, the probability of the
joint outcome $(k,b)$ is
\begin{align}
    p_Q(k,b\mid E^\rho,M^{(k)})
    &=
    \bra{\Phi_d}
    E_k^\rho\otimes M_b^{(k)}
    \ket{\Phi_d}
    \nonumber\\
    &=
    \frac1d
    \tr\left[
        (E_k^\rho)^\top M_b^{(k)}
    \right]
    \nonumber\\
    &=
    \frac1{d^2}
    \tr\left[
        U_k\rho U_k^\dagger
        U_kM_bU_k^\dagger
    \right]
    \nonumber\\
    &=
    \frac1{d^2}\tr(\rho M_b).
    \label{eq:rsp-joint-probability}
\end{align}

By assumption, the local measurement statistics of $\ket{\Phi_d}$ admit a
classical simulation with message dimension $d_C$. Hence there exist shared
randomness $\lambda$ and response functions such that, for every pair of local measurements $E$ and $N$,
\begin{equation}
    p_Q(k,b\mid E,N)
    =
    \int_\Lambda
    \sum_{c=1}^{d_C}
    p_A(k,c\mid E,\lambda)
    p_B(b\mid N,c,\lambda)
    \,\dd\mu(\lambda).
    \label{eq:bell-model-rsp-proof}
\end{equation}
In particular, taking $E=E^\rho$ and $N=M^{(k)}$ gives
\begin{equation}
    \frac1{d^2}\tr(\rho M_b)
    =
    \int_\Lambda
    \sum_{c=1}^{d_C}
    p_A(k,c\mid E^\rho,\lambda)
    p_B(b\mid M^{(k)},c,\lambda)
    \,\dd\mu(\lambda).
    \label{eq:bell-model-rsp-specialized}
\end{equation}

We now turn this local measurement statistics simulation into the quantum communication protocol. Given
the state $\rho$, encoder of the classical simulation to the measurement $E^\rho$
and obtains a pair $(k,c)$. She sends both values to Bob, which requires a
message alphabet of dimension at most $d^2d_C$. Given the target measurement
$M$ and the received pair $(k,c)$, Bob applies the corresponding decoder associated
with $M^{(k)}$. Thus define
\begin{align}
    \widetilde p_A(k,c\mid\rho,\lambda)
    &:=
    p_A(k,c\mid E^\rho,\lambda),
    \\
    \widetilde p_B(b\mid M,k,c,\lambda)
    &:=
    p_B(b\mid M^{(k)},c,\lambda).
\end{align}

Finally, summing \cref{eq:bell-model-rsp-specialized} over $k$ yields
\begin{align}
    &
    \int_\Lambda
    \sum_{k=1}^{d^2}
    \sum_{c=1}^{d_C}
    \widetilde p_A(k,c\mid\rho,\lambda)
    \widetilde p_B(b\mid M,k,c,\lambda)
    \,\dd\mu(\lambda)
    \nonumber\\
    &\hspace{3cm}
    =
    \sum_{k=1}^{d^2}
    \frac1{d^2}
    \tr(\rho M_b)
    =
    \tr(\rho M_b).
\end{align}
This is an exact classical simulation of arbitrary $d$-dimensional quantum
communication with message dimension at most $d^2d_C$.
\end{proof}

The proposition above, together with the converse reduction discussed in
\cref{sec:bell-qutrit}, shows that, at the level of finiteness, the
classical simulation of $d$-dimensional quantum communication and the
classical simulation of all measurement statistics of $\ket{\Phi_d}$ are
equivalent, up to a finite overhead in the message dimension.

In particular, our impossibility result for quantum communication immediately
gives the following consequence.

\begin{corollary}[Infinite cost for maximally entangled ququarts]
\label{cor:bell-ququart-full-infinite}

The local measurement statistics of the maximally entangled ququart state $\ket{\Phi_4}$ do not admit a
classical simulation with any finite message dimension.
\end{corollary}

\begin{proof}
If such a simulation existed with some finite $d_C$, then
\cref{prop:bell-to-pm-max-ent} would yield an exact classical simulation of
arbitrary ququart communication with message dimension at most $16d_C$.
This contradicts \cref{thm:d4-infinite}.
\end{proof}

The same argument, together with the higher-dimensional part of
\cref{thm:d4-infinite}, gives the corresponding statement in every larger
dimension.

\begin{corollary}[Higher-dimensional maximally entangled states]
\label{cor:bell-max-ent-full-higher-d}

For every $d\geq4$, the maximally entangled state $\ket{\Phi_d}$ does not
admit an exact classical simulation with any finite message dimension.
\end{corollary}

\subsection{Infinite cost already for binary projective measurements}
\label{sec:bell-ququart-binary}

The construction in \cref{sec:bell-ququart-full} relies on the
$d^2$-outcome POVM $E^\rho$ on Alice's side defined in
\cref{eq:rsp-povm}. While this is enough to establish that the statistics generated by local measurements in maximally
entangled ququart cannot be simulated with any finite message dimension, it
does not show whether the same obstruction already appears under simpler
measurement restrictions.

This question is particularly natural in view of the proof of
\cref{thm:d4-infinite}, where infinite classical communication is already
required in the quantum communication setting when restricted to pure states
and binary projective measurements with a rank-one effect, \cref{sec:ReductionPureStatesAndMeas_ququart}. We now show that the same phenomenon occurs for the local measurement statistics of the maximally entangled ququart. More precisely, the statistics generated by binary projective measurements with a rank-one effect on the maximally
entangled ququart state already admits no exact classical simulation with a finite
message alphabet. The proof closely follows that of \cref{thm:d4-infinite}, again relying on
the zero-slope argument of \cref{lem:zero-slope}.

We use the notation introduced in
\cref{sec:ReductionPureStatesAndMeas_ququart}. In particular, for
$x,y\in\CP^3$,
\begin{equation}
    P_x=\ketbra{x},
    \qquad
    f(x,y)=\tr(P_xP_y)=\abs{\langle x|y\rangle}^2.
    \label{eq:bell-fidelity}
\end{equation}

Fix the computational basis $\{\ket{j}\}_{j=0}^3$ and let $\bar{x}$ denote the ray obtained by entrywise complex conjugation in this basis. Alice and Bob perform the binary projective measurements
\begin{equation}
    M_x^A
    =
    \left\{
        P_{\bar{x}},
        \eye_4-P_{\bar{x}}
    \right\},
    \qquad
    M_y^B
    =
    \left\{
        P_y,
        \eye_4-P_y
    \right\}.
    \label{eq:bell-binary-measurements}
\end{equation}
We label the rank-one effect in each measurement by the outcome $0$.\footnote{The complex conjugation on Alice's side is only a convenient parametrization, chosen so that the joint probability below is directly
expressed in terms of the fidelity $f(x,y)$. One could instead use $P_x$; this would replace $f(x,y)$ by $f(\bar{x},y)$ without affecting the argument, while making the parallel with \cref{thm:d4-infinite}
slightly less transparent.}

For $\ket{\Phi_4}$, \cref{eq:max-ent-identity-general} gives
\begin{equation}
    \bra{\Phi_4}
    A\otimes B
    \ket{\Phi_4}
    =
    \frac14\tr(A^\top B).
    \label{eq:max-ent-transpose-identity}
\end{equation}
Since
\begin{equation}
    P_{\bar{x}}^\top=P_x,
\end{equation}
the probability of the joint outcome $a=b=0$ is
\begin{align}
    p_Q(0,0\mid x,y)
    &=
    \bra{\Phi_4}
    P_{\bar{x}}\otimes P_y
    \ket{\Phi_4}
    \nonumber\\
    &=
    \frac14\tr(P_xP_y)
    =
    \frac14 f(x,y).
    \label{eq:bell-ququart-target}
\end{align}

Thus, up to the constant factor $1/4$, the same fidelity dependence that
arises in the problem of classically simulating ququart quantum communication
also appears here. The proof of \cref{thm:d4-infinite} can therefore be
closely followed, again using the zero-slope argument of
\cref{lem:zero-slope}, to show that the probabilities in
\cref{eq:bell-ququart-target} cannot be reproduced exactly with any finite
message alphabet.

\begin{theorem}[Infinite communication with binary projective measurements]
\label{thm:bell-ququart-binary-infinite}

The local measurement statistics generated by $\ket{\Phi_4}$ and the family of binary projective measurements in \cref{eq:bell-binary-measurements} do not admit an exact classical simulation with any finite message dimension $d_C$.
\end{theorem}

\begin{proof}
Assume, toward a contradiction, that such a simulation exists for some
$d_C<\infty$. Applying \cref{eq:bell-classical-simulation} to the joint
outcome $a=b=0$ and the measurements in
\cref{eq:bell-binary-measurements}, and using
\cref{eq:bell-ququart-target}, gives, for every $x,y\in\CP^3$,
\begin{equation}
    \frac14 f(x,y)
    =
    \int_\Lambda
    \sum_{c=1}^{d_C}
    p_A(0,c\mid x,\lambda)
    p_B(0\mid y,c,\lambda)
    \,\dd\mu(\lambda).
    \label{eq:bell-factorization}
\end{equation}

For every $\lambda\in\Lambda$ and $c\in[d_C]$, define
\begin{align}
    A_{\lambda,c}(x)
    &:=
    p_A(0,c\mid x,\lambda),
    \\
    B_{\lambda,c}(y)
    &:=
    p_B(0\mid y,c,\lambda).
\end{align}

For each fixed $(\lambda,c)$, let
\begin{equation}
    a_{\lambda,c}(t)
    :=
    \int_{\CP^3\times\CP^3}
    A_{\lambda,c}(x)
    B_{\lambda,c}(y)
    \,\dd\eta_t(x,y),
    \label{eq:bell-branch-curve}
\end{equation}
where $\eta_t$ is the fixed-fidelity measure introduced in
\cref{sec:fixed-fidelity-pairs}.

We now average \cref{eq:bell-factorization} with respect to $\eta_t$.
On the left-hand side, $f(x,y)=t$ on the support of $\eta_t$. On the
classical side, the integrand is measurable and nonnegative, so the
Fubini--Tonelli theorem allows us to exchange the integrations over
$\CP^3\times\CP^3$ and $\Lambda$. Hence,
\begin{equation}
    \frac{t}{4}
    =
    \sum_{c=1}^{d_C}
    \int_\Lambda
    a_{\lambda,c}(t)
    \,\dd\mu(\lambda).
    \label{eq:bell-averaged-factorization}
\end{equation}

At $t=0$,
\begin{equation}
    0
    =
    \sum_{c=1}^{d_C}
    \int_\Lambda
    a_{\lambda,c}(0)
    \,\dd\mu(\lambda).
    \label{eq:bell-average-zero}
\end{equation}
Every term in this expression is nonnegative. Therefore, for every
$c\in[d_C]$,
\begin{equation}
    a_{\lambda,c}(0)=0
\end{equation}
for $\mu$-almost every $\lambda$.

For each $c\in[d_C]$, define
\begin{equation}
    \Gamma_c
    :=
    \left\{
        \lambda\in\Lambda:
        a_{\lambda,c}(0)>0
    \right\},
\end{equation}
and let
\begin{equation}
    \Gamma
    :=
    \bigcup_{c\in[d_C]}\Gamma_c.
\end{equation}
Since $d_C$ is finite, $\mu(\Gamma)=0$. Thus, for every
$\lambda\in\Lambda\setminus\Gamma$ and every $c\in[d_C]$, the hypothesis
of \cref{lem:zero-slope} is satisfied. Consequently,
\begin{equation}
    \frac{a_{\lambda,c}(t)}{t}
    \longrightarrow0
    \qquad
    \text{as }t\to0^+,
    \label{eq:bell-branch-zero-slope}
\end{equation}
and, for $0<t\leq\frac12$,
\begin{equation}
    0
    \leq
    \frac{a_{\lambda,c}(t)}{t}
    \leq8.
    \label{eq:bell-branch-uniform-bound}
\end{equation}

Dividing \cref{eq:bell-averaged-factorization} by $t>0$ gives
\begin{equation}
    \frac14
    =
    \sum_{c=1}^{d_C}
    \int_{\Lambda\setminus\Gamma}
    \frac{a_{\lambda,c}(t)}{t}
    \,\dd\mu(\lambda).
    \label{eq:bell-final-average}
\end{equation}
By \cref{eq:bell-branch-zero-slope,eq:bell-branch-uniform-bound} and the
dominated convergence theorem, for every $c\in[d_C]$,
\begin{equation}
    \lim_{t\to0^+}
    \int_{\Lambda\setminus\Gamma}
    \frac{a_{\lambda,c}(t)}{t}
    \,\dd\mu(\lambda)
    =
    0.
\end{equation}
Since the number of messages is finite, we may also pass the limit through
the sum over $c$. Therefore, the right-hand side of
\cref{eq:bell-final-average} converges to zero as $t\to0^+$, contradicting
the fact that its left-hand side is equal to $1/4$.

Hence no finite message dimension can reproduce these statistics exactly.
\end{proof}

The same restricted family can be embedded into maximally entangled systems
of any larger local dimension.

\begin{corollary}[Higher-dimensional maximally entangled states]
\label{cor:bell-max-ent-binary-higher-d}

For every $d\geq4$, there exists a family of binary local projective
measurements on $\ket{\Phi_d}$ whose local measurement statistics cannot be simulated with any finite amount of one-way classical communication.
\end{corollary}

\begin{proof}
Fix a four-dimensional subspace of $\C^d$ and identify it with $\C^4$.
Restrict $x,y$ to rays in the corresponding copy of $\CP^3$, and define the
measurements as in \cref{eq:bell-binary-measurements}, with the complementary
effects now taken with respect to $\eye_d$.

For these measurements,
\begin{equation}
    p_Q(0,0\mid x,y)
    =
    \bra{\Phi_d}
    P_{\bar{x}}\otimes P_y
    \ket{\Phi_d}
    =
    \frac1d\tr(P_xP_y)
    =
    \frac1d f(x,y).
\end{equation}
If a finite classical simulation existed, averaging its decomposition over
the same measures $\eta_t$ on $\CP^3\times\CP^3$ would give
\begin{equation}
    \frac{t}{d}
    =
    \sum_{c=1}^{d_C}
    \int_\Lambda
    a_{\lambda,c}(t)
    \,\dd\mu(\lambda).
\end{equation}
The proof of \cref{thm:bell-ququart-binary-infinite} then applies without
change: \cref{lem:zero-slope} forces the right-hand side divided by $t$ to
converge to zero as $t\to0^+$, whereas the left-hand side divided by $t$ is
the nonzero constant $1/d$.
\end{proof}

\subsection{Relation with the Regev--Toner simulation of binary correlators}
\label{sec:regev-toner}

A particularly important comparison is with the work of Regev and Toner
\cite{RegevToner}. They showed that, for an arbitrary bipartite state on
$\C^d\otimes\C^d$ and arbitrary binary local measurements, the corresponding
quantum correlator can be reproduced exactly with only two bits of one-way
classical communication, for every finite $d$.

This result applies, in particular, to the maximally entangled ququart and
the binary projective measurements considered in
\cref{sec:bell-ququart-binary}. Nevertheless, it does not contradict
\cref{thm:bell-ququart-binary-infinite}. As already emphasized by Regev and
Toner, the classical communication protocol they construct reproduces only
the binary correlator. The resulting classical model has uniform marginals,
which in general differ from the corresponding quantum marginals, and
therefore does not reproduce the full local measurement statistics.

What is particularly striking is the abrupt separation exhibited by this binary-outcome setting: while its correlators can be reproduced exactly with only two bits of classical communication, reproducing the complete quantum statistics requires an infinite classical message dimension.

\subsubsection{The Regev--Toner simulation problem}

Consider local measurement statistics with binary outcomes $a,b\in\{0,1\}$. The correlator associated with the inputs $x$ and $y$ is
\begin{align}
    \langle \alpha_x\beta_y\rangle
    &:=
    p(0,0\mid x,y)
    +
    p(1,1\mid x,y)
    -
    p(0,1\mid x,y)
    -
    p(1,0\mid x,y).
    \label{eq:binary-correlator}
\end{align}

The simulation problem considered by Regev and Toner is to reproduce these
correlators, rather than the complete conditional distribution
$p(a,b\mid x,y)$. They show that, for arbitrary bipartite quantum states of
finite local dimension and arbitrary binary local measurements, all such
quantum correlators can be reproduced exactly using only two bits of
one-way classical communication \cite{RegevToner}.

Reproducing the correlator is, however, strictly weaker than reproducing the
full local measurements statistics. For binary outcomes, the conditional distribution is
determined not only by $\langle\alpha_x\beta_y\rangle$, but also by the two
local averages $\langle\alpha_x\rangle : p(0|x) - p(1|x)$ and $\langle\beta_y\rangle := p(0|y) - p(1|y)$. More
explicitly,
\begin{equation}
    p(a,b\mid x,y)
    =
    \frac14
    \left[
        1
        +
        (-1)^a\langle\alpha_x\rangle
        +
        (-1)^b\langle\beta_y\rangle
        +
        (-1)^{a+b}\langle\alpha_x\beta_y\rangle
    \right].
    \label{eq:binary-behavior-moments}
\end{equation}

Regev and Toner explicitly point out this limitation. The classical
communication protocol they construct reproduces the desired correlators,
but its local marginals are uniform,
\begin{equation}
    p_C(a\mid x)
    =
    p_C(b\mid y)
    =
    \frac12,
    \label{eq:RT-uniform-marginals}
\end{equation}
for all $a,b,x,y$. These marginals need not coincide with the corresponding
quantum marginals, and therefore their protocol does not in general
reproduce the complete quantum statistics.

In the next subsection, we make this distinction explicit for the maximally
entangled ququart and the binary rank-one projective measurements appearing
in \cref{thm:bell-ququart-binary-infinite}, by writing out the full quantum
behavior and comparing its marginals with those of the Regev--Toner model.

\subsubsection{The ququart statistics: correlator and marginals}

Consider now the binary local measurement statistics generated by the maximally entangled
ququart state $\ket{\Phi_4}$ and the binary projective measurements introduced in
\cref{eq:bell-binary-measurements},
\begin{equation}
    M_x^A
    =
    \left\{
        P_{\bar{x}},
        \eye_4-P_{\bar{x}}
    \right\},
    \qquad
    M_y^B
    =
    \left\{
        P_y,
        \eye_4-P_y
    \right\},
    \label{eq:RT-ququart-measurements}
\end{equation}
where outcome $0$ corresponds to the rank-one effect.

The probability of the joint outcome
$(0,0)$ is
\begin{align}
    p_Q(0,0\mid x,y)
    &=
    \bra{\Phi_4}
    P_{\bar{x}}\otimes P_y
    \ket{\Phi_4} =
    \frac14\tr(P_xP_y)
    =
    \frac14\abs{\langle x|y\rangle}^2.
    \label{eq:RT-ququart-p00}
\end{align}

The other three probabilities can be obtained in the same way. Namely,
\begin{align}
    p_Q(0,1\mid x,y)
    &=
    \frac14
    \tr\!\left[
        P_x(\eye_4-P_y)
    \right]
    =
    \frac{1-\abs{\langle x|y\rangle}^2}{4},
    \label{eq:RT-ququart-p01}
    \\
    p_Q(1,0\mid x,y)
    &=
    \frac14
    \tr\!\left[
        (\eye_4-P_x)P_y
    \right]
    =
    \frac{1-\abs{\langle x|y\rangle}^2}{4},
    \label{eq:RT-ququart-p10}
    \\
    p_Q(1,1\mid x,y)
    &=
    \frac14
    \tr\!\left[
        (\eye_4-P_x)(\eye_4-P_y)
    \right]
    =
    \frac{2+\abs{\langle x|y\rangle}^2}{4}.
    \label{eq:RT-ququart-p11}
\end{align}

The local marginals follow immediately:
\begin{align}
    p_Q(0\mid x)
    &=
    p_Q(0,0\mid x,y)
    +
    p_Q(0,1\mid x,y)
    =
    \frac14,
    \\
    p_Q(0\mid y)
    &=
    p_Q(0,0\mid x,y)
    +
    p_Q(1,0\mid x,y)
    =
    \frac14.
\end{align}

This should be contrasted with the classical communication model of
Regev and Toner, whose marginals satisfy
\cref{eq:RT-uniform-marginals}. Thus, although the model reproduces the
correct correlator, it cannot reproduce the quantum statistics above, since
the local marginals are different.

It is nevertheless interesting that the correlator of this quantum statistics recovers precisely
the fidelity function that played a central role in the previous arguments.
Using \cref{eq:binary-correlator} together with
\crefrange{eq:RT-ququart-p00}{eq:RT-ququart-p11}, we obtain
\begin{align}
    \langle\alpha_x\beta_y\rangle_Q
    &=
    p_Q(0,0\mid x,y)
    +
    p_Q(1,1\mid x,y)
    -
    p_Q(0,1\mid x,y)
    -
    p_Q(1,0\mid x,y)
    \nonumber\\
    &=
    \xy^2\\
    &=
    f(x,y).
    \label{eq:RT-ququart-correlator}
\end{align}

By the Regev--Toner result, the correlator
$\langle\alpha_x\beta_y\rangle_Q=f(x,y)$ can be reproduced exactly using
only two bits of one-way classical communication. In sharp contrast, the
same local measurement statistics satisfy
$p_Q(0,0\mid x,y)=f(x,y)/4$, and
\cref{thm:bell-ququart-binary-infinite} shows that reproducing this joint
probability within an exact classical simulation requires an infinite
classical message dimension.

Given that the correlator is again exactly the fidelity $f(x,y)$, it is
natural to ask why the zero-slope argument used in
\cref{thm:d4-infinite,thm:bell-ququart-binary-infinite} cannot be applied
directly to it. We address this point in the next subsection.

\subsubsection{Why the zero-slope argument does not apply to correlators}

The identity
\begin{equation}
    \langle\alpha_x\beta_y\rangle_Q
    =
    f(x,y)
\end{equation}
may suggest applying directly the same zero-slope argument used in
\cref{thm:d4-infinite,thm:bell-ququart-binary-infinite}.
Such an argument would, however, imply that this correlator cannot be
simulated with finite classical communication, in direct contradiction with
the Regev--Toner result. The obstruction to applying the argument can be
seen explicitly at the level of the classical decomposition.

Consider a one-way classical simulation with finite message alphabet
$c\in[d_C]$. For every $\lambda$ and $c$, define
\begin{align}
    A_{\lambda,c}(x)
    &:=
    p_A(0,c\mid x,\lambda)
    -
    p_A(1,c\mid x,\lambda),
    \\
    B_{\lambda,c}(y)
    &:=
    p_B(0\mid y,c,\lambda)
    -
    p_B(1\mid y,c,\lambda).
\end{align}
The classical correlator can then be written as
\begin{equation}
    \langle\alpha_x\beta_y\rangle_C
    =
    \sum_{c=1}^{d_C}
    \int_\Lambda
    A_{\lambda,c}(x)
    B_{\lambda,c}(y)
    \,\dd\mu(\lambda).
    \label{eq:RT-signed-correlator-factorization}
\end{equation}
In contrast with the decomposition of the joint probability in
\cref{eq:bell-factorization}, the functions appearing here are not
nonnegative. Instead,
\begin{equation}
    A_{\lambda,c}(x),B_{\lambda,c}(y)\in[-1,1].
\end{equation}

To make the comparison with the previous proofs explicit, define
\begin{equation}
    r_{\lambda,c}(t)
    :=
    \int_{\CP^3\times\CP^3}
    A_{\lambda,c}(x)
    B_{\lambda,c}(y)
    \,\dd\eta_t(x,y).
    \label{eq:RT-correlator-branch}
\end{equation}
Averaging \cref{eq:RT-signed-correlator-factorization} over the
fixed-fidelity measure $\eta_t$ and using
\cref{eq:RT-ququart-correlator} gives
\begin{equation}
    t
    =
    \sum_{c=1}^{d_C}
    \int_\Lambda
    r_{\lambda,c}(t)
    \,\dd\mu(\lambda).
    \label{eq:RT-correlator-average}
\end{equation}
At $t=0$, this yields
\begin{equation}
    0
    =
    \sum_{c=1}^{d_C}
    \int_\Lambda
    r_{\lambda,c}(0)
    \,\dd\mu(\lambda).
    \label{eq:RT-correlator-zero}
\end{equation}

This is precisely where the argument differs from the proof of
\cref{thm:bell-ququart-binary-infinite}. There, each contribution to
$p_C(0,0\mid x,y)$ is nonnegative. Consequently, the vanishing of the total
average at $t=0$ forces each contribution to vanish individually, which
allows \cref{lem:zero-slope} to be applied to each term separately. In
\cref{eq:RT-correlator-zero}, by contrast, the quantities
$r_{\lambda,c}(0)$ may have either sign. The vanishing of their sum may
therefore result from cancellations, and no termwise vanishing condition
can be inferred.

Moreover, even if a particular signed contribution happened to satisfy
$r_{\lambda,c}(0)=0$, the zero-slope lemma would still not apply, since its
nonnegativity hypothesis is essential. As shown explicitly in
\cref{rmk:zero-slope-sharpness}, there exist functions
$A,B:\CP^3\to[-1,1]$ for which
\begin{equation}
    a_{A,B}(0)=0
\end{equation}
but
\begin{equation}
    \lim_{t\to0^+}
    \frac{a_{A,B}(t)}{t}
    >0.
\end{equation}
Thus, once signed functions are allowed, the zero-slope conclusion itself
fails.

%% file: 4_appendix_d3proof_v3.tex
We consider the definitions in \cref{app:preliminaries} and let $d_Q = 3$.
Alice prepares a qutrit state $\rho$, and Bob performs an arbitrary finite-outcome qutrit POVM $M=\{M_b\}_{b\in B}$, where $M_b\geq0$ and $\sum_{b\in B}M_b=\eye_3$.
The quantum behavior is then
\begin{equation}
	p_Q(b\mid \rho, M) = \tr(M_b \, \rho)
	\label{eq:d3-povm-target}
\end{equation}
In this section, we construct a protocol that can perfectly simulate \cref{eq:d3-povm-target} using only shared randomness and a finite amount of classical communication.
More specifically, we prove the following theorem.

\begin{theorem}[Finite classical communication cost of simulating qutrit communication]
	\label{thm:d3-finite}
	There exists a shared randomness space $\Lambda$ with distribution $\mu$, an encoder $p_A$, and a decoder $p_B$ such that \cref{eq:classical-simulation} reproduces \cref{eq:d3-povm-target} for every qutrit state $\rho$ and every finite-outcome qutrit POVM $M$, with classical communication of dimension
	\begin{equation}
		d_C<2^{357}.
	\end{equation}
	Thus $357$ one-way classical bits suffice to simulate qutrit communication.
\end{theorem}

Exact average-cost and worst-case protocols were previously known for qubit prepare-and-measure scenarios and for the closely related Bell nonlocality simulation problem \cite{maudlin1992bell,BrassardCost1999,Steiner2000,massar2001classical,CerfGisinMassar2000,DegorreLaplanteRoland2005,TonerBacon,RennerTavakoliQuintino}.
To the best of our knowledge, before the present construction only lower bounds on $d_C$ were known for $d_Q > 2$ \cite{BuhrmanQuantumClassical1998,BrassardCost1999,MontinaCommunication2011,HavlicekSimpleCommunication2020,SchlosserKleinmann2026}.

This appendix is organized as follows.
First, in \cref{sec:average-communication}, we construct a protocol that simulates binary projective measurements on pure qutrits and has finite average communication cost, but with unbounded message dimension.
In \cref{sec:finite-sampling}, we modify this protocol to obtain finite communication.
Finally, in \cref{sec:d3-more-outcomes}, we extend the construction to arbitrary finite-outcome POVMs and mixed states.

\subsection{A protocol for binary measurements with finite average communication}
\label{sec:average-communication}

We first restrict to pure qutrit states and binary projective measurements with a rank-one effect. As in \cref{sec:ReductionPureStatesAndMeas_ququart}, we parametrize both by rays $x,y\in\CP^2$, with associated rank-one projectors
$$
P_x:=\ketbra{x},
\qquad
P_y:=\ketbra{y}.
$$
Alice prepares $P_x$, while Bob performs the binary measurement
\begin{equation}
M_y:=\{P_y,\eye_3-P_y\}.
\label{eq:restricted-scenario}
\end{equation}
We label by $b=0$ the outcome associated with $P_y$ and by $b=1$ the outcome associated with $(\eye_3-P_y)$. In this restricted scenario,
\begin{equation}
p_Q(0\mid x,y)
=
\tr(P_xP_y)
=
|\langle x|y\rangle|^2.
\label{eq:born-target}
\end{equation}
It is therefore sufficient to reproduce the $b=0$ probability, since the probability of $b=1$ follows from normalization.

\subsubsection{A postselected local construction}
\label{sec:postselected-local-construction}

Before introducing any communication, we first identify a simpler probabilistic construction that already contains the desired quantum statistics. Alice and Bob will share a random variable and apply local tests depending only on their respective inputs. We will choose the distribution of the shared variable so that, conditioned on Alice's test succeeding, the probability that Bob's test succeeds is exactly the Born probability
$$
|\langle x|y\rangle|^2.
$$
This does not yet define a classical simulation protocol, since the conditioning event depends on Alice's input and occurs only with some probability. The communication protocol introduced in the next subsection will be used to realize this conditioning.

To this end, consider a shared random variable
$$
\theta=(h,\alpha,\beta)\in\CP^2\times\mathcal R,
$$
where
\begin{equation}
\mathcal R
:=
\left\{
(\alpha,\beta)\in[0,1]^2:
\alpha\geq\beta,\ 
\alpha+\beta\geq1
\right\}.
\label{eq:R}
\end{equation}
The ray $h$ is sampled according to normalized Haar measure on $\CP^2$. Independently of $h$, the thresholds are sampled as follows: with probability $1/3$,
\begin{equation}
\alpha\sim\operatorname{Unif}[1/2,1],
\qquad
\beta=1-\alpha,
\end{equation}
while with probability $2/3$, the pair $(\alpha,\beta)$ is sampled uniformly from $\mathcal R$.

We denote by $\nu$ the resulting probability distribution of $\theta$. More explicitly, let $\sigma$ denote the normalized Haar measure on $\CP^2$, let $\nu_{\mathrm{line}}$ denote the probability measure on $\mathcal R$ obtained by sampling
$$
\alpha\sim\operatorname{Unif}[1/2,1]
$$
and setting $\beta=1-\alpha$, and let $\nu_{\mathcal R}$ denote the uniform probability measure on $\mathcal R$. Then
\begin{equation}
\nu
=
\sigma\otimes
\left(
\frac13\nu_{\mathrm{line}}
+
\frac23\nu_{\mathcal R}
\right).
\label{eq:single-coordinate-distribution}
\end{equation}

Given inputs $x,y\in\CP^2$, Alice can locally compute
$$
|\langle x|h\rangle|^2,
$$
while Bob can locally compute
$$
|\langle y|h\rangle|^2.
$$
We say that Alice's local test succeeds when
\begin{equation}\label{eq:AliceTestSucceedsLocal}
    |\langle x,h\rangle|^2>\alpha
\end{equation}

while Bob's local test succeeds when
\begin{equation}
    |\langle y|h\rangle|^2>\beta.
\end{equation}

The probability that both local tests succeed is
\begin{equation}
\Pr_\nu\left\{
|\langle x|h\rangle|^2>\alpha,\ 
|\langle y|h\rangle|^2>\beta
\right\}
=
\int_{\CP^2\times\mathcal R}
\1_{\{|\langle x|h\rangle|^2>\alpha\}}
\1_{\{|\langle y|h\rangle|^2>\beta\}}
\,\dd\nu(h,\alpha,\beta).
\label{eq:joint-acceptance-nu}
\end{equation}
Using the decomposition of $\nu$ in \cref{eq:single-coordinate-distribution}, this becomes
\begin{align}
\Pr_\nu\left\{
|\langle x|h\rangle|^2>\alpha,\ 
|\langle y|h\rangle|^2>\beta
\right\}
=
&\frac13
\int_{\CP^2\times\mathcal R}
\1_{\{|\langle x|h\rangle|^2>\alpha\}}
\1_{\{|\langle y|h\rangle|^2>\beta\}}
\,\dd(\sigma\otimes\nu_{\mathrm{line}})
\notag\\
&\quad
+
\frac23
\int_{\CP^2\times\mathcal R}
\1_{\{|\langle x|h\rangle|^2>\alpha\}}
\1_{\{|\langle y|h\rangle|^2>\beta\}}
\,\dd(\sigma\otimes\nu_{\mathcal R}).
\label{eq:joint-acceptance-decomposition}
\end{align}
The first term corresponds to $\alpha$ uniformly distributed on $[1/2,1]$ with $\beta=1-\alpha$. Since the corresponding density for $\alpha$ is $2$, while the uniform density on $\mathcal R$ is $4$ because $\operatorname{area}(\mathcal R)=1/4$, we obtain
\begin{align}
\Pr_\nu\left\{
|\langle x|h\rangle|^2>\alpha,\ 
|\langle y|h\rangle|^2>\beta
\right\}
&=
\frac{2}{3}
\int_{\CP^2}
\int_{1/2}^1
\1_{\{|\langle x|h\rangle|^2>\alpha\}}
\1_{\{|\langle y|h\rangle|^2>1-\alpha\}}
\,\dd\alpha\,\dd h
\notag\\
&\quad
+
\frac{8}{3}
\int_{\CP^2}
\int_{\mathcal R}
\1_{\{|\langle x|h\rangle|^2>\alpha\}}
\1_{\{|\langle y|h\rangle|^2>\beta\}}
\,\dd\alpha\,\dd\beta\,\dd h.
\label{eq:single-joint-acceptance-integral}
\end{align}
Here $\dd h$ denotes normalized Haar measure. The following lemma evaluates this integral.

\begin{restatable}{lemma}{jointacceptancequtrit}
\label{lem:qutrit-joint-acceptance}
For every $x,y\in\CP^2$,
\begin{equation}
\frac{2}{3}
\int_{\CP^2}
\int_{1/2}^1
\1_{\{|\langle x|h\rangle|^2>\alpha\}}
\1_{\{|\langle y|h\rangle|^2>1-\alpha\}}
\,\dd\alpha\,\dd h
+
\frac{8}{3}
\int_{\CP^2}
\int_{\mathcal R}
\1_{\{|\langle x|h\rangle|^2>\alpha\}}
\1_{\{|\langle y|h\rangle|^2>\beta\}}
\,\dd\alpha\,\dd\beta\,\dd h
=
\frac{|\langle x|y\rangle|^2}{18}.
\label{eq:single-joint-acceptance}
\end{equation}
\end{restatable}

The proof of \cref{lem:qutrit-joint-acceptance} is given in \cref{app:haar-integrals}. It follows from a Haar integral identity specific to $\CP^2$ together with the two components of the threshold distribution in \cref{eq:single-coordinate-distribution}.

We can now determine Alice's success probability from the same lemma. Taking $y=x$, the two local overlaps coincide. Moreover, $\alpha\geq\beta$ almost surely under $\nu$, so whenever Alice's test succeeds,
$$
|\langle x|h\rangle|^2>\alpha,
$$
Bob's test also succeeds. Hence
\begin{equation}
\Pr_\nu\left\{
|\langle x|h\rangle|^2>\alpha
\right\}
= \Pr_\nu\left\{
|\langle x|h\rangle|^2>\alpha,\ 
|\langle x|h\rangle|^2>\beta
\right\} =
\frac1{18}.
\label{eq:single-alice-acceptance}
\end{equation}
In particular, Alice's success probability is independent of her input $x$.

For arbitrary $x,y\in\CP^2$, combining \cref{eq:single-joint-acceptance,eq:single-alice-acceptance} gives
\begin{equation}
\Pr_\nu\left\{
|\langle y|h\rangle|^2>\beta
\;\middle|\;
|\langle x|h\rangle|^2>\alpha
\right\}
=
|\langle x|y\rangle|^2.
\label{eq:single-conditional-acceptance}
\end{equation}
Thus, conditioned on Alice's test succeeding, Bob's local test reproduces exactly the quantum probability
$$
p_Q(0\mid x,y)=|\langle x|y\rangle|^2.
$$

Equation \cref{eq:single-conditional-acceptance} already contains the essential ingredient of the simulation. If Alice and Bob could sample $\theta$ according to $\nu$ conditioned on the event
$$
|\langle x|h\rangle|^2>\alpha,
$$
Bob could reproduce the $b=0$ outcome probability using only his local test. The difficulty is that this conditioning depends on Alice's input $x$, whereas the shared randomness must be sampled before the inputs are received. Moreover, an unconditioned sample satisfies Alice's test only with probability $1/18$.

In the next subsection, we overcome this by letting Alice and Bob share independent samples from $\nu$. Alice selects the first sample for which her test succeeds and communicates its index to Bob. This implements the required conditioning through rejection sampling and turns the construction above into an exact classical simulation protocol.

\subsubsection{From postselection to a protocol with average finite communication}
\label{sec:from-postselection-to-average-communication}

We now turn the postselected construction above into a classical simulation protocol. The basic idea is to let Alice and Bob share an infinite sequence of independent samples from $\nu$, with Alice scanning them until her local test succeeds. She then communicates the index of the successful sample to Bob, who applies his local test to that same sample. The selected sample therefore has the conditional statistics identified in the previous subsection, under which Bob's local test reproduces the quantum probability $|\langle x,y\rangle|^2$.

More precisely, let
\begin{equation}
\Lambda
=
\bigl(\CP^2\times\mathcal R\bigr)^{\mathbb N},
\qquad
\lambda=(\theta_1,\theta_2,\ldots),
\qquad
\theta_k=(h_k,\alpha_k,\beta_k),
\label{eq:average-proposals}
\end{equation}
where the variables $\theta_k$ are independent and identically distributed according to $\nu$. We equip $\Lambda$ with the corresponding product probability measure
\begin{equation}
\mu:=\nu^{\otimes\mathbb N}.
\end{equation}
The sequence $\lambda$ is generated independently of the inputs and made available to both Alice and Bob as shared randomness.

After receiving $x$, Alice examines the coordinates $\theta_1,\theta_2,\ldots$ in order. For each $k$, she computes
$$
a_k:=|\langle x|h_k\rangle|^2
$$
and checks whether $a_k>\alpha_k$. If such a coordinate exists, she stops at the first one and communicates its index
\begin{equation}
c(x,\lambda)
:=
\min\left\{
k\geq1:
a_k>\alpha_k
\right\}
\label{eq:first-accepted-index}
\end{equation}
to Bob. If no such coordinate exists, we set $c=1$ by convention.

Upon receiving $c$ and his input $y$, Bob computes
\begin{equation}
b_c:=|\langle y|h_c\rangle|^2
\end{equation}
and outputs
\begin{equation}
b=
\begin{cases}
0, & b_c>\beta_c,\\
1, & b_c\leq\beta_c.
\end{cases}
\label{eq:average-bob-output}
\end{equation}

The exceptional case in which Alice never finds a successful coordinate occurs with probability zero. Indeed, by \cref{eq:single-alice-acceptance}, each coordinate independently satisfies Alice's test with probability $1/18$. Hence, for every $N\geq1$,
\begin{equation}
\Pr_\mu\left\{
a_k\leq\alpha_k
\text{ for all }k=1,\ldots,N
\right\}
=
\left(\frac{17}{18}\right)^N.
\end{equation}
Therefore,
\begin{equation}
\Pr_\mu\left\{
a_k\leq\alpha_k
\text{ for every }k\geq1
\right\}
=
\lim_{N\to\infty}
\left(\frac{17}{18}\right)^N
=
0.
\end{equation}
Thus Alice finds a successful coordinate almost surely, and the convention $c=1$ in the exceptional case does not affect the output statistics.

The rules above define a deterministic encoder and decoder of the same form as in \cref{def:classical-simulation}, except that at this stage the message alphabet is the infinite set $\mathbb N$. More explicitly, Alice's encoder is
\begin{equation}
p_A(c\mid x,\lambda)
=
\1_{\{c=c(x,\lambda)\}},
\qquad c\in\mathbb N,
\end{equation}
where $c(x,\lambda)$ is defined in \cref{eq:first-accepted-index}.\footnote{We use the same symbol $c$ for a possible message value and for the deterministic response function $c(x,\lambda)$ that selects the message for fixed $x$ and $\lambda$. This slight abuse of notation is adopted for readability.} On the other hand, Bob's decoder is determined by
\begin{equation}
p_B(0\mid y,c,\lambda)
=
\1_{\{b_c>\beta_c\}},
\qquad
p_B(1\mid y,c,\lambda)
=
\1_{\{b_c\leq\beta_c\}}.
\end{equation}
Thus the remaining task is to verify that these response functions reproduce the quantum probabilities after averaging over the shared randomness $\lambda$.

For the outcome $b=0$, the classical probability obtained from these response functions is
\begin{align}
p_C(0\mid x,y)
&=
\int_\Lambda
\sum_{k=1}^{\infty}
p_A(k\mid x,\lambda)
p_B(0\mid y,k,\lambda)
\,\dd\mu(\lambda)
\notag\\
&=
\int_\Lambda
\sum_{k=1}^{\infty}
\1_{\{c(x,\lambda)=k\}}
\1_{\{b_k>\beta_k\}}
\,\dd\mu(\lambda)
\notag\\
&=
\sum_{k=1}^{\infty}
\Pr_\mu\{c=k,\ b_k>\beta_k\}
\notag\\
&=
\sum_{k=1}^{\infty}
\Pr_\mu\{c=k\}
\Pr_\mu\{b_k>\beta_k\mid c=k\}.
\label{eq:average-classical-decomposition}
\end{align}
The interchange of the sum and the integral is justified by Tonelli's theorem, since all terms are nonnegative.

To evaluate the conditional probability in \cref{eq:average-classical-decomposition}, consider the event that the $k$th coordinate is the first successful one,
\begin{equation}
\mathsf E_k
:=
\bigcap_{\ell=1}^{k-1}
\{a_\ell\leq\alpha_\ell\}
\cap
\{a_k>\alpha_k\}.
\end{equation}
Since the exceptional event in which no coordinate succeeds is $\mu$-null, the events $\{c=k\}$ and $\mathsf E_k$ agree up to a $\mu$-null set. The failures of the first $k-1$ coordinates depend only on $\theta_1,\ldots,\theta_{k-1}$ and are therefore independent of $\theta_k$. Hence, conditioned on $\mathsf E_k$, the variable $\theta_k$ has the same distribution as a single sample from $\nu$ conditioned on Alice's local test succeeding. By \cref{eq:single-conditional-acceptance},
\begin{equation}
\Pr_\mu\{b_k>\beta_k\mid c=k\}
= \Pr_\mu\{b_k>\beta_k\mid \mathsf E_k\}=
|\langle x,y\rangle|^2
\end{equation}
for every $k\geq1$.

Substituting this into \cref{eq:average-classical-decomposition} gives
\begin{align}
p_C(0\mid x,y)
&=
|\langle x,y\rangle|^2
\sum_{k=1}^{\infty}
\Pr_\mu\{c=k\}
\notag\\
&=
|\langle x,y\rangle|^2
=
p_Q(0\mid x,y),
\label{eq:average-classical-output}
\end{align}
Hence the protocol reproduces exactly the quantum statistics of pure qutrit states and binary projective measurements.

However, the communicated index is not bounded. Indeed, for every $N\geq1$,
\begin{equation}
\Pr_\mu\{c>N\}
=
\left(\frac{17}{18}\right)^{N}
>
0.
\label{eq:positive-probability-each-message}
\end{equation}
Thus arbitrarily large message values occur with nonzero probability. Consequently, the message alphabet is genuinely infinite, and the protocol does not yet satisfy the finite-message requirement in \cref{def:classical-simulation}. Nevertheless, its communication cost is finite on average, as shown below. In \cref{sec:finite-sampling}, we modify the construction to obtain a finite message alphabet and hence a protocol in the sense of \cref{def:classical-simulation}.

\begin{remark}[Average communication cost]\label{rmk:average-communication-cost}
Since each coordinate independently passes Alice's test with probability $1/18$, the communicated index $c$ follows the geometric distribution
\begin{equation}
\Pr_\mu\{c=k\}
=
\left(\frac{17}{18}\right)^{k-1}
\frac1{18},
\qquad
k\geq1.
\label{eq:geometric-index}
\end{equation}

A simple way to encode the index is to use the unary code
\begin{equation}
1\longmapsto0,
\qquad
2\longmapsto10,
\qquad
3\longmapsto110,
\qquad
4\longmapsto1110,
\qquad
\ldots .
\label{eq:unary-code}
\end{equation}
The message corresponding to $c=k$ then has length $k$ bits. Hence the average amount of bits communicated is given by
\begin{equation}
\mathbb E_\mu[\ell(c)]
=
 \sum_{k=1}^{\infty} \frac{k}{18}\left(\frac{17}{18}\right)^{k-1}
=
18.
\end{equation}

This encoding is not optimal. More efficient encodings, such as Golomb codes \cite{Golomb1966}, can reduce the average communication cost, but we will not pursue this optimization here, our interest is only in the finiteness of the average communication cost.
\end{remark}

\subsection{A finite-communication protocol for binary measurements}
\label{sec:finite-sampling}

The protocol in \cref{sec:average-communication} reproduces the desired quantum statistics with finite communication on average, but the index communicated by Alice is unbounded. We now modify the construction so that Alice always sends a message from a finite alphabet.

The rough idea is to fix a sufficiently large integer $N$ and truncate the previous protocol to the first $N$ shared coordinates, with a few additional selection conditions. If none of these $N$ coordinates is selected, Alice instead sends a message associated with an auxiliary \emph{fallback} procedure. The fallback compensates for the truncation so that the resulting finite protocol still reproduces exactly the quantum statistics.

The fallback relies on a finite covering of $\CP^2$, which we construct first in \cref{sec:finite-covering}. We then define the protocol and analyze its communication cost in \cref{sec:protocol}; its correctness is established in \cref{sec:worst-case-correctness}.

\subsubsection{Finite covering}
\label{sec:finite-covering}

Here we construct a covering of $\CP^2$ by finitely many balls of equal radius; see \cref{fig:worst-case-communication-protocol}(a).

Define the chordal distance between two rays $z,w\in\CP^2$ by
\begin{equation}
d(z,w)
:=
\sqrt{1-|\langle z|w\rangle|^2}.
\label{eq:chordal-distance}
\end{equation}
This defines a metric on $\CP^2$. For $v\in\CP^2$ and $0<r<1$, let
\begin{equation}
\mathcal D_r(v)
:=
\left\{
z\in\CP^2:
d(z,v)<r
\right\}.
\label{eq:general-chordal-ball}
\end{equation}

We will use the following standard bound on the number of such balls required to cover complex projective space; see \cite[Sec.~5.1.1, pp.~107--108, Eq.~(5.2), and Exercise~5.11, p.~112]{AubrunSzarek2017}.

\begin{lemma}[Finite covering, \cite{AubrunSzarek2017}]
\label{lem:finite-cap-cover}
For every $0<r<1$, there exist rays $v_1,\ldots,v_{K_{\mathrm{cover}}}\in\CP^2$ such that
\begin{equation}
\CP^2
=
\bigcup_{j=1}^{K_{\mathrm{cover}}}
\mathcal D_r(v_j),
\end{equation}
where
\begin{equation}
K_{\mathrm{cover}}
\leq
\left(\frac{2}{r}\right)^4.
\label{eq:general-K-bound}
\end{equation}
\end{lemma}

We choose
\begin{equation}
r=\frac1{\sqrt{24}}
\end{equation}
and, for simplicity, write
\begin{equation}
\mathcal D(v)
:=
\mathcal D_{1/\sqrt{24}}(v).
\label{eq:fallback-cap}
\end{equation}
With this choice, \cref{lem:finite-cap-cover} gives
\begin{equation}
K_{\mathrm{cover}}
\leq
(2\sqrt{24})^4
=
9216.
\label{eq:K-bound}
\end{equation}

The radius also determines a useful lower bound on the overlap between two rays belonging to the same covering ball. Indeed, if $x,h\in\mathcal D_r(v)$, the triangle inequality gives
$$
d(x,h)
\leq
d(x,v)+d(v,h)
<
2r,
$$
and therefore
\begin{equation}
|\langle x|h\rangle|^2
=
1-d(x,h)^2
>
1-4r^2.
\label{eq:same-ball-overlap}
\end{equation}
For $r=1/\sqrt{24}$, this yields
\begin{equation}
|\langle x|h\rangle|^2
>
\frac56
\qquad
\text{for all }
x,h\in\mathcal D(v).
\label{eq:fallback-overlap-goal}
\end{equation}

The choice $r=1/\sqrt{24}$ is made for convenience, as it yields the simple overlap bound in \cref{eq:fallback-overlap-goal}. This choice is not optimized: different radii lead to different covering sizes and ultimately to different communication bounds. We return to this point in \cref{rem:binary-protocol-optimization}.

\subsubsection{Protocol definition and communication cost}
\label{sec:protocol}

We now describe the finite-message protocol. We fix
\begin{equation}
N:=112
\label{eq:choice-of-N}
\end{equation}
and introduce the constants
\begin{equation}
\eta:=\frac1{576},
\qquad
\delta
:=
\left(
1-\frac{575}{10368}
\right)^N,
\qquad
\gamma
:=
\frac{1-\eta}{1-\delta}.
\label{eq:finite-protocol-constants}
\end{equation}
The value $N=112$ is chosen as the smallest integer for which $\delta<\eta$, and consequently
$
0<\gamma<1.
$
The meaning of these constants will be established in the proof of correctness. Roughly, $\eta$ will determine the weight of the fallback part of the protocol, $\delta$ the probability that the truncated search fails, and $\gamma$ a correction factor used to adjust their relative weights.

\paragraph{Shared randomness space and distribution.}

Let
\begin{align}
\Lambda
&=
\bigl(\CP^2\times\mathcal R\bigr)^N
\times[0,1]
\times
\prod_{j=1}^{K_{\mathrm{cover}}}\mathcal D(v_j)
\times[3/4,5/6],
\notag\\
\lambda
&=
\bigl(
\theta_1,\ldots,\theta_N,
\Gamma,
h_1^{\mathrm{fb}},\ldots,h_{K_{\mathrm{cover}}}^{\mathrm{fb}},
\alpha_{\mathrm{fb}}
\bigr),
\qquad
\theta_k=(h_k,\alpha_k,\beta_k).
\label{eq:worst-case-shared-randomness}
\end{align}
The components of $\lambda$ are sampled independently as follows:
\begin{itemize}
\item $\theta_1,\ldots,\theta_N$ are independent and identically distributed according to $\nu$, as in \cref{eq:single-coordinate-distribution};
\item $\Gamma\sim\operatorname{Unif}[0,1]$;
\item for each $j\in[K_{\mathrm{cover}}]$, the ray $h_j^{\mathrm{fb}}$ is sampled according to normalized Haar measure on $\mathcal D(v_j)$;
\item $\alpha_{\mathrm{fb}}\sim\operatorname{Unif}[3/4,5/6]$.
\end{itemize}
Their joint distribution defines the product probability measure $\mu$ on $\Lambda$. As before, $\lambda$ is generated independently of the inputs and made available to both parties as shared randomness.

The first $N$ components $\theta_1,\ldots,\theta_N$ are simply a truncation of the sequence used in \cref{sec:average-communication}. The remaining variables will be used to compensate for this truncation. The interval $[3/4,5/6]$ is chosen to simplify the numerical parameters of the construction and is not optimized. We discuss more general choices and their effect on the communication cost in \cref{rem:binary-protocol-optimization}.

\paragraph{Alice's encoder.}

Alice's encoder has two branches: a \emph{truncated} part and a \emph{fallback}. The truncated is based on the average-communication protocol restricted to the first $N$ shared coordinates, with two additional requirements. If these requirements are not met, Alice uses the fallback.

Both branches depend on the covering ball containing $x$. After receiving $x$, Alice therefore determines the first such ball,
\begin{equation}
j_x
:=
\min\left\{
j\in[K_{\mathrm{cover}}]:
x\in\mathcal D(v_j)
\right\}.
\label{eq:cover-selector}
\end{equation}
Such an index always exists by \cref{lem:finite-cap-cover}.

We now describe the truncated branch. Its first additional requirement is
$$
\Gamma\leq\gamma.
$$
When this holds, Alice considers the coordinates $\theta_1,\ldots,\theta_N$. As in the average-communication protocol \cref{eq:AliceTestSucceedsLocal}, a coordinate $\theta_k=(h_k,\alpha_k,\beta_k)$ is first required to satisfy
$$
|\langle x|h_k\rangle|^2>\alpha_k.
$$
It is convenient to denote the set of single-coordinate values satisfying this test by
\begin{equation}
\mathsf A_x
:=
\left\{
(h,\alpha,\beta)\in\CP^2\times\mathcal R:
|\langle x|h\rangle|^2>\alpha
\right\}.
\label{eq:single-coordinate-acceptance-set}
\end{equation}
Thus the original acceptance condition is simply $\theta_k\in\mathsf A_x$.

The truncated part imposes one further criterion: Alice excludes coordinates belonging to
\begin{equation}
\mathsf S_x
:=
\left\{
(h,\alpha,\beta)\in\CP^2\times\mathcal R:
h\in\mathcal D(v_{j_x}),
\quad
\frac34\leq\alpha\leq\frac56,
\quad
\beta=1-\alpha
\right\}.
\label{eq:reserved-slice}
\end{equation}
Accordingly, the coordinates eligible for the truncated part are those belonging to
\begin{equation}
\mathsf O_x
:=
\mathsf A_x\setminus\mathsf S_x.
\label{eq:single-coordinate-ordinary-set}
\end{equation}
The candidate indices are therefore
\begin{equation}
\mathcal K_x(\lambda)
:=
\left\{
k\in[N]:
\theta_k\in\mathsf O_x
\right\}.
\label{eq:ordinary-candidate-set}
\end{equation}

If $\Gamma\leq\gamma$ and $\mathcal K_x(\lambda)$ is nonempty, Alice sends the smallest index in $\mathcal K_x(\lambda)$. Otherwise, she switches to the fallback. In that case, she communicates the index $j_x$ of the covering ball containing $x$, shifted by $N$. More precisely,
\begin{equation}
c(x,\lambda)
=
\begin{cases}
\min\mathcal K_x(\lambda),
&
\Gamma\leq\gamma
\text{ and }
\mathcal K_x(\lambda)\neq\varnothing,
\\[1ex]
N+j_x,
&
\text{otherwise}.
\end{cases}
\label{eq:finite-binary-encoder}
\end{equation}
Hence
$$
c(x,\lambda)\in[N+K_{\mathrm{cover}}].
$$
In the notation of \cref{def:classical-simulation}, Alice's encoder is
\begin{equation}
p_A(c\mid x,\lambda)
=
\1_{\{c=c(x,\lambda)\}},
\qquad
c\in[N+K_{\mathrm{cover}}].
\label{eq:finite-binary-pA}
\end{equation}

Relative to the average-communication protocol, the finite construction therefore modifies the selection rule in three ways: only the first $N$ coordinates are considered, a coordinate satisfying the original acceptance condition $\theta_k\in\mathsf A_x$ is eligible only if $\theta_k\notin\mathsf S_x$, and the truncated part is used only when $\Gamma\leq\gamma$. If either $\Gamma>\gamma$ or no eligible coordinate exists, the protocol switches to the fallback. The role of these modifications in recovering exactly the quantum statistics will be established in the proof of correctness.

\paragraph{Bob's decoder.}

Upon receiving $c$ and his measurement ray $y$, Bob proceeds according to the type of message received.

If $c\leq N$, he uses the coordinate $\theta_c$ and outputs
\begin{equation}\label{eq:Output_Bob_Binary_finite_truncaded}
b=
\begin{cases}
0, & |\langle y|h_c\rangle|^2>\beta_c,\\
1, & |\langle y|h_c\rangle|^2\leq\beta_c.
\end{cases}
\end{equation}
If $c>N$, he sets $j=c-N$, uses the fallback ray $h_j^{\mathrm{fb}}$, and outputs
\begin{equation}\label{eq:Output_Bob_Binary_finite_fallback}
b=
\begin{cases}
0, & |\langle y|h_j^{\mathrm{fb}}\rangle|^2>1-\alpha_{\mathrm{fb}},\\
1, & |\langle y|h_j^{\mathrm{fb}}\rangle|^2\leq1-\alpha_{\mathrm{fb}}.
\end{cases}
\end{equation}

Equivalently, Bob's decoder is
\begin{equation}
p_B(0\mid y,c,\lambda)
=
\begin{cases}
\1_{\{|\langle y|h_c\rangle|^2>\beta_c\}},
&
c\leq N,
\\[1ex]
\1_{\{|\langle y|h_{c-N}^{\mathrm{fb}}\rangle|^2>1-\alpha_{\mathrm{fb}}\}},
&
c>N,
\end{cases}
\label{eq:finite-binary-pB}
\end{equation}
with
\begin{equation}
p_B(1\mid y,c,\lambda)
=
1-p_B(0\mid y,c,\lambda).
\end{equation}

\paragraph{Communication cost.}

Alice always sends a value from the finite alphabet $[N+K_{\mathrm{cover}}]$.
Since $N=112$ and \cref{lem:finite-cap-cover} gives $K_{\mathrm{cover}}\leq9216$,
\begin{equation}
d_C
=
N+K_{\mathrm{cover}}
\leq
112+9216
=
9328
<
2^{14}.
\label{eq:binary-message-bound}
\end{equation}
Thus $14$ classical bits suffice to implement the binary protocol.

\subsubsection{Correctness}
\label{sec:worst-case-correctness}

We now prove that the finite protocol exactly reproduces the quantum probabilities. More precisely, for every $x,y\in\CP^2$, we want to show that
\begin{equation}
p_C(0\mid x,y)
=
|\langle x|y\rangle|^2.
\label{eq:finite-binary-correctness-goal}
\end{equation}
Here, by definition of the classical protocol,
\begin{equation}
p_C(0\mid x,y)
=
\int_\Lambda
\sum_{c=1}^{N+K_{\mathrm{cover}}}
p_A(c\mid x,\lambda)
p_B(0\mid y,c,\lambda)
\,\dd\mu(\lambda).
\label{eq:finite-binary-classical-probability}
\end{equation}
Recalling that messages $c\leq N$ correspond to the truncated part of the protocol and messages $c>N$ to the fallback, it is useful to decompose this probability as
\begin{align}
p_C(0\mid x,y)
&=
\Pr_\mu\{c\leq N\}
\Pr_\mu\{b=0\mid c\leq N\}
\notag\\
&\quad+
\Pr_\mu\{c>N\}
\Pr_\mu\{b=0\mid c>N\}.
\label{eq:finite-binary-output-decomposition}
\end{align}

The construction is based on the identity established in \cref{eq:single-conditional-acceptance}. In terms of the acceptance set $\mathsf A_x$ introduced here, this identity reads
\begin{equation}
\Pr_\nu\left\{
|\langle y|h\rangle|^2>\beta
\,\middle|\,
\mathsf A_x
\right\}
=
|\langle x|y\rangle|^2.
\label{eq:finite-coordinate-quantum-statistics}
\end{equation}
Thus, as in the average-communication protocol, it suffices to reproduce through the communicated message the conditional probability above.

The finite protocol does this by decomposing the acceptance set $\mathsf A_x$ into the two parts $\mathsf O_x$ and $\mathsf S_x$ in such a way that the corresponding conditional probabilities can be associated with the truncated and fallback parts of the finite protocol. We will then show that these two contributions, together with their respective weights, coincide exactly with the decomposition of $p_C(0\mid x,y)$ in \cref{eq:finite-binary-output-decomposition}.

\paragraph{Decomposition of the target probability.}

We begin by observing that
\begin{equation}
\mathsf S_x
\subseteq
\mathsf A_x.
\label{eq:reserved-event-always-accepted}
\end{equation}
Indeed, if $\theta=(h,\alpha,\beta)\in\mathsf S_x$, then both $x$ and $h$ belong to $\mathcal D(v_{j_x})$. Hence \cref{eq:fallback-overlap-goal} gives
\begin{equation}
|\langle x|h\rangle|^2
>
\frac56.
\end{equation}
Since $\alpha\leq5/6$ on $\mathsf S_x$, it follows that $|\langle x,h\rangle|^2>\alpha$, and therefore $\theta\in\mathsf A_x$.

We next determine the relative weight of $\mathsf S_x$ inside $\mathsf A_x$. The measure $\nu_{\mathcal R}$ is uniform on $\mathcal R$ and therefore assigns zero probability to the boundary $\beta=1-\alpha$. Hence only the $\nu_{\mathrm{line}}$ component of \cref{eq:single-coordinate-distribution} contributes to $\mathsf S_x$.

A chordal ball of radius $1/\sqrt{24}$ in $\CP^2$ has normalized Haar measure
\begin{equation}
\sigma\bigl(\mathcal D(v_{j_x})\bigr)
=
\left(\frac1{\sqrt{24}}\right)^4
=
\frac1{576},
\end{equation}
Moreover, under $\nu_{\mathrm{line}}$, the variable $\alpha$ is uniformly distributed on $[1/2,1]$ and therefore has density $2$. Thus
\begin{align}
\Pr_\nu(\mathsf S_x)
&=
\frac13
\int_{\mathcal D(v_{j_x})}
\dd\sigma(h)
\int_{3/4}^{5/6}
2\,\dd\alpha
\notag\\
&=
\frac13
\cdot
\frac1{576}
\cdot
2\left(\frac56-\frac34\right)
\notag\\
&=
\frac1{10368}.
\label{eq:reserved-slice-probability}
\end{align}
Recall from \cref{eq:single-alice-acceptance} that
\begin{equation}
\Pr_\nu(\mathsf A_x)
=
\frac1{18}.
\label{eq:finite-coordinate-acceptance}
\end{equation}
Since $\mathsf S_x\subseteq\mathsf A_x$, we obtain from \cref{eq:reserved-slice-probability,eq:finite-coordinate-acceptance}
\begin{align}
\Pr_\nu(\mathsf S_x\mid\mathsf A_x)
=
\frac{\Pr_\nu(\mathsf S_x)}
{\Pr_\nu(\mathsf A_x)}=
\frac{1/10368}{1/18}
=
\frac1{576}
=
\eta.
\label{eq:value-of-eta}
\end{align}

Recall from \cref{eq:single-coordinate-ordinary-set} that $\mathsf O_x := \mathsf A_x\setminus\mathsf S_x.$ Since $\mathsf S_x\subseteq\mathsf A_x$, the set $\mathsf A_x$ therefore decomposes as the disjoint union
\begin{equation}
\mathsf A_x
=
\mathsf O_x
\mathbin{\dot\cup}
\mathsf S_x.
\end{equation}
Consequently,
\begin{align}
\Pr_\nu(\mathsf O_x\mid\mathsf A_x)
&=
1-\Pr_\nu(\mathsf S_x\mid\mathsf A_x)
=
1-\eta.
\label{eq:ordinary-conditional-weight}
\end{align}

Applying the law of total probability to the partition $\mathsf A_x
=
\mathsf O_x
\mathbin{\dot\cup}
\mathsf S_x,
$
we obtain
\begin{align}
\Pr_\nu\left\{
|\langle y|h\rangle|^2>\beta
\,\middle|\,
\mathsf A_x
\right\}
&=
\Pr_\nu(\mathsf O_x\mid\mathsf A_x)
\Pr_\nu\left\{
|\langle y|h\rangle|^2>\beta
\,\middle|\,
\mathsf O_x
\right\}
\notag\\
&\quad+
\Pr_\nu(\mathsf S_x\mid\mathsf A_x)
\Pr_\nu\left\{
|\langle y|h\rangle|^2>\beta
\,\middle|\,
\mathsf S_x
\right\}.
\label{eq:conditional-response-decomposition}
\end{align}
Combining the decomposition above with \cref{eq:finite-coordinate-quantum-statistics,eq:value-of-eta,eq:ordinary-conditional-weight}, we obtain
\begin{align}
|\langle x|y\rangle|^2
&=
(1-\eta)
\Pr_\nu\left\{
|\langle y|h\rangle|^2>\beta
\,\middle|\,
\mathsf O_x
\right\}
+
\eta
\Pr_\nu\left\{
|\langle y|h\rangle|^2>\beta
\,\middle|\,
\mathsf S_x
\right\}.
\label{eq:target-ordinary-fallback-decomposition}
\end{align}

The two terms in \cref{eq:target-ordinary-fallback-decomposition} naturally correspond to the two parts of the finite protocol: the conditional probability associated with $\mathsf O_x$ will be reproduced by the truncated part, while the conditional probability associated with $\mathsf S_x$ will be reproduced by the fallback. More precisely, we will show that
\begin{equation}
\Pr_\mu\{c\leq N\}=1-\eta,
\qquad
\Pr_\mu\{c>N\}=\eta,
\label{eq:desired-part-weights}
\end{equation}
and that
\begin{align}
\Pr_\mu\{b=0\mid c\leq N\}
&=
\Pr_\nu\left\{
|\langle y|h\rangle|^2>\beta
\,\middle|\,
\mathsf O_x
\right\},
\label{eq:desired-truncated-statistics}\\
\Pr_\mu\{b=0\mid c>N\}
&=
\Pr_\nu\left\{
|\langle y|h\rangle|^2>\beta
\,\middle|\,
\mathsf S_x
\right\}.
\label{eq:desired-fallback-statistics}
\end{align}
Together with \cref{eq:finite-binary-output-decomposition} and \cref{eq:target-ordinary-fallback-decomposition}, these identities will establish the desired equality
$$
p_C(0\mid x,y)=|\langle x|y\rangle|^2.
$$

\paragraph{Weights of the truncated and fallback parts.}

We first prove \cref{eq:desired-part-weights}. By \cref{eq:finite-binary-encoder}, Alice sends a message $c\leq N$ precisely when $\mathcal K_x(\lambda)\neq\varnothing$ and $\Gamma\leq\gamma$. Since $\Gamma$ is independent of $\theta_1, \ldots,\theta_N$ and uniformly distributed on $[0,1]$,
\begin{align}
\Pr_\mu\{c\leq N\}
&=
\Pr_\mu\left\{
\mathcal K_x(\lambda)\neq\varnothing,
\ \Gamma\leq\gamma
\right\}
\notag\\
&=
\gamma\,
\Pr_\mu\left\{
\mathcal K_x(\lambda)\neq\varnothing
\right\}.
\label{eq:ordinary-message-intermediate}
\end{align}

By \cref{eq:ordinary-candidate-set}, the set $\mathcal K_x(\lambda)$ is empty precisely when none of the first $N$ coordinates belongs to $\mathsf O_x$. Since $\theta_1,\ldots,\theta_N$ are independent and identically distributed according to $\nu$,
\begin{align}
\Pr_\mu\left\{
\mathcal K_x(\lambda)=\varnothing
\right\}
&=
\left(
1-\Pr_\nu(\mathsf O_x)
\right)^N.
\label{eq:truncated-failure-preliminary}
\end{align}

It remains to compute $\Pr_\nu(\mathsf O_x)$. Since $\mathsf O_x\subseteq\mathsf A_x$, \cref{eq:finite-coordinate-acceptance,eq:ordinary-conditional-weight} gives
\begin{align}
\Pr_\nu(\mathsf O_x)
&=
\Pr_\nu(\mathsf A_x)
\Pr_\nu(\mathsf O_x\mid\mathsf A_x)
\notag\\
&=
\frac1{18}(1-\eta)
\notag\\
&=
\frac{575}{10368}.
\label{eq:ordinary-success-probability}
\end{align}
Therefore,
\begin{align}
\Pr_\mu\left\{
\mathcal K_x(\lambda)=\varnothing
\right\}
&=
\left(
1-\frac{575}{10368}
\right)^N
\notag\\
&=
\delta.
\label{eq:truncated-failure-probability}
\end{align}

Substituting this into \cref{eq:ordinary-message-intermediate}, we obtain
\begin{align}
\Pr_\mu\{c\leq N\}
&=
\gamma(1-\delta)
\notag\\
&=
1-\eta,
\label{eq:ordinary-message-weight}
\end{align}
where the last equality follows from the definition of $\gamma$ in \cref{eq:finite-protocol-constants}. Consequently,
\begin{equation}
\Pr_\mu\{c>N\}
=
\eta.
\label{eq:binary-fallback-probability}
\end{equation}
This proves \cref{eq:desired-part-weights}.

This calculation also explains the choices of $N$ and $\gamma$, \cref{eq:choice-of-N,eq:finite-protocol-constants}. For $\gamma$ to define a valid probability, we require $\gamma\leq1$, which is equivalent to $\delta\leq\eta$. We chose $N=112$ as the smallest positive integer for which $\delta<\eta$, and $\gamma$ then adjusts the probability $(1-\delta)$ of finding an eligible coordinate to exactly the required weight $(1-\eta)$. 

\paragraph{Output statistics in the truncated part.}

We next prove \cref{eq:desired-truncated-statistics}. From Bob's output rule in \cref{eq:Output_Bob_Binary_finite_truncaded}, for each $k\in[N]$,
\begin{equation}
\Pr_\mu\{b=0\mid c=k\}
=
\Pr_\mu\left\{
|\langle y|h_k\rangle|^2>\beta_k
\,\middle|\,
c=k
\right\}.
\label{eq:ordinary-response-start}
\end{equation}
By \cref{eq:finite-binary-encoder}, Alice sends $c=k$ precisely when $\theta_k$ belongs to $\mathsf O_x$, none of the preceding coordinates $\theta_l$ does, and $\Gamma\leq\gamma$. Hence
\begin{equation}
\{c=k\}
=
\{\Gamma\leq\gamma\}
\cap
\{\theta_k\in\mathsf O_x\}
\cap
\bigcap_{\ell=1}^{k-1}
\{\theta_\ell\notin\mathsf O_x\}.
\label{eq:ordinary-message-k}
\end{equation}
The event
$$
\left\{
|\langle y|h_k\rangle|^2>\beta_k
\right\}
$$
depends only on $\theta_k$. It is therefore independent of $\{\Gamma\leq\gamma\}$ and of all the conditions involving the preceding coordinates $\theta_1,\ldots,\theta_{k-1}$. Among the conditions defining $\{c=k\}$, only $\theta_k\in\mathsf O_x$ affects its conditional probability. Hence
\begin{align}
\Pr_\mu\{b=0\mid c=k\}
&=
\Pr_\mu\left\{
|\langle y|h_k\rangle|^2>\beta_k
\,\middle|\,
\theta_k\in\mathsf O_x
\right\}
\notag\\
&=
\Pr_\nu\left\{
|\langle y|h\rangle|^2>\beta
\,\middle|\,
\mathsf O_x
\right\}.
\label{eq:ordinary-response-given-message}
\end{align}
The last equality follows because $\theta_k$ is distributed according to $\nu$. In particular, the right-hand side is the same for every $k\in[N]$. Consequently,
\begin{align}
\Pr_\mu\{b=0\mid c\leq N\}
&=
\sum_{k=1}^N
\Pr_\mu\{b=0\mid c=k\}
\Pr_\mu\{c=k\mid c\leq N\}
\notag\\
&=
\Pr_\nu\left\{
|\langle y|h\rangle|^2>\beta
\,\middle|\,
\mathsf O_x
\right\}
\sum_{k=1}^N
\Pr_\mu\{c=k\mid c\leq N\}
\notag\\
&=
\Pr_\nu\left\{
|\langle y|h\rangle|^2>\beta
\,\middle|\,
\mathsf O_x
\right\}.
\label{eq:ordinary-output-distribution}
\end{align}
This proves \cref{eq:desired-truncated-statistics}.

\paragraph{Output statistics in the fallback.}

It remains to prove \cref{eq:desired-fallback-statistics}. From Bob's fallback output rule in \cref{eq:Output_Bob_Binary_finite_fallback},
\begin{equation}
\Pr_\mu\{b=0\mid c>N\}
=
\Pr_\mu\left\{
|\langle y|h_{j_x}^{\mathrm{fb}}\rangle|^2
>
1-\alpha_{\mathrm{fb}}
\,\middle|\,
c>N
\right\}.
\label{eq:fallback-response-start}
\end{equation}
By \cref{eq:finite-binary-encoder}, the event $\{c>N\}$ occurs whenever either $\Gamma>\gamma$ or no coordinate is eligible for the truncated part. Hence
\begin{equation}
\{c>N\}
=
\{\Gamma>\gamma\}
\cup
\bigcap_{k=1}^N
\{\theta_k\notin\mathsf O_x\}.
\label{eq:fallback-message-event}
\end{equation}
This event depends only on $\Gamma$ and the coordinates $\theta_1,\ldots,\theta_N$, and is therefore independent of the fallback variables $h_{j_x}^{\mathrm{fb}}$ and $\alpha_{\mathrm{fb}}$. Consequently,
\begin{equation}
\Pr_\mu\{b=0\mid c>N\}
=
\Pr_\mu\left\{
|\langle y|h_{j_x}^{\mathrm{fb}}\rangle|^2
>
1-\alpha_{\mathrm{fb}}
\right\}.
\label{eq:fallback-response-unconditioned}
\end{equation}

We now compare this probability with the conditional distribution of a single coordinate $\theta=(h,\alpha,\beta)$ given $\mathsf S_x$. Since $\mathsf S_x$ imposes the condition $\beta=1-\alpha$, and the component $\nu_{\mathcal R}$ assigns zero probability to this boundary, conditioning on $\mathsf S_x$ selects the $\nu_{\mathrm{line}}$ component of $\nu$. Within this component, the conditions defining $\mathsf S_x$ further restrict
$$
h\in\mathcal D(v_{j_x}),
\qquad
\alpha\in[3/4,5/6],
\qquad
\beta=1-\alpha.
$$
Since $h$ and $\alpha$ are independent under $\nu_{\mathrm{line}}$, their conditional distributions are
$$
h\sim\sigma(\,\cdot\mid\mathcal D(v_{j_x})),
\qquad
\alpha\sim\operatorname{Unif}[3/4,5/6],
$$
with $\beta=1-\alpha$. This is precisely the joint distribution used by the fallback,
$$
h_{j_x}^{\mathrm{fb}}
\sim
\sigma(\,\cdot\mid\mathcal D(v_{j_x})),
\qquad
\alpha_{\mathrm{fb}}
\sim
\operatorname{Unif}[3/4,5/6],
$$
with the two variables independent. Therefore,
\begin{align}
\Pr_\mu\left\{
|\langle y|h_{j_x}^{\mathrm{fb}}\rangle|^2
>
1-\alpha_{\mathrm{fb}}
\right\}
&=
\Pr_\nu\left\{
|\langle y|h\rangle|^2>\beta
\,\middle|\,
\mathsf S_x
\right\}.
\label{eq:fallback-distribution-identification}
\end{align}
Combining \cref{eq:fallback-response-unconditioned,eq:fallback-distribution-identification}, we obtain
\begin{equation}
\Pr_\mu\{b=0\mid c>N\}
=
\Pr_\nu\left\{
|\langle y|h\rangle|^2>\beta
\,\middle|\,
\mathsf S_x
\right\},
\label{eq:fallback-output-distribution}
\end{equation}
which proves \cref{eq:desired-fallback-statistics}.

This completes the proof. Indeed, combining \cref{eq:finite-binary-output-decomposition,eq:target-ordinary-fallback-decomposition,eq:ordinary-message-weight,eq:binary-fallback-probability,eq:ordinary-output-distribution,eq:fallback-output-distribution}, we obtain
\begin{align}
p_C(0\mid x,y)
&=
|\langle x|y\rangle|^2
=
p_Q(0\mid x,y).
\label{eq:finite-binary-exactness}
\end{align}
The probability of the $b=1$ outcome agrees by normalization.

Hence pure qutrit states and binary projective measurements admit an exact classical simulation with
\begin{equation}
d_C
\leq
9328
<
2^{14}.
\end{equation}

\begin{remark}[Optimization of the communication cost]
\label{rem:binary-protocol-optimization}
The numerical parameters used above are not optimized. In particular, the covering radius $1/\sqrt{24}$ and the interval $[3/4,5/6]$ defining $\mathsf S_x$ were chosen to keep the construction and its constants simple.

More generally, suppose that the covering balls have radius $r$ and that the interval in the definition of $\mathsf S_x$ is replaced by $[L,U]\subseteq[1/2,1]$. If $x$ and $h$ belong to the same covering ball, then
$$
|\langle x|h\rangle|^2>1-4r^2.
$$
Thus, to guarantee $\mathsf S_x\subseteq\mathsf A_x$, it is sufficient to require
\begin{equation}
U\leq1-4r^2.
\end{equation}
Since we also require a nontrivial interval with $L\geq1/2$, this construction requires
\begin{equation}
r<\frac1{\sqrt8}.
\end{equation}
Increasing $r$ reduces the number of balls required to cover $\CP^2$, but at the same time decreases the range of admissible values for $U$.

For our choice $r=1/\sqrt{24}$, the bound above gives $U\leq5/6$, explaining the upper endpoint of the interval. The lower endpoint $L=3/4$, on the other hand, is chosen for numerical convenience: since a chordal ball of radius $r$ in $\CP^2$ has Haar measure $r^4$, the corresponding fallback weight is
\begin{equation}
\eta
=
\Pr_\nu(\mathsf S_x\mid\mathsf A_x)
=
12r^4(U-L).
\end{equation}
With $[L,U]=[3/4,5/6]$, this simplifies to $\eta=r^4$, which is also equal to volume of the balls.

Different admissible choices of $r$, $L$, and $U$ therefore change both the number of covering balls and the value of $\eta$. The latter determines how large $N$ must be in order for
\begin{equation}
\left(
1-\frac{1-\eta}{18}
\right)^N
<
\eta.
\end{equation}
Consequently, improving the communication bound amounts to balancing the size of the covering against the truncation length $N$. We do not pursue this optimization here.
\end{remark}

\subsection{Extension to any POVM}
\label{sec:d3-more-outcomes}

We now extend the finite protocol for binary measurements defined in \cref{sec:finite-sampling} to arbitrary qutrit POVMs $M = \{ M_b \}_{b \in B}$ with any finite number $\abs{B}$ of outcomes.

The general idea is straightforward.
Instead of starting from a general POVM, assume that the states are pure and $\operatorname{rank}(M_b) = 1$ for every $b$ (later, we will recover the general case by coarse graining the effects and decomposing mixed states as a convex combination of pure states).
For each effect, we can define a binary measurement as described below,
and use the simulation protocol for binary measurements to simulate $p_Q(b^\prime \mid \rho, M^{(b)})$.
The protocol below provides a way to recombine these simulations, so that they correctly reproduce $p_Q(b \mid \rho, M)$.
This recombination, however, requires some care to preserve normalization, for reasons that will be discussed in \cref{sec:povm-binary-weights}. After that, in \cref{sec:d3-povm-protocol} we define the protocol, and in \cref{sec:d3-povm-correctness} prove its correctness.

\subsubsection{Binarization}
\label{sec:povm-binary-weights}

Let $x\in\CP^2$, and $M=\{M_b\}_{b\in B}$ be a finite-outcome qutrit POVM with rank-$1$ effects,
\begin{equation}
	M_b = \gamma_b\dyad{y_b}, \qquad 0<\gamma_b\leq1, \qquad \sum_{b\in B}\gamma_b\dyad{y_b} = \eye_3,
	\label{eq:povm-rank-one-decomposition}
\end{equation}
where $y_b\in\CP^2$.
The upper bound on $\gamma_b$ follows from $M_b\leq\eye_3$, while taking the trace of the last equality gives
\begin{equation}
	\sum_{b\in B}\gamma_b = 3.
	\label{eq:povm-sum-gamma}
\end{equation}
For each effect $M_b$, we define a binary POVM
\begin{equation}
	M^{(b)} = \{ \dyad{y_b} ,\eye_3- \dyad{y_b} \}.
\end{equation}

Consider the simulation protocol for binary measurements from \cref{sec:protocol}.
Let $(\Lambda_{\mathrm{bin}},\mu_{\mathrm{bin}})$ be the shared randomness space of that protocol, and $c(x, \lambda)$ be Alice's message for the state $x$ for a given value of the shared variable $\lambda \in \Lambda_{\mathrm{bin}}$ (\cref{eq:finite-binary-encoder}).
Let Bob's response function be as in the binary protocol (\cref{eq:Output_Bob_Binary_finite_fallback}):
\begin{equation}
	p_B^{\mathrm{bin}}(0\mid c(x, \lambda), y_b, \lambda)
	=
	\begin{cases}
		\1_{\{\abs{\langle y_b|h_c\rangle}^2>\beta_c\}}, & c\leq N,\\[1ex]
		\1_{\{ |\langle y_b|h_{c-N}^{\mathrm{fb}}\rangle |^2>1-\alpha_{\mathrm{fb}}\}}, & c>N .
	\end{cases}
\label{eq:bob-det-response-binary}
\end{equation}
Correctness of the binary protocol gives
\begin{equation}
	\int_{\Lambda_{\mathrm{bin}}}p_B^{\mathrm{bin}}(0\mid c(x, \lambda),y_b,\lambda)\,\dd\mu_{\mathrm{bin}}(\lambda) = |\langle x|y_b\rangle |^2
	\label{eq:povm-binary-response}
\end{equation}
thus the outcome associated to the effect $M_b = \gamma_b \dyad{y_b}$ of the original POVM $M$ can be simulated by using the binary protocol and keeping the result $b = 0$ of the binary protocol with probability $\gamma_b$.

Let us try to combine the binary protocols in the following way.
For each $b\in B$, define
\begin{equation}
	w_b(c(x, \lambda), \lambda) := \gamma_b p_B^{\mathrm{bin}}(0\mid c(x, \lambda), y_b,\lambda), \qquad W(c,\lambda) := \sum_{b\in B}w_b(c,\lambda).
	\label{eq:povm-scores}
\end{equation}
It follows from \cref{eq:povm-binary-response,eq:povm-rank-one-decomposition} that
\begin{align}
	\int_{\Lambda_{\mathrm{bin}}}w_b(c(x, \lambda),\lambda)\,\dd\mu_{\mathrm{bin}}(\lambda) &= \gamma_b\abs{\langle x|y_b\rangle}^2 = \langle x|M_b|x\rangle .
	\label{eq:povm-score-average}
\end{align}
Thus, if the weights $w_b(c(x, \lambda), \lambda)$ formed a probability distribution over $b$ for every fixed $(c(x, \lambda), \lambda)$, Bob could output $b$ with probability $w_b$ and \cref{eq:povm-score-average} shows this would reproduce the POVM statistics.
However, this is not generally true, since $W(c(x, \lambda),\lambda)$ is not necessarily equal to $1$.
On the other hand, note that that the average of $W$ over $\lambda$ is such that
\begin{equation}
	\int_{\Lambda_{\mathrm{bin}}}W(c(x, \lambda),\lambda)\,\dd\mu_{\mathrm{bin}}(\lambda) = 1.
    \label{eq:povm-score-average-normalization}
\end{equation}
This fact will become useful later.

\subsubsection{Shared randomness and the protocol}
\label{sec:d3-povm-protocol}

We now define the multi-outcome protocol.
The validity of the sampling procedures defined in this section and the correctness of the protocol will be discussed in \cref{sec:d3-povm-correctness}.

It is instructive to first discuss the reasoning behind the definition, which is not too different from the ideas previously used when moving from the average-case binary protocol to the finite binary protocol.
Recall from the above section that, from one shared variable of the binary protocol, Bob can use the weights $w_b$ to evaluate the binary response for every ray $y_b$, and this would in principle yield the correct averages $\langle x|M_b|x\rangle$.
Since these weights do not form a probability distribution over the outcomes, Bob will instead use modified versions of them as sub-normalized probabilities: for each shared variable, he either selects an outcome $b$ or rejects that shared variable.
Because rejection can occur, Alice and Bob would need to sample several independent copies of the binary shared variable and Bob test them sequentially until one produces an outcome.
However, for any finite number of copies there will remain a nonzero probability that every copy is rejected, and since the number of copies will ultimately incur in additional communication cost, we must truncate them to some finite $n$.
To treat the case where Bob rejects all copies of the binary random variables, we introduce a new fallback, called the \emph{multi-outcome fallback}, which must produce an outcome whenever none of the binary copies does.
These two branches are then combined through a correction described below, so that the complete protocol reproduces exactly the desired POVM probabilities.

\paragraph{Shared randomness.}

Let $n$ be the number of copies of the binary protocol used to construct the multi-outcome protocol. For reasons that will later become clear, we choose $n := 26$.

Before receiving their inputs, Alice and Bob sample $n$ independent shared variables $\lambda^{(1)},\ldots,\lambda^{(n)}$, each from $\mu_{\mathrm{bin}}$.
They also sample the multi-outcome fallback, which uses the same distribution as in the binary case, namely, rays $h_j^{\mathrm{mo}}$ according to the normalized Haar measure on each covering ball $\mathcal D(v_j)$, and a threshold $\alpha_{\mathrm{mo}}\sim\operatorname{Unif}[3/4,5/6]$.

The shared randomness space then is
\begin{align}
	\Lambda &= \Lambda_{\mathrm{bin}}^n\times\prod_{j=1}^{K_{\mathrm{cover}}}\mathcal D(v_j)\times[3/4,5/6], \notag\\
	\lambda &= \bigl(\lambda^{(1)},\ldots,\lambda^{(n)},h_1^{\mathrm{mo}},\ldots,h_{K_{\mathrm{cover}}}^{\mathrm{mo}},\alpha_{\mathrm{mo}}\bigr) ,
	\label{eq:povm-shared-randomness}
\end{align}
Their joint distribution defines the product probability measure $\mu$ on $\Lambda$.

\paragraph{Alice's encoder.}
After receiving $x$, Alice computes
\begin{align}
	j_x &:= \min\{j \mid x\in\mathcal D(v_j)\} \\ 
	c_k &:= c(x, \lambda^{(k)}), \quad \text{for } k \in \{1, \ldots,n \} .
\end{align}
She communicates to Bob the message
\begin{equation}
	c^{\mathrm{mo}} = (j_x,c_1,\ldots,c_n).
	\label{eq:povm-composite-message}
\end{equation}

Note the similarity with the binary protocol: each message $c_k$ has the same encoding as the binary protocol, and $j_x$ is defined in the same way as the binary fallback.

\paragraph{Bob's decoder.}
Upon receiving $c^{\mathrm{mo}}$ and $M$, Bob does his decoding procedure according to two branches, which we will call \emph{ordinary} and \emph{fallback}.

First, in the ordinary branch he considers each component $c_k$ of $c^{\mathrm{mo}}$ by using sub-normalized probabilities $\hat{w}_b$ built from the $w_b$ (\cref{eq:povm-scores}). Using $c$ to denote a generic component, we define
\begin{equation}
	\widehat w_b(c,\lambda) := \frac{1}{1-\varepsilon} \left( w_b(c,\lambda)-\frac{\varepsilon}{\eta}r_b^{\mathrm{bin,fb}}(c,\lambda) \right),
    \qquad
	r_b^{\mathrm{bin,fb}}(c,\lambda) :=
	\begin{cases}
		0, & c\leq N,\\
		\frac{w_b(c,\lambda)}{W(c,\lambda)}, & c>N .
	\end{cases}
	\label{eq:povm-corrected-weights}
\end{equation}
Here, $\eta = 1 / 576$ stands for the probability that the binary protocol uses its own fallback \cref{eq:value-of-eta}, while $\varepsilon:=\left( 3/4 \right)^n$, for reasons explained later, is the probability that all $n$ copies of the binary shared variables are rejected.

Before continuing with the description of the protocol, let us explain some pieces of this definition.
The idea behind subtracting $( \varepsilon / \eta )r_b^{\mathrm{bin,fb}}$ from $\hat{w}_b$ is to remove part of the contribution of the fallback branch that is internal to the binary protocol.
Indeed, by the definition of $r_b^{\mathrm{bin,fb}}$, nothing is removed in the ordinary branch (i.e., when $c \leq N$).
This is necessary because the new multi-outcome fallback, define shortly, will use the same type of fallback construction and therefore supplies this contribution again whenever all ordinary trials fail.
Since the internal binary fallback occurs with probability $\eta$, multiplying its normalized contribution by $\varepsilon/\eta$ gives it average weight $\varepsilon$, matching the probability with which the multi-outcome fallback is used.
Lastly, the overall factor $1/(1-\varepsilon)$ rescales the remaining weights to account for the fact that the ordinary branch produces an outcome with total probability $1-\varepsilon$.

With this in hand, Bob considers each $\lambda^{(k)}$ in order for $k \in \{1, \ldots, n \}$, and samples $Y_k\in B\cup\{\perp\}$ according to
\begin{equation}
	\label{eq:povm-ordinary-rule}
\begin{aligned}
	\Pr\{Y_k=b\mid c_k,\lambda^{(k)},M\} &= \frac{\widehat w_b(c_k,\lambda^{(k)})}{4}, \\
	\Pr\{Y_k=\,\perp\, \mid c_k,\lambda^{(k)},M\} &= 1-\frac{ \sum_{b \in B} \hat{w}_b(c_k, \lambda^{(k)})}{4}.
\end{aligned}
\end{equation}
The additional event $\perp$ means that no outcome $b$ was identified, and in this case he proceeds to the next $k$.
Otherwise, if he obtains some $Y_k\in B$, he outputs $Y_k$.
It will be shown in the next section that this is a well-defined probability distribution.

When all $n$ trials result in $Y_k = \,\perp$, he uses the multi-outcome fallback.
In this branch, he uses the multi-outcome fallback coordinate $(h_{j_x}^{\mathrm{mo}}, \alpha_{\mathrm{mo}}, 1 - \alpha_{\mathrm{mo}})$ to compute
\begin{equation}
	w_b^{\mathrm{mo}} := \gamma_b\1_{\{ | \langle y_b|h_{j_x}^{\mathrm{mo}}\rangle|^2>1-\alpha_{\mathrm{mo}}\}}, \qquad W^{\mathrm{mo}} := \sum_{b\in B}w_b^{\mathrm{mo}},
	\label{eq:povm-fallback-scores}
\end{equation}
and then outputs $b$ with probability
\begin{equation}
	\pi_b^{\mathrm{mo}} = \frac{w_b^{\mathrm{mo}}}{W^{\mathrm{mo}}},
\label{eq:multioutcome-fallback-probabilities}
\end{equation}
This is similar to the fallback branch of the binary protocol, but with the difference that Bob evaluates this response for every POVM effect, multiplies it by the corresponding coefficient $\gamma_b$, and normalizes the resulting weights to select a single outcome $b$.
We will soon show that $\pi_b^{\mathrm{mo}}$ forms a probability distribution.

Combining the ordinary fallback branches, the probability that the multi-outcome protocol outputs $b$ is therefore
\begin{equation}
\begin{aligned}
	p_C(b\mid x,M)
	&=
	\sum_{k=1}^n\Pr\{Y_1=\cdots=Y_{k-1}=\,\perp,\ Y_k=b\mid x,M\}\\
    &\qquad+\Pr\{Y_1=\cdots=Y_n=\,\perp\mid x,M\}\int_\Lambda\pi_b^{\mathrm{mo}}\,\dd\mu.
	\label{eq:povm-output-decomposition}
\end{aligned}
\end{equation}

\paragraph{Communication cost.}
Alice sends the message $c^{\mathrm{mo}} = (j_x,c_1,\ldots,c_n)$, where each $c_k$ is equivalent to a message in the binary protocol, and $j_x$ selects a ball in the covering as the fallback.
Thus each $c_k$ has at most $N+K_{\mathrm{cover}}\leq9328$ possible values, while $j_x$ has at most $K_{\mathrm{cover}}\leq9216$ possible values.
Therefore
\begin{equation}
	d_C \leq K_{\mathrm{cover}}\bigl(N+K_{\mathrm{cover}}\bigr)^{26} \leq9216(9328)^{26}<2^{357} ,
	\label{eq:povm-message-bound}
\end{equation}
so the multi-outcome protocol uses at most $357$ one-way classical bits.

\subsubsection{Validity and correctness}
\label{sec:d3-povm-correctness}

We must establish the validity and correctness of the above protocol.
To establish validity, we must show that \cref{eq:povm-ordinary-rule,eq:multioutcome-fallback-probabilities} are well-defined probability distributions, since these are the two sampling procedures used in the decoding.
To establish correctness, the output probability of the protocol, given by \cref{eq:povm-output-decomposition}, must satisfy
\begin{equation}
	p_C(b\mid x,M)
	=
	\langle x|M_b|x\rangle
\end{equation}
for every $b\in B$.

\paragraph{Validity.}
We first establish that $\pi_b^{\mathrm{mo}}$ in \cref{eq:multioutcome-fallback-probabilities} is a probability distribution.
Recall that in a multi-outcome fallback coordinate, $\beta_{\mathrm{mo}} = 1-\alpha_{\mathrm{mo}}$.
For brevity, define the set of outcomes that pass the fallback test in \cref{eq:povm-fallback-scores} by
$\mathrm A_{\mathrm{mo}}:=\{b\in B \,:\, |\langle y_b|h_{j_x}^{\mathrm{mo}}\rangle|^2>\beta_{\mathrm{mo}}\}$.
Then
\begin{equation}
	W^{\mathrm{mo}}
	=
	\sum_{b\in\mathrm A_{\mathrm{mo}}}
	\gamma_b.
\end{equation}
Using \cref{eq:povm-rank-one-decomposition,eq:povm-sum-gamma} and splitting the sum according to $\mathrm{A}_{\mathrm{mo}}$, we obtain
\begin{align}
	1
	&=
	\sum_{b\in B} \gamma_b |\langle y_b|h_{j_x}^{\mathrm{mo}}\rangle|^2
	\leq
	\sum_{b\in\mathrm A_{\mathrm{mo}}} \gamma_b + \beta_{\mathrm{mo}} \sum_{b\notin\mathrm A_{\mathrm{mo}}} \gamma_b
	=
	W^{\mathrm{mo}} + \beta_{\mathrm{mo}}\bigl(3-W^{\mathrm{mo}}\bigr).
	\label{eq:povm-safe-estimate}
\end{align}
Since $\alpha_{\mathrm{mo}}\in[3/4,5/6]$, we have
$\beta_{\mathrm{mo}}\in[1/6,1/4]$, and hence
\begin{equation}
	W^{\mathrm{mo}}
	\geq
	\frac{1-3\beta_{\mathrm{mo}}}{1-\beta_{\mathrm{mo}}}
	\geq
	\frac13.
\label{eq:w-mo-positive}
\end{equation}
The denominator defining $\pi_b^{\mathrm{mo}}$ is therefore strictly positive, so it is true that $\pi_b^{\mathrm{mo}} \geq 0$ for every $b$ and $\sum_b \pi_b^{\mathrm{mo}} = 1$.

\medskip
We will now establish that \cref{eq:povm-ordinary-rule} also defines a probability distribution.
This distribution is defined in terms of $\widehat{w}_b(c, \lambda)$ in \cref{eq:povm-corrected-weights}, and it suffices to show that $\widehat{w}_b(c, \lambda) \geq 0$ and $\sum_b \widehat{w}_b(c, \lambda) \leq 4$, for every $c$ and $\lambda$.

Consider first the nonnegativity of $\widehat{w}_b$ for $c \leq N$.
From \cref{eq:povm-corrected-weights},
\begin{equation}
	\widehat w_b(c,\lambda) = \frac{w_b(c,\lambda)}{1 - \varepsilon}, \quad c \leq N .
\end{equation}
Recall from \cref{sec:d3-povm-protocol} that $\varepsilon = (3 / 4)^{26}$, thus $1 - \varepsilon > 0$.
By \cref{eq:bob-det-response-binary,eq:povm-scores}, the numerator is
\begin{equation}
    w_b(c, \lambda) = \gamma_b \1_{\{\abs{\langle y_b|h_c\rangle}^2>\beta_c\}}, \quad c\leq N
\end{equation}
which is also nonnegative for any $c$ and $\lambda$. 

Now consider the case $c > N$, for which
\begin{equation}
	\widehat w_b(c,\lambda) =
    \frac{w_b(c, \lambda)}{1 - \varepsilon}
    \left( 1 - \frac{\varepsilon}{\eta \, \sum_{b^\prime} w_{b^\prime}(c, \lambda)} \right)
    , \quad c > N .
\end{equation}
To prove its nonnegativity we must show that, for $c > N$ and any $\lambda$, we have $w_b(c, \lambda) \geq 0$ and $\sum_b w_b(c, \lambda) \geq \varepsilon / \eta$.
By \cref{eq:bob-det-response-binary,eq:povm-scores},
\begin{equation}
	w_b(c,\lambda)
	=
	\gamma_b
	\1_{\{ |\langle y_b|h_{c-N}^{\mathrm{fb}}\rangle|^2 > \beta_{\mathrm{fb}}\}}, \quad c > N ,
\end{equation}
thus $w_b(c, \lambda) \geq 0$.
Recall that the thresholds $\alpha_{\mathrm{fb}}$ and the rays $h^{\mathrm{fb}}_j$, used in the internal fallback of the binary protocol, are distributed exactly as $h^{\mathrm{mo}}_j$ and $\alpha_{\mathrm{mo}}$.
Therefore the same argument used to establish \cref{eq:w-mo-positive} applies in this case, giving
\begin{equation}
	\sum_{b \in B} w_b(c, \lambda) \geq \frac13,
	\quad
	c>N.
\end{equation}
Since $\varepsilon / \eta < 1/3$, this concludes the argument that $\widehat{w}_b(c, \lambda) \geq 0$.

\smallskip
It is only left to show that $\sum_b \widehat{w}_b(c, \lambda) \leq 4$.
By \cref{eq:povm-scores,eq:povm-sum-gamma} and the fact that
$p_B^{\mathrm{bin}}(0\mid c(x, \lambda) ,y_b,\lambda)\in\{0,1\}$, we have
\begin{equation}
	\sum_{b\in B}w_b(c,\lambda)
	=
	\sum_{b\in B}
	\gamma_b p_B^{\mathrm{bin}}(0\mid c(x, \lambda), y_b,\lambda)
	\leq
	\sum_{b\in B}\gamma_b
	=
	3.
\end{equation}
Summing the expressions for $\widehat w_b(c,\lambda)$ established above
gives
\begin{equation}
	\sum_{b\in B}\widehat w_b(c,\lambda)
	=
	\begin{cases}
		\dfrac{\sum_{b\in B}w_b(c,\lambda)}{1-\varepsilon},
		& c\leq N,\\[3ex]
		\dfrac{\sum_{b\in B}w_b(c,\lambda)-\varepsilon/\eta}
		{1-\varepsilon},
		& c>N.
	\end{cases}
\label{eq:sum-wb-hat-cases}
\end{equation}
In both cases,
\begin{equation}
	\sum_{b\in B}\widehat w_b(c,\lambda)
	\leq
	\frac{3}{1-\varepsilon}
	<
	4,
\end{equation}
where the last inequality follows from
$\varepsilon=(3/4)^{26}<1/4$.

\smallskip
This concludes the proof that both probability distributions used in the decoding procedure of the multi-outcome protocol are well-defined.

\paragraph{Correctness.}
Recall that Bob first considers the ordinary outcomes $Y_1,\ldots,Y_n$ in \cref{eq:povm-ordinary-rule} and outputs the first one that belongs to $B$.
If all these outcomes equal $\perp$, he instead samples an outcome from the fallback $\pi^{\mathrm{mo}}$ in \cref{eq:multioutcome-fallback-probabilities}.
The probabilities of his outputs are then given by \cref{eq:povm-output-decomposition}, whose terms we now analyze.

\smallskip
Let us first determine the probability that each ordinary trial returns $\perp$.
By \cref{eq:povm-ordinary-rule}, for every $k\in\{1,\ldots,n\}$ this probability is
\begin{equation}
	\Pr\{Y_k=\,\perp\,\mid x,M\}
	=
	1-\frac14\int_{\Lambda_{\mathrm{bin}}}\sum_{b\in B}\widehat w_b(c(x, \lambda),\lambda)\,\dd\mu_{\mathrm{bin}}(\lambda).
\end{equation}
The two expressions for $\sum_{b\in B}\widehat w_b(c,\lambda)$ established in \cref{eq:sum-wb-hat-cases} can be combined as
\begin{equation}
	\sum_{b\in B}\widehat w_b(c,\lambda)
	=
	\frac{1}{1-\varepsilon}\left(\sum_{b\in B}w_b(c,\lambda)-\frac{\varepsilon}{\eta}\1_{\{c>N\}}\right).
\end{equation}
We thus have
\begin{equation}
\begin{aligned}
	\int_{\Lambda_{\mathrm{bin}}}\sum_{b\in B}\widehat w_b(c(x, \lambda),\lambda)\,\dd\mu_{\mathrm{bin}}(\lambda)
	&=
	\frac{1}{1-\varepsilon}\int_{\Lambda_{\mathrm{bin}}}\left(\sum_{b\in B}w_b(c(x, \lambda),\lambda)-\frac{\varepsilon}{\eta}\1_{\{c(x, \lambda)>N\}}\right)\,\dd\mu_{\mathrm{bin}}(\lambda)
	\\
	&=
	\frac{1}{1-\varepsilon}\left(\int_{\Lambda_{\mathrm{bin}}}\sum_{b\in B}w_b(c(x, \lambda),\lambda)\,\dd\mu_{\mathrm{bin}}(\lambda)-\frac{\varepsilon}{\eta}\Pr_{\mu_{\mathrm{bin}}}\{c(x, \lambda)>N\}\right)
	\\
	&=
	1.
	\label{eq:povm-corrected-average}
\end{aligned}
\end{equation}
where we have used \cref{eq:povm-score-average} to evaluate the first integral and \cref{eq:binary-fallback-probability} to evaluate $\Pr_{\mu_{\mathrm{bin}}}\{c(x, \lambda)>N\}=\eta$.
Therefore we obtain
\begin{equation}
\begin{aligned}
	\Pr\{Y_k=\,\perp\,\mid x,M\}
	&=
	1-\frac14
	= \frac34, \\
	\Pr\{Y_k=b\mid x,M\}
	&=
	\frac14\int_{\Lambda_{\mathrm{bin}}}\widehat w_b(c(x, \lambda),\lambda)\,\dd\mu_{\mathrm{bin}}(\lambda).
	\label{eq:povm-one-trial-probabilities}
\end{aligned}
\end{equation}
where in the first expression we have used \cref{eq:povm-corrected-average}.
The variables $Y_1,\ldots,Y_n$ are independent because the binary shared variables and Bob's corresponding local random variables are independent.
Thus all $n$ trials fail with probability
\begin{equation}
	\Pr\{Y_1=\cdots=Y_n= \,\perp\,\}
	=
	\left(\frac34\right)^n
	=
	\varepsilon.
	\label{eq:povm-all-trials-fail}
\end{equation}
With this, we may now rewrite the first term in Bob's output probability (\cref{eq:povm-output-decomposition}), namely,
\begin{equation}
    \sum_{k=1}^n\Pr\{Y_1=\cdots=Y_{k-1}=\perp,\ Y_k=b\mid x,M\} .
\end{equation}
By independence of the trials and \cref{eq:povm-one-trial-probabilities}, for every $k\in\{ 1, \ldots, n \}$,
\begin{equation}
\Pr{Y_1=\cdots=Y_{k-1}=\,\perp, Y_k=b\mid x,M} = \left(\frac34\right)^{k-1}\frac14\int_{\Lambda_{\mathrm{bin}}}\widehat w_b(c_x(\lambda),\lambda),\dd\mu_{\mathrm{bin}}(\lambda).
\end{equation}
Summing over $k$, and using $(1/4) \sum_{k=1}^n(3/4)^{k-1}=1-\varepsilon$, we obtain
\begin{equation}
\sum_{k=1}^n\Pr{Y_1=\cdots=Y_{k-1}=\perp,\ Y_k=b\mid x,M} = (1-\varepsilon)\int_{\Lambda_{\mathrm{bin}}}\widehat w_b(c_x(\lambda),\lambda),\dd\mu_{\mathrm{bin}}(\lambda).
\label{eq:first-term-in-bobs-povm-output}
\end{equation}

\smallskip
We now consider the second term in Bob's output distribution.
For this, we must determine the probability that the multi-outcome fallback returns $b$, namely $\int_\Lambda\pi_b^{\mathrm{mo}}\,\dd\mu$.
By definition,
\begin{equation}
\pi_b^{\mathrm{mo}} = \frac{\gamma_b\1_{{ | \langle y_b|h_{j_x}^{\mathrm{mo}}\rangle |^2>1-\alpha_{\mathrm{mo}}}}}{\sum_{b'\in B}\gamma_{b'}\1_{{| \langle y_{b'}|h_{j_x}^{\mathrm{mo}}\rangle |^2>1-\alpha_{\mathrm{mo}}}}}.
\end{equation}
Also by definition, the coordinate $(h_{j_x}^{\mathrm{mo}}, \alpha_{\mathrm{mo}}, 1 - \alpha_{\mathrm{mo}})$ has the same distribution as the binary fallback coordinate $(h_{j_x}^{\mathrm{fb}},\alpha_{\mathrm{fb}}, 1 - \alpha_{\mathrm{fb}})$,
and evaluating the binary scores at the fallback message $N+j_x$ gives
\begin{equation}
\frac{w_b(N+j_x,\lambda)}{W(N+j_x,\lambda)} = \frac{\gamma_b\1_{{ | \langle y_b|h_{j_x}^{\mathrm{fb}}\rangle|^2>1-\alpha_{\mathrm{fb}}}}}{\sum_{b'\in B}\gamma_{b'}\1_{{ |\langle y_{b'}|h_{j_x}^{\mathrm{fb}}\rangle|^2>1-\alpha_{\mathrm{fb}}}}}.
\end{equation}
Consequently,
\begin{equation}
\int_\Lambda\pi_b^{\mathrm{mo}} \dd\mu = \int_{\Lambda_{\mathrm{bin}}}\frac{w_b(N+j_x,\lambda)}{W(N+j_x,\lambda)} \,\dd\mu_{\mathrm{bin}}(\lambda).
\label{eq:integral-gamma-with-binary-measure}
\end{equation}
Recall that the event $c(x,\lambda)>N$ has probability $\eta$, that on this event $c(x,\lambda)=N+j_x$, and that the event is independent of the binary fallback variables. Since $r_b^{\mathrm{bin,fb}}$ (\cref{eq:povm-corrected-weights}) vanishes outside this event, its definition can equivalently be written as
\begin{equation}
    r_b^{\mathrm{bin,fb}}(c(x,\lambda),\lambda) = \1_{\{c(x,\lambda)>N\}}\frac{w_b(N+j_x,\lambda)}{W(N+j_x,\lambda)}.
\end{equation}
Integrating this identity and using independence gives
\begin{equation}
	\int_\Lambda\pi_b^{\mathrm{mo}}\,\dd\mu
    = \frac1\eta \int_{\Lambda_{\mathrm{bin}}}r_b^{\mathrm{bin,fb}}(c(x, \lambda),\lambda)\,\dd\mu_{\mathrm{bin}}(\lambda).
	\label{eq:povm-fallback-conditional-law}
\end{equation}

Substituting \cref{eq:first-term-in-bobs-povm-output,eq:povm-fallback-conditional-law} into \cref{eq:povm-output-decomposition} gives
\begin{equation}
\begin{aligned}
	p_C(b\mid x,M)
	&=
	(1-\varepsilon)\int_{\Lambda_{\mathrm{bin}}}\widehat w_b(c(x, \lambda),\lambda)\,\dd\mu_{\mathrm{bin}}(\lambda)+\varepsilon\int_\Lambda\pi_b^{\mathrm{mo}}\,\dd\mu
	\\
	&=
	(1-\varepsilon)\int_{\Lambda_{\mathrm{bin}}}\widehat w_b(c(x, \lambda),\lambda)\,\dd\mu_{\mathrm{bin}}(\lambda)+\frac{\varepsilon}{\eta}\int_{\Lambda_{\mathrm{bin}}}r_b^{\mathrm{bin,fb}}(c(x, \lambda),\lambda)\,\dd\mu_{\mathrm{bin}}(\lambda)
	\\
	&=
	\int_{\Lambda_{\mathrm{bin}}}w_b(c(x, \lambda),\lambda)\,\dd\mu_{\mathrm{bin}}(\lambda)
	\\
	&=
	\langle x|M_b|x\rangle.
	\label{eq:povm-exactness}
\end{aligned}
\end{equation}
The third equality follows by integrating \cref{eq:povm-corrected-weights}, and the last follows from the earlier fact that the binarization reproduces the desired probabilities on average (\cref{eq:povm-score-average}) .
Thus the protocol exactly reproduces the quantum outcome distribution for pure states and rank-$1$ POVMs.

\subsubsection{Larger ranks and mixed states}
\label{sec:d3-povm-generalization}

Generalizing this to effects of larger ranks can be done with the standard observation that a rank-$1$ POVM can be coarse-grained into larger ranks. Likewise, the generalization to mixed states follows directly from decomposing a mixed state as a convex combination of pure states.

\smallskip
In more detail, let $M=\{M_b\}_{b\in B}$ be an arbitrary qutrit POVM with a finite number of outcomes.
For each $b$, use its spectral decomposition to write
\begin{equation}
	M_b = \sum_{k=1}^{r_b}\widetilde M_{b,k}, \qquad \widetilde M_{b,k} = \gamma_{b,k}\dyad{y_{b,k}}.
	\label{eq:povm-rank-one-refinement}
\end{equation}
Here $r_b:=\operatorname{rank}(M_b)$, $0<\gamma_{b,k}\leq1$, and $y_{b,k}\in\CP^2$.
Since $\sum_{b,k}\widetilde M_{b,k}=\sum_bM_b=\eye_3$, the effects $\{\widetilde M_{b,k}\}_{b,k}$ form a finite-outcome rank-$1$ POVM.
Bob applies the rank-$1$ protocol to this refined POVM and locally maps each outcome $(b,k)$ to $b$.
Then
\begin{equation}
	p_C(b\mid x,M) = \sum_{k=1}^{r_b}p_C((b,k)\mid x,\widetilde M) = \sum_{k=1}^{r_b}\langle x|\widetilde M_{b,k}|x\rangle = \langle x|M_b|x\rangle.
	\label{eq:povm-coarse-graining}
\end{equation}

For the generalization to mixed states, write $\rho=\sum_{j=1}^3p_j\dyad{x_j}$.
Alice samples $j$ locally with probability $p_j$ and applies the protocol to $\ket{x_j}$.
Then
\begin{equation}
	p_C(b\mid\rho,M) = \sum_{j=1}^3p_jp_C(b\mid x_j,M) = \tr(\rho M_b).
\end{equation}

Neither of these steps requires additional communication.
This completes the proof of \cref{thm:d3-finite}.

%% file: 5_appendix_technicalities.tex
\begingroup

\subsection{The two Haar integrals}
\label{app:haar-integrals}

Throughout this subsection, $x,y\in\CP^2$ are fixed rays, represented by
unit vectors in $\mathbb C^3$. For $z\in\CP^d$, let
\[
    P_zv:=\langle z| v \rangle z, \quad v \in\CP^d .
\]
All integrals over $\CP^d$ are taken with respect to the normalized
$U(d+1)$-invariant probability measure $\dd h$. For $h\in\CP^2$, set
\begin{equation}
    a:=\abs{\langle x|h\rangle}^2,
    \qquad
    b:=\abs{\langle y|h\rangle}^2,
    \qquad
    \kappa:=(a+b-1)_+,
    \label{eq:kappa-definition}
\end{equation}
where $r_+:=\max\{r,0\}$, and define the region
\begin{equation}
    \mathcal R
    :=
    \left\{
        (\alpha,\beta)\in[0,1]^2:
        \alpha\geq\beta,\quad
        \alpha+\beta\geq1
    \right\}.
    \label{eq:R-definition}
\end{equation}

Our main goal in this section is to establish the following identity:
\begin{align}
\abs{\langle x|y\rangle}^2=
12
\int_{\CP^2}
\int_{1/2}^1
\1_{\{a>\alpha\}}
\1_{\{b>1-\alpha\}}
\,\dd\alpha\,\dd h
+
48
\int_{\CP^2}
\int_{\mathcal R}
\1_{\{a>\alpha\}}
\1_{\{b>\beta\}}
\,\dd\alpha\,\dd\beta\,\dd h
.
\label{eq:two-haar-integrals}
\end{align}
This identity can be computed directly using the simplex representation of the invariant measure and an explicit diagonalization of \(P_x+P_y-I\). Here, however, we give a more general argument that yields the corresponding moment formula on \(\CP^d\). To this end, we first compute the moments of \(\kappa\).

\begin{lemma}
\label{lem:positive-hinge-moments-general-d}
Let $d,n\geq1$ be integers and let $x,y\in\CP^d$. For
\[
    \kappa(h)
    :=
    \left(
        \abs{\langle x|h\rangle}^2
        +
        \abs{\langle y|h\rangle}^2
        -1
    \right)_+,
\]
set
\[
    s:=\abs{\langle x|y\rangle}.
\]
Then
\begin{equation}
    \int_{\CP^d}\kappa(h)^n\,\dd h
    =
    \frac{d!\,n!}{2(n+d)!}
    \frac{s^{n+d-1}}{(1+s)^{d-1}}.
    \label{eq:positive-hinge-moments-general-d}
\end{equation}
\end{lemma}

\begin{proof}
If $s=0$, then $x\perp y$, and Bessel's inequality gives
\[
    \abs{\langle x|h\rangle}^2
    +
    \abs{\langle y|h\rangle}^2
    \leq1.
\]
Hence $\kappa=0$, and \eqref{eq:positive-hinge-moments-general-d}
follows. Assume henceforth that $s>0$.

Consider
\begin{equation}
    H:=P_x+P_y-I_{d+1}.
    \label{eq:general-d-hinge-operator}
\end{equation}
After multiplying one representative by a phase, we may assume
$\langle x|y\rangle=s$. The operator $P_x+P_y$ has eigenvalues
\[
    1+s,\qquad 1-s,\qquad
    \underbrace{0,\ldots,0}_{d-1\text{ times}},
\]
and therefore
\begin{equation}
    \operatorname{spec}(H)
    =
    \left\{
        s,-s,
        \underbrace{-1,\ldots,-1}_{d-1\text{ times}}
    \right\}.
    \label{eq:general-d-hinge-spectrum}
\end{equation}

Let $e_1,\ldots,e_{d+1}$ be an orthonormal eigenbasis corresponding
to \eqref{eq:general-d-hinge-spectrum}, and set
\[
    X_j:=\abs{\langle e_j|h\rangle}^2.
\]
Then
\[
    X_j\geq0,
    \qquad
    \sum_{j=1}^{d+1}X_j=1.
\]
For $h$ distributed according to the normalized invariant measure on
$\CP^d$, the vector
\[
    (X_1,\ldots,X_{d+1})
\]
has the Dirichlet distribution with parameters $(1,\ldots,1)$.
Equivalently, $(X_2,\ldots,X_{d+1})$ has constant density $d!$ on the simplex
\[
    \Delta_d
    :=
    \left\{
        (X_2,\ldots,X_{d+1})\in[0,\infty)^d:
        \sum_{j=2}^{d+1}X_j\leq1
    \right\}.
\]
For further details on integrals of this type and their simplex representation, see \cite{GallaySerre2012}. Using $X_1=1-\sum_{j=2}^{d+1}X_j$, we obtain
\begin{equation}
\begin{aligned}
    \langle h|Hh\rangle
    &=
    sX_1-sX_2-\sum_{j=3}^{d+1}X_j =
    s-2sX_2-(1+s)\sum_{j=3}^{d+1}X_j.
\end{aligned}
\label{eq:general-d-hinge-simplex-form}
\end{equation}
Thus $\kappa>0$ precisely when
\begin{equation}
    2sX_2
    +(1+s)\sum_{j=3}^{d+1}X_j
    <s.
    \label{eq:positive-simplex-region}
\end{equation}
Since $2s\leq1+s$, \eqref{eq:positive-simplex-region} implies
\[
    2s\sum_{j=2}^{d+1}X_j<s,
\]
and hence
\[
    \sum_{j=2}^{d+1}X_j<\frac12.
\]
Thus the constraint defining $\Delta_d$ is automatic on the region
\eqref{eq:positive-simplex-region}. Make the change of variables
\begin{equation}
    Y_2:=2X_2,
    \qquad
    Y_j:=\frac{1+s}{s}X_j,
    \qquad 3\leq j\leq d+1.
    \label{eq:hinge-change-variables}
\end{equation}
Then
\[
    Y_j\geq0,
    \qquad
    \sum_{j=2}^{d+1}Y_j<1,
\]
and
\[
    \kappa
    =
    s\left(
        1-\sum_{j=2}^{d+1}Y_j
    \right).
\]
Moreover,
\[
    \dd X_2\cdots\dd X_{d+1}
    =
    \frac12
    \left(\frac{s}{1+s}\right)^{d-1}
    \dd Y_2\cdots\dd Y_{d+1}.
\]
Therefore
\begin{align}
    \int_{\CP^d}\kappa^n\,\dd h
    &=
    \frac{d!}{2}
    \left(\frac{s}{1+s}\right)^{d-1}
    s^n
    \int_{\substack{Y_j\geq0\\
                    \sum_{j=2}^{d+1}Y_j\leq1}}
    \left(
        1-\sum_{j=2}^{d+1}Y_j
    \right)^n
    \dd Y_2\cdots\dd Y_{d+1}.
    \label{eq:hinge-simplex-integral}
\end{align}
The standard simplex integral gives
\begin{equation}
    \int_{\substack{Y_j\geq0\\
                    \sum_{j=2}^{d+1}Y_j\leq1}}
    \left(
        1-\sum_{j=2}^{d+1}Y_j
    \right)^n
    \dd Y_2\cdots\dd Y_{d+1}
    =
    \frac{n!}{(n+d)!}.
    \label{eq:standard-simplex-beta}
\end{equation}
Substitution into \eqref{eq:hinge-simplex-integral} yields
\[
    \int_{\CP^d}\kappa^n\,\dd h
    =
    \frac{d!\,n!}{2(n+d)!}
    \frac{s^{n+d-1}}{(1+s)^{d-1}}.
\]
\end{proof}

For $d=2$, \cref{lem:positive-hinge-moments-general-d} gives
\begin{equation}
    \int_{\CP^2}\kappa\,\dd h
    =
    \frac{s^2}{6(1+s)},
    \qquad
    \int_{\CP^2}\kappa^2\,\dd h
    =
    \frac{s^3}{12(1+s)},
    \qquad
    s=\abs{\langle x|y\rangle}.
    \label{eq:qutrit-hinge-moments}
\end{equation}
Consequently,
\begin{equation}
   \abs{\langle x|y\rangle}^2= s^2= 6\int_{\CP^2}
    \bigl(\kappa+2\kappa^2\bigr)\,\dd h
      .
    \label{eq:hinge-identity}
\end{equation}

It remains to express the two terms in
\eqref{eq:hinge-identity} in terms of the integrals in
\eqref{eq:two-haar-integrals}. For $\alpha,\beta\in[0,1]$, define
\begin{equation}
    G(\alpha,\beta)
    :=
    \int_{\CP^2}
    \1_{\{a>\alpha\}}
    \1_{\{b>\beta\}}
    \,\dd h.
    \label{eq:G-definition}
\end{equation}
By unitary invariance, we have
\begin{equation}
    G(\alpha,\beta)=G(\beta,\alpha).
    \label{eq:threshold-swap}
\end{equation}

For each fixed $h$,
\begin{equation}
    \kappa=\kappa(h)   =
    \int_0^1
    \1_{\{a>\alpha\}}
    \1_{\{b>1-\alpha\}}
    \,\dd\alpha.
    \label{eq:layer-one}
\end{equation}
Hence, by \eqref{eq:threshold-swap},
\begin{align}
    \int_{\CP^2}\kappa\,\dd h
    &=
    \int_0^1 G(\alpha,1-\alpha)\,\dd\alpha
    \notag\\
    &=
    2
    \int_{1/2}^1
    G(\alpha,1-\alpha)\,\dd\alpha
    \notag\\
    &=
    2
    \int_{\CP^2}
    \int_{1/2}^1
    \1_{\{a>\alpha\}}
    \1_{\{b>1-\alpha\}}
    \,\dd\alpha\,\dd h.
    \label{eq:oriented-layer-one}
\end{align}

Similarly, squaring \eqref{eq:layer-one} and using symmetry in the two
integration variables gives
\begin{align}
    \kappa^2
    &=
    2
    \int_{0\leq\sigma\leq\tau\leq1}
    \1_{\{a>\tau\}}
    \1_{\{b>1-\sigma\}}
    \,\dd\sigma\,\dd\tau=
    2
    \int_{\mathcal T}
    \1_{\{a>\alpha\}}
    \1_{\{b>\beta\}}
    \,\dd\alpha\,\dd\beta,
    \label{eq:layer-two}
\end{align}
where
\[
    \mathcal T
    :=
    \left\{
        (\alpha,\beta)\in[0,1]^2:
        \alpha+\beta\geq1
    \right\}.
\]
Since, up to a set of measure zero,
\[
    \mathcal T
    =
    \mathcal R
    \cup
    \left\{
        (\alpha,\beta)\in\mathcal T:
        \beta\geq\alpha
    \right\},
\]
and $G(\alpha,\beta)=G(\beta,\alpha)$, we obtain
\begin{align}
    \int_{\CP^2}\kappa^2\,\dd h
    &=
    4
    \int_{\mathcal R}
    G(\alpha,\beta)\,\dd\alpha\,\dd\beta
  =
    4
    \int_{\CP^2}
    \int_{\mathcal R}
    \1_{\{a>\alpha\}}
    \1_{\{b>\beta\}}
    \,\dd\alpha\,\dd\beta\,\dd h.
    \label{eq:oriented-layer-two}
\end{align}
Substituting \eqref{eq:oriented-layer-one} and
\eqref{eq:oriented-layer-two} into \eqref{eq:hinge-identity} we obtain
\begin{align}
    \abs{\langle x|y\rangle}^2
    &=
    12
    \int_{\CP^2}
    \int_{1/2}^1
    \1_{\{a>\alpha\}}
    \1_{\{b>1-\alpha\}}
    \,\dd\alpha\,\dd h    +
    48
    \int_{\CP^2}
    \int_{\mathcal R}
    \1_{\{a>\alpha\}}
    \1_{\{b>\beta\}}
    \,\dd\alpha\,\dd\beta\,\dd h,
    \label{eq:oriented-cap-representation}
\end{align}

Therefore, \eqref{eq:two-haar-integrals} follows.

\begin{remark}
The introduction of $\kappa=(a+b-1)_+$ reduces the two threshold integrals to the first two moments of a single scalar quantity. In particular, the dependence on \(x\) and \(y\) is transferred to the parameter \(s=\abs{\langle x| y\rangle}\), which makes the computation more direct.
\end{remark}

\endgroup